\documentclass[a4paper,oneside]{memoir}

\pdfoutput=1
\usepackage[english]{babel}
\usepackage[T1]{fontenc}
\usepackage{xcolor}
\usepackage{enumerate}
\usepackage{url}
\usepackage{todonotes}
\usepackage[unicode=true,hidelinks]{hyperref}
\usepackage{amsmath, listings, amsthm, xspace, semantic, stmaryrd, mathrsfs}
\usepackage{color}
\usepackage{tikz-cd}
\usepackage{mathpartir}
\usepackage{mathtools}
\usepackage[utf8x]{inputenc}
\usepackage{amssymb}
\usepackage{subcaption}
\usepackage[all]{xy}
\usepackage{amsbsy}
\usepackage[threshold=1]{csquotes}
\usepackage{refcount,hyperref}

\makeatletter
\def\@getref#1.#2{#2}
\newcommand{\getref}[1]{%
  \begingroup
  \edef\x{\endgroup\noexpand\@getref\getrefnumber{#1}}\x}
\newcommand{\hgetref}[1]{\hyperref[#1]{\getref{#1}}}
\makeatother

\usepackage{listings}
\usepackage{lstcoq}
\definecolor{ltblue}{rgb}{0,0.4,0.4}
\definecolor{dkblue}{rgb}{0,0.1,0.6}
\definecolor{dkgreen}{rgb}{0,0.35,0}
\definecolor{dkviolet}{rgb}{0.3,0,0.5}
\definecolor{dkred}{rgb}{0.5,0,0}

\definecolor{keywordcolor}{rgb}{0.7, 0.1, 0.1}   
\definecolor{tacticcolor}{rgb}{0.1, 0.2, 0.6}    
\definecolor{commentcolor}{rgb}{0.4, 0.4, 0.4}   
\definecolor{symbolcolor}{rgb}{0.0, 0.1, 0.6}    
\definecolor{sortcolor}{rgb}{0.1, 0.5, 0.1}      

\usepackage{lsthaskell}
\usepackage{lstcontracts}

\usepackage[english]{babel}

\newcommand\tuple[1]{(#1)}
\newcommand\synhor{\textsc{hor}}

\newcommand\ilSem[2]{{\mathcal {IL}}\left\llbracket#1\right\rrbracket_{#2}}
\newcommand\sem[1]{\left\llbracket#1\right\rrbracket}
\newcommand\tSem[2]{{\mathcal {T}}\left\llbracket#1\right\rrbracket_{#2}}
\newcommand\cSem[2]{{\mathcal C}\left\llbracket#1\right\rrbracket_{#2}}
\newcommand\gSem[1]{\left\llbracket#1\right\rrbracket}
\newcommand\tysem[1]{\left\llbracket#1\right\rrbracket}
\newcommand\compileExp[2] {{\tau_{\textrm{e}}}\left\llbracket#1\right\rrbracket_{#2}}
\newcommand\compileContr[2]{{\tau_{\textrm{c}}}\left\llbracket#1\right\rrbracket_{#2}}

\newcommand\eSem[2]{{\mathcal E}\left\llbracket#1\right\rrbracket_{#2}}
\newcommand\cRed[2]{\stackrel{#2}{\Longrightarrow_{#1}}}
\newcommand\opsem[1]{\left\llbracket#1\right\rrbracket}
\newcommand{\crule}[3][]{\frac{#2}{#3}\ \mbox{\scriptsize \textit{#1}}}
\newcommand\smartScale[2]{\underline{\mathsf{scale}}(#1,#2)}
\newcommand\smartBoth[2]{\underline{\mathsf{both}}(#1,#2)}
\newcommand\smartTransl[2]{\underline{\mathsf{translate}}(#1,#2)}
\newcommand\smartLetc[3]{\smartLetA{#1}{#2}\;\smartLetB{#3}}
\newcommand\smartLetA[2]{\underline{\mathbf{let}}\;#1 = #2}
\newcommand\smartLetB[1]{\underline{\mathbf{in}}\;#1}
\newcommand\smartTplus{\underline{\mathsf{tplus}}}
\newcommand\reals{{\mathbb R}}
\newcommand\nats{{\mathbb N}}
\newcommand\ints{{\mathbb Z}}
\newcommand\bools{{\mathbb B}}

\newcommand\delay[2]{\mathit{delay}(#1,#2)}

\newcommand\type[1]{\mathtt{#1}}
\newcommand\obs{\mathtt{obs}}
\newcommand\acc{\mathtt{acc}}
\newcommand\advEnv[2]{#1/#2}
\newcommand\TC[1]{\mathcal{TC}(#1)}

\newcommand\instdec{\mathsf{inst}}
\newcommand\inst[2]{\instdec(#1,#2)}
\newcommand\cutPayoff[1]{\mathsf{cutPayoff}(#1)}
\newcommand\zero{\mathtt{zero}}
\newcommand\scale[2]{\mathtt{scale}(#1,#2)}
\newcommand\both[2]{\mathtt{both}(#1,#2)}
\newcommand\transl[2]{\mathtt{translate}(#1,#2)}
\newcommand\promote[2]{\mathsf{translExp}(#1,#2)}
\newcommand\transfer[3]{\mathtt{transfer}{(#1,#2,\mathtt{#3})}}
\newcommand\letc[3]{\letA\;#1 = #2\;\letB\;#3}
\newcommand\letA{\mathtt{let}}
\newcommand\letB{\mathtt{in}}

\newcommand\ifwithin[4]{\mathtt{ifWithin}(#1,#2,#3,#4)}
\newcommand\pto{\rightharpoonup}

\newcommand\id[1]{\ensuremath{\mathit{#1}}}

\reservestyle{\command}{\mathtt}
\command{if[if],model[model],not[not],neg[neg],loopif[loopif],
  payoff[payoff],add[add],sub[sub],mult[mult],div[div],lt[lt],and[and],true[true],
  or[or],scale[scale],translate[translate],zero[zero],
  transfer[transfer],both[both],obs[obs],now[now],let[let],Tnum[Tnum],Tvar[Tvar],
  ifWithin[ifWithin],rNum[rNum],bNum[bNum],Tplus[Tplus],
  Texpr[Texpr],TnumZ[TnumZ],cond[cond],tplus[tplus],ltn[ltn]}

\newcommand{\kt}[1]{\texttt{#1}}
\newcommand{\note}[1]{#1}

\theoremstyle{plain}
\newtheorem{thm}{Theorem}[section]

\newtheorem{lemma}{Lemma}[section]
\newtheorem{corollary}{Corollary}[section]

\theoremstyle{definition}
\newtheorem{defn}{Definition}[section]
\newtheorem{remark}{Remark}[section]
\newtheoremstyle{case}{}{}{}{}{}{:}{ }{}
\theoremstyle{case}
\newtheorem{case}{Case}
\makeatletter
\@addtoreset{case}{theorem}
\@addtoreset{case}{lemma}
\makeatother
\newtheorem{example}{Example}[chapter]

\newcommand{\icode}[1]{\lstinline[mathescape]!#1!}
\newcommand{\icodet}[1]{\textup{\lstinline{#1}}}

\newcommand{\Ref}[1]{(\getref{#1})}
\newcommand{\reff}[1]{(\ref{#1})}

\newcommand{\Atom}{\mathbb{A}}
\newcommand{\Perm}{\id{Perm}~\Atom}
\newcommand{\SubPerm}[1]{\id{Perm}_{#1}~\Atom}
\newcommand{\swap}[2]{(#1 ~ #2)}
\newcommand{\action}[2]{#1 \cdot #2}
\newcommand{\aaction}[3]{#2 \cdot_{#1} #3}
\newcommand{\supp}[1]{\id{supp}~ #1}
\newcommand{\asupp}[2]{\id{supp}_{#1}~#2}
\newcommand{\card}[1]{\id{card}~#1}
\newcommand{\Nom}{\mathbf{Nom}}

\newcommand{\finmap}{\stackrel{\textrm{fin}}{\rightarrow}}
\newcommand{\TE}{\id{TE}\,}
\newcommand{\VE}{\id{VE}\,}
\newcommand{\ME}{\id{ME}\,}

\newcommand{\Env}{{\textrm{Env}}}
\newcommand{\TEnv}{{\textrm{TEnv}}}
\newcommand{\VEnv}{{\textrm{VEnv}}}
\newcommand{\MEnv}{{\textrm{MEnv}}}
\newcommand{\MTEnv}{{\textrm{MTEnv}}}
\newcommand{\Mid}{{\textrm{Mid}}}
\newcommand{\MTid}{{\textrm{MTid}}}
\newcommand{\Mod}{{\textrm{Mod}}}
\newcommand{\FunSig}{{\textrm{FunSig}}}
\newcommand{\TSet}{{\textrm{TSet}}}
\newcommand{\LSet}{{\textrm{LSet}}}
\newcommand{\NSet}{{\textrm{NSet}}}
\newcommand{\Finset}{{\textrm{Fin}_\Atom}}
\newcommand{\MTy}{{\textrm{MTy}}}
\newcommand{\kw}[1]{\mbox{\bfseries{#1}}}
\newcommand{\sembox}[1]{\hfill \normalfont \mbox{\fbox{\(#1\)}}}
\newcommand{\titlesembox}[2]{\subsubsection*{#1}\sembox{#2}}
\newcommand{\Dom}{\textrm{Dom}}

\newcommand{\SEP}{\hspace{2mm}|\hspace{2mm}}
\newcommand{\rar}{\rightarrow}
\newcommand{\cbra}[1]{\texttt{\{}#1\texttt{\}}}
\newcommand{\eps}{\epsilon}
\newcommand{\Rar}{\Rightarrow}

\newcommand{\RAR}{{\footnotesize\Rightarrow}}

\newcommand{\para}[1]{\texttt{(} #1 \texttt{)}}

\newcommand{\fraccn}[2]{\vspace{2mm}\refstepcounter{equation}\mbox{$\frac{\begin{array}{c} #1 \end{array}}{\begin{array}{c} #2 \end{array}}$}~(\arabic{equation})}

\newcommand{\onepart}[1]{\noindent\hfill#1\hfill\mbox{~}}
\newcommand{\twopart}[2]{\noindent\hfill#1\hfill#2\hfill\mbox{~}}

\newcommand{\SP}{\hspace{3mm}}
\newcommand{\LSP}{\hspace{6mm}}

\newcommand\vd\vdash

\newcommand{\Ec}{{\mathcal E}}
\newcommand{\Mc}{{\mathcal M}}

\newcommand{\VEc}{{\mathcal VE}}
\newcommand{\MEc}{{\mathcal ME}}

\newcommand{\Cl}{\mathtt{Cl}}

\newcommand{\MLTT}{\textsf{MLTT}}
\newcommand{\HOTT}{\textsf{HoTT}}

\newcommand{\UIP}{\textsf{UIP}}

\newcommand{\deltop}{\Delta^{\scalebox{0.6}{\textrm{op}}}_{\scalebox{0.6}{$\boldsymbol{+}$}}}

\renewcommand{\C}{\mathcal{C}}

\newcommand{\UU}{\ensuremath{\mathcal{U}}}

\newcommand{\prd}[1]{\Pi_{#1}\mkern1mu}
\newcommand{\smsimple}[1]{\Sigma_{#1}}
\newcommand{\sm}[1]{\Sigma \left(#1\right).\,}

\newcommand{\refl}[1]{\ensuremath{\mathsf{refl}_{#1}}}
\newcommand{\jdeq}{\equiv}
\newcommand{\defeq}{\vcentcolon\equiv}
\newcommand{\trans}[2]{\ensuremath{{#1}_{*}\mathopen{}\left({#2}\right)\mathclose{}}}
\newcommand{\transf}[1]{\ensuremath{{#1}_{*}}} 

\newcommand{\N}{\ensuremath{\mathbb{N}}}

\newcommand{\strict}[1]{#1^\mathrm{s}}

\newcommand{\emptyt}{\ensuremath{\mathbf{0}}}
\newcommand{\unit}{\ensuremath{\mathbf{1}}}

\newcommand{\Fin}{\mathsf{Fin}}

\newcommand{\steq}{\stackrel{\mathrm s}{=}}
\newcommand{\stiso}{\strict \simeq}

\newcommand{\obj}[1]{\vert #1 \vert}

\newcommand{\strictN}{\strict \N}
\newcommand{\strictNop}{(\strictN)^{\scalebox{0.6}{\textrm{op}}}}

\newcommand{\Prop}{\mathsf{Prop}}

\makeatletter
\newcommand{\deflabel}[1]{
\textsc{#1}%
\def\@currentlabelname{\ensuremath{#1}}%
}
\makeatother

\maxsecnumdepth{subsection}

\newlength\dropp
\makeatletter
\newcommand*\titleM{\begingroup
\setlength\dropp{0.08\textheight}
\centering
\vspace*{\dropp}
{\Huge\bfseries \sffamily Adventures in Formalisation: \\[1em] Financial Contracts, Modules, and Two-Level Type Theory\par~}\\[\baselineskip]
\vfill
{\large\bfseries Danil Annenkov}\\
  {DIKU, Department of Computer Science\\
  University of Copenhagen, Denmark}\\[3em]
{\scshape November 2017}\\
\vfill
{\large\bfseries PhD Thesis}\\
{\large\itshape This thesis has been submitted to the PhD School of the Faculty of Science, \\University of Copenhagen, Denmark}\\
\endgroup}
\makeatother

\begin{document}

\begin{titlingpage}
\calccentering{\unitlength}                         
\begin{adjustwidth*}{\unitlength}{-\unitlength}     
    \begin{adjustwidth}{-1cm}{-1cm}                 
      \titleM
    \end{adjustwidth}
\end{adjustwidth*}
\end{titlingpage}

\frontmatter

\begin{abstract}
Over the last few decades, software has become essential for the proper
functioning of systems in the modern world. Formal verification
techniques are slowly being adopted in various industrial
application areas, and there is a big demand for research in the theory and
practice of formal techniques to achieve a wider acceptance of tools for
verification.

We present three projects concerned with applications of certified
programming techniques and proof assistants in the area of programming language
theory and mathematics.

The first project is about a certified compilation technique for a domain-specific
programming language for financial contracts (the CL language). The code in
CL is translated into a simple expression language well-suited for
integration with software components implementing Monte Carlo simulation
techniques (pricing engines). The compilation procedure is accompanied with
formal proofs of correctness carried out in the Coq proof assistant. Moreover,
we develop techniques for capturing the dynamic behaviour of contracts with the
passage of time. These techniques potentially allow for efficient integration
of contract specifications with high-performance pricing engines running on
GPGPU hardware.

The second project presents a number of techniques that allow for formal
reasoning with nested and mutually inductive structures built up from finite
maps and sets (also called semantic objects), and at the same time allow for
working with binding structures over sets of variables. The techniques, which
build on the theory of nominal sets combined with the ability to work with
multiple isomorphic representations of finite maps, make it possible to give a
formal treatment, in Coq, of a higher-order module system, including the
ability to eliminate entirely, at compile time, abstraction barriers introduced
by the module system. The development is based on earlier work on static
interpretation of modules and provides the foundation for a higher-order module
language for Futhark, an optimising compiler targeting data-parallel
architectures, such as GPGPUs.

The third project is related to homotopy type theory (HoTT), a new branch of
mathematics based on a fascinating idea connecting type theory and homotopy
theory. HoTT provides us with a new foundation for mathematics allowing for
developing machine-checkable proofs in various areas of computer science and
mathematics. Along with Vladimir Voevodsky's univalence axiom, HoTT offers a
formal treatment of the informal mathematical principle: equivalent structures
can be identified. However, in some cases, the notion of weak equality available
in HoTT leads to the ``infinite coherence'' problem when defining internally
certain structures (such as a type of $n$-restricted semi-simplicial types,
inverse diagrams and so on). We explain the basic idea of \emph{two-level type
  theory}, a version of Martin-L\"of type theory with two equality types: the
first acts as the usual equality of homotopy type theory, while the second
allows us to reason about strict equality. In this system, we can formalise
results of partially meta-theoretic nature. We develop and explore in details
how two-level type theory can be implemented in a proof assistant, providing a
prototype implementation in the proof assistant \texttt{Lean}. We demonstrate an
application of two-level type theory by developing some results from the theory
of inverse diagrams using our Lean implementation.
\end{abstract}

\clearpage

\renewcommand{\abstractname}{Resumé}
\begin{abstract}
Denne afhandling består af tre dele og omhandler teknikker til
udvikling af certificeret software samt anvendelse af bevisassistenter
indenfor områder som programmeringssprogsteori og matematik.

Den første del omhandler en certificeret oversættelsesteknik for et
domainespecifikt programmeringssprog til finansielle kontrakter
(sproget CL). Kode i CL oversættes til et simpelt udtrykssprog, som er
velegnet til integration med softwarekomponenter, der implementerer
Monte-Carlo simuleringsteknikker
(prisberegningssoftware). Oversættelsesproceduren er akkompagneret af
et formelt korrekthedsbevis, der er etableret ved brug af
bevisassistenten Coq.

Den anden del omhandler en række teknikker, der tillader formel
ræsonnement med nestede og gensidigt induktive strukturer bygget op af
endelige afbildninger og mængder (også kaldet semantiske
objekter). Teknikkerne, som bygger på teorien om nominelle mængder
kombineret med muligheden for at arbejde med multible isomorfe
repræsentationer af endelige afbildninger, gør det muligt at give en
formel behandling, i Coq, af et højere-ordens
modulsystem. Behandlingen understøtter muligheden for at eliminere
alle modulkonstruktioner og abstraktionsbarrierer på
oversættelsestidspunktet. Teknikken baserer sig på tidligere arbejde
indenfor statisk fortolkning af moduler og giver et fundament for et
højere-ordens modulsprog for Futhark, en optimerende oversætter
målrettet data-parallelle arkitekturer som GPGPUer.

Den tredie del omhandler en implementation af to-niveau typeteori, en
version af Martin-L\"ofs typeteori indeholdende to lighedstyper. Den
første fungerer som det sædvanlige lighedsbegreb fra homotopy
typeteori, mens den anden tillader ræssonnementer omkring stringent
lighed. I dette system er det muligt at formalisere resultater af
delvist meta-teoretisk natur. Det undersøges i detaljer hvordan
to-niveau typeteori kan implementeres i en bevisassistent og der
udvikles en prototypeimplementation i bevisassistenten Lean. Ydermere
demonstreres anvendelsen af to-niveau typeteori ved udvikling af nogle
resultater (i den udviklede Lean-implementation) indenfor teorien om
inverse diagrammer.
\end{abstract}

\renewcommand{\abstractname}{Abstract}

\clearpage

\tableofcontents*

\chapter*{Preface}
\addcontentsline{toc}{chapter}{Preface}{}
This dissertation has been submitted to the PhD School of Science, Faculty of
Science, University of Copenhagen, in partial fulfilment of the degree of PhD
at Department of Computer Science (DIKU).

The main content of the dissertation consists of three chapters, an
introduction, and a conclusion. Some results of Chapter \ref{chpt:contracts}
where presented at the Nordic Workshop on Programming Theory 2016 (NWPT'16) by
the author. \note{The full source code of the formalisation presented in Chapter \ref{chpt:contracts}
is available online: \url{https://github.com/annenkov/contracts}.}

Chapter \ref{chpt:modules} presents the author's contribution to
ongoing work on a module system development and formalisation in collaboration
with Martin Elsman, Cosmin Oancea, and Troels Henriksen at the HIPERFIT Research
Center, DIKU, University of Copenhagen. \note{The source code of the implementation of nominal sets in Coq
(using type classes instead of modules) and the proof of normalisation from
Section \ref{sec:modules:stlc-norm} are available online: \url{https://github.com/annenkov/stlcnorm}.}

Chapter \ref{chpt:tltt} presents the author's contribution to the work submitted for a publication with Paolo
Capriotti and Nicolai Kraus, University Of Nottingham (publication preprint
\cite{ann-cap-kra:two-level}). The results of this work were presented by the author
at the workshop on Homotopy Type Theory/Univalent Foundations(co-located with FSCD 2017).
{\note{The full source code of the formalisation presented in Chapter \ref{chpt:tltt} is available online:
\url{https://github.com/annenkov/two-level}.}

\clearpage

\chapter*{Acknowledgments}
\addcontentsline{toc}{chapter}{Acknowledgments}{}

I would like to express my sincere gratitude to the people who made my PhD studies
an excellent and exciting experience.

To my supervisor, Martin Elsman, for introducing me to the programming language
research, for guiding and supporting me in pursuing my ideas, and for making me
a better researcher.

To Fritz Henglein, for his hospitality and for the HIPERFIT Research Center allowing
me to work on exiting problems in the area of finance.

To Patrick Bahr, for his beautiful work on contract formalisation in Coq, which
provided me with the inspiration to work in the area of certified programming.

To Omri Ross, for sharing his experience in the area of finance and exploring
the applications of the contract language to real-world problems.

To Andrzej Filinski, for his Semantics and Types course.

To my colleagues at the DIKU APL section, for the excellent research atmosphere
and for interesting discussions over lunch.

To Thorsen Altenkirch, for inviting me to the FP Lab at University of
Nottingham. To Paolo Capriotti and Nicolai Kraus, for all your patience
explaining me category theory and inspiring me to study mathematics and pursue
research in homotopy type theory. And to all the members of the FP Lab for the
pleasant atmosphere and pubs every Friday.

I am so grateful to my wife Anna and to my daughter Arina, for supporting my
crazy idea of pursuing PhD studies in Denmark and for being with me during this
wonderful adventure. I am so happy that I have you, girls!

\clearpage

\mainmatter
\chapter{Introduction}
Over the last few decades, software has become essential for the proper
functioning of systems in the modern world. Some of these systems are
safety-critical, such as embedded systems in avionics, or nuclear power
plants. But not only in these areas do software correctness play such a
prominent role. For example, in the financial sector, software systems are
responsible for executing financial transactions and managing assets by means of
\emph{smart contracts} on distributed ledgers \cite{ethereumyellow2015}. Formal
verification techniques are slowly being adopted in various industrial
application areas, including the area of finance and financial algorithms
\cite{Passmore2017}. There is a big demand for research in the theoretical
foundations for and practical aspects of formal techniques aimed at achieving a
wider acceptance of tools for verification.

In general, our every day life relies more and more on complex software
systems. Moreover, the development of complex software systems is a very costly
process and discovering errors at the deployment stage may cause a significant
increase in the overall cost of a system. For the last decade, software
verification techniques have become available for the wider use, due to advances
both in theories and in tools for formal verification. There are various
approaches to formal software verification. Here we will focus on a particular
direction based on various flavors of type theory and tools implementing
them. These tools are called interactive theorem provers, or proof
assistants. A number of large-scale verification projects use proof assistants, and
we will mention some of them.

The CompCert project \cite{2006-Leroy-compcert} is one of the large-scale
verification efforts for real-world software. It is a verified compiler for a
significant subset of the C programming language carried out in the Coq proof
assistant. The C programming language is widely used for development of
numerous applications including critical systems. CompCert is used as a part of a
verified toolchain in a number of projects.

Another example of this kind is the JSCert project. JSCert is a specification
of ECMAScript 5 (JavaScript) in Coq. JavaScript is widely used in web
development to write complicated applications running in browsers. Many modern
web applications have large code bases in JavaScript. It is also a well known
fact that JavaScript is a language with many pitfalls, which is why ``The
JSCert project aims to \emph{really} understand JavaScript''. Moreover, the
project features an interpreter in OCaml obtained using Coq's code extraction
facilities.

Apart from the areas related to software verification and programming language
semantics mechanisation, in mathematics, one often wants proofs to be verified
by some automatic procedure in order to ensure correctness. Mathematical proofs
can be very complicated and may require the consideration of a large number of
cases. Examples of large-scale developments in this area include proofs of the
four-color theorem \cite{Gonthier2008} and the Feit–Thompson theorem
\cite{Gonthier2013}.

Moreover, in such abstract areas of mathematics as homotopy theory, for a long
time, it has been almost impossible to use proof assistants to carry out
proofs. But with recent development of homotopy type theory \cite{hott-book}
(HoTT), it has become possible to carry out proofs in homotopy theory
\cite{Licata:2013:CFG:2591370.2591407} and many other areas of mathematics in
the language of type theory. This, in turn, allows for developing formalisations
in proof assistants. HoTT offers a new foundation of mathematics, where types
(or spaces, in the homotopical interpretation) become the basic objects for
developing mathematics. The \emph{Unimath} project \cite{UniMath} takes this
approach and aims at implementing a large body of mathematics in the Coq proof
assistant. Moreover, from the dependently-typed programming perspective, HoTT
offers a generic programming technique, allowing to change between different
isomorphic representations of the same abstract data structure.

\section{Type Theory and the Curry-Howard Correspondence}
Type theory originates from Bertrand Russel's approach to avoid paradoxes in
set theory. Since that time, type theory has been developed by many scientists
including Alonzo Church, Haskell Curry, William Howard, Stephen Kleene, Kurt
G\"odel, Nicolaas de Bruijn, Per Martin-L\"of, and others.  In the form of
\emph{dependent} type theory, it became a foundation for various proof assistants.

The central notion in type theory is a \emph{typing judgment}. That is, a term $a$ has a type
$A$:
\[ a : A \]
Notice the similarity with the set-theoretic proposition $a \in A$. The important
difference is that in type theory, each term comes with the type and internally
in type theory we cannot ask if some term has type $A$ or $B$. Essentially this is
how programmers in statically typed programming languages think about programs
and data types. Functional programming languages have especially strong
connection to type theory, since theoretical foundations for such languages are
variations of the lambda-calculi.

An important step in the development of type theory was the discovery of the
connection between logic and type theory, which is now known as the Curry-Howard
correspondence \cite{Curry:functoriality, howard:formulae-as-types} (for the
details of the discovery and the development of the Curry-Howard correspondence,
see \cite{Wadler:2015:PT:2847579.2699407}). This correspondence is also known in
the literature as propositions-as-types, formulae-as-types, and the Curry-Howard
isomorphism.

Roughly speaking, the idea of this correspondence is that type theory
corresponds to \emph{intuitionistic} logic, propositions correspond to types,
and proofs correspond to terms (or programs). The summary of the correspondence is
given in Table \ref{table:intro:props-as-types}.

\begin{table}
\def\arraystretch{1.5}
\begin{tabular}{ c | c }
  \hline
  \multicolumn{1}{c}{\bfseries Logic} & \multicolumn{1}{c}{\bfseries Type Theory }\\
  \hline
  Proposition $A$ & Type $A$\\
  Proof of a proposition $A$ & Term (program) $a : A$\\
  Proof normalisation & Program execution\\
  True & Unit  \\
  False & Empty type \\
  Conjunction $A \land B$ & Product type $A \times B$ \\
  Disjunction $A \lor B$ & Coproduct type $A + B$ \\
  Implication $A => B$ & Function type $A -> B$\\
  Universal quantification $\forall x \in A, B(x) $ & Dependent product $\Pi (x : A). B(x)$\\
  Existential quantification $\exists x \in A, B(x) $ & Dependent sum $\Sigma (x : A). B(x)$\\
  \hline
\end{tabular}
\caption{Propositions-as-types.}\label{table:intro:props-as-types}
\end{table}

This correspondence was extended further: Joachim Lambek showed the correspondence
between the lambda-calculus and Cartesian closed categories \cite{lambek:lambda-ccc}.

The work of Per Martin-L\"of \cite{martin-lof:bibliopolis} was another important
step in the development of type theory. A number of modern proof assistants
implement some variation of Martin-L\"of type theory. Recent
research in type theory have lead to the discovery of another deep
connection: the connection between type theory and homotopy theory. Homotopy type theory
\cite{hott-book} establishes a correspondence between types and spaces (more
precisely, $\infty$-groupoids), terms and points, identity types and paths in
a space. Moreover, homotopy type theory refines the correspondence outlined in
Table \ref{table:intro:props-as-types}: propositions correspond not to any
types, but to certain types that have at most one inhabitant. These types in the
context of homotopy type theory are called \emph{hProps}.

According to the Curry-Howard correspondence, finding a proof of some theorem
is the same as finding a term of the given type, i.e. \emph{inhabiting}
the type. Following this idea, the process of proving a theorem corresponds to
the process of writing a program that is accepted be the type-checker. This
proving-as-programming idea have lead to software tools supporting this
paradigm, namely, proof assistants.

\section{Proof Assistants and Certified Programming}
Proof assistants, or interactive theorem provers are tools that allow for
stating and proving theorems by interacting with users. That is, users write proofs
in a specialised language and the tool verifies correctness of these proofs. Proof assistants
often offer some degree of proof automation by implementing decision and
semi-decision procedures, or interacting with automated theorem provers (SAT
and SMT solvers). Some proof assistants allow for writing user-defined
automation scripts, or write extensions using a plug-in system. In the present
work we will focus on application of proof assistant based on dependent type
theory (Coq \cite{bertot:coq}, Agda \cite{norell:towards}, Lean
\cite{moura:lean}), although there are other tools based on different
foundations, such as variations of set theory and higher-order logic (Mizar
\cite{Mizar:JFR1980}, Isabelle/HOL \cite{Nipkow:2002:IPA:1791547}), and
meta-logical frameworks (Twelf \cite{Pfenning:1999:SDT:648235.753634}).

That is, in this thesis we will explore and discuss how dependent type
theory (and its implementation in a proof assistant) can be used in particular
cases.  We will explore how to use the expressivity of dependent types to
encode invariants of structures used in formalisations in such a way that it
simplifies development of the formalisation.

With connection to proof assistant technology, by \emph{certified programming}
following \cite{cpdt} we mean a process that produces a program along with
a witness of its correctness with respect to the specification.  The benefit of
using proof assistants based on type theory is that programs and proofs live in
the same realm. That is, one can write (functional) programs and reason
about their properties using the same language. Moreover, some proof assistants
allow for the extraction of computational parts of an implementation into some
(usually functional) programming language.

Next, we briefly describe some proof assistants implementing dependent type theory,
outlining their main features.

\lstset{language=Coq}
\subsection{Coq}
The theoretical foundation of the Coq proof assistant is the calculus of
constructions \cite{COQUAND198895} extended with inductive definitions leading
to the calculus of inductive constructions \cite{Coquand1990,bertot:coq}. The type
theory of Coq distinguishes between two kinds of types: \icode{Prop} and
\icode{Set}. The type of propositions \icode{Prop} is used to encode properties
that will be erased during program extraction, while \icode{Set} is used for
programs containing computational content. Moreover, \icode{Prop} is
impredicative, which means that statements quantifying over \icode{Prop} still
belong to \icode{Prop}. Impredicativity of \icode{Prop} makes working with
logical connectives more convenient, but to maintain consistency, the elimination
principle for propositions only allows the result of elimination to be in
\icode{Prop}.

Coq also features a hierarchy of type universes \icode{Type}$_i$ for
$i\in \nats, i \ge 1$. Types \icode{Prop} and \icode{Set} belong to
\icode{Type}$_1$, and \icode{Type}$_i$:\icode{Type}$_{i+1}$. Most of the time
users do not have to be explicit about universe levels; universe constraints
are handled by the system automatically. Recent versions of Coq support
universe polymorphism.

Coq features the following languages:
\begin{itemize}
\item the \emph{Gallina} language, which is essentially a dependently typed
  programming language (also includes the language of commands, called The
  Vernacular);
\item the \emph{Ltac} language, allowing for writing \emph{tactics} for proof
  automation.
\end{itemize}

The tactic language is often used to build complicated proof terms and to implement
certain proof search strategies. While Gallina programs are always terminating,
since consistency of the underlying logic depends on this property, tactics written
in the Ltac language may fail to terminate without affecting the consistency.
The standard library of Coq offers various useful primitive tactics, along with
proof searching procedures \icode{auto} and \icode{eauto}, which use a user-defined
database of lemmas (``hints'') when trying to solve a goal. The library also
contains several decision procedures, such as \icode{omega} for the Presburger
arithmetic and \icode{tauto} for the intuitionistic propositional calculus.

Coq supports type classes, which are useful for operation overloading and for
proof automation through the resolution mechanism.

It is possible to obtain an implementation in OCaml, Haskell or Scheme from the
Coq formalisation through the code extraction mechanism, provided that the
development follows certain criteria.

\lstset{language=Lean}
\subsection{Lean}
We give a brief outline of features of the Lean proof assistant version 2
\cite{moura:lean}, since this version has been used in this
thesis.\footnote{The Lean 2 repository can be found at
  \url{https://github.com/leanprover/lean2}}  Lean 2 has two different modes:
\begin{itemize}
\item The ``strict'' mode, based on a similar theoretical foundation as Coq:
  the calculus of inductive constructions with impredicative \icode{Prop} and
  definitional proof irrelevance;
\item the ``HoTT'' mode, supporting homotopy type theory (without impredicative
  or proof irrelevant universes), including some higher-inductive types.
\end{itemize}

Lean features the powerful elaboration mechanism allowing to infer universe levels,
implicit arguments (including type class instances), supports notation
overloading, and so on. Type classes allow for some proofs to be automated by the
elaboration mechanism. One of the main motivations behind Lean is to bridge
the gap between automated and interactive theorem provers by pushing automatic
inference as far as possible.

Proofs in Lean can be written in \emph{term} mode, which is basically syntactic
sugar for proof terms in Lean's functional language to make proofs more readable.
Alternatively, one can use tactics, similarly to Coq. Although, unlike in Coq, one
can switch to the ``tactic'' mode in any place of the definition by using
\icode{begin ... end} and filling-in the proof using tactics.

\subsection{Agda}
Agda is a dependently typed programming language implementing a predicative
extension of Martin-L\"of type theory \cite{martin-lof:bibliopolis}. It does
not have a Prop-Set distinction as Coq.

Agda has a hierarchy of universes and supports universe polymorphism, but one has to
be explicit about universe levels.

In comparison with Coq, Agda has more experimental features (like
inductive-recursive and inductive-inductive definitions), and possibility to
turn off some checks (like termination or strict positivity of inductive
definitions), making a system possibly inconsistent, but more suitable for
experimentation.

Working in Adga resembles a programming activity in a functional language. Agda
supports powerful dependent pattern-matching constructs allowing one to write
proofs as programs directly. Similarly to type classes in Coq and Lean, Agda supports
instance arguments that can be inferred using the instance resolution mechanism.
There is no build-in support for tactics in Agda, although there are some developments,
supporting mechanisms similar to Coq's \icode{auto} \cite{Kokke2015}.

\section{Thesis}
The thesis considers three applications of proof assistant technology.

The semantics of domain-specific languages is an important and, at the same
time, a realistic target for the application of formalisation and verification
techniques. Financial contract specifications are often used in a specialised
form in software components performing simulations to estimate the possible
price of a contract. These software components are called pricing engines, which
are optimised for performing simulations efficiently. The first question we consider
is the following:
\begin{itemize}
\item Is it possible to develop a certified implementation of a compilation
  technique for the domain-specific financial contract specification language
  in the style of \cite{BahrBertholdElsman}, allowing for efficient interaction
  with the pricing engine?
\end{itemize}

Module systems provide a powerful abstraction mechanism allowing for writing
generic highly parameterised code. For some application domains it is important
to have static guarantees that module abstractions introduce no
overhead. Formalisation of module systems is hard, and it has turned out to be
essential to develop a number of techniques allowing for development of a module
system formalisation in the Coq proof assistant. Thus, the second question is
the following:
\begin{itemize}
\item Is it possible to develop a formalisation in the Coq proof assistant of a
  higher-order module system for the data-parallel array language Futhark
  \cite{Henriksen:Futhark} in the style of \cite{elsman99}, aiming to keep it
  as close as possible to a pen-and-paper formalisation and to the implementation
  in the Futhark compiler?
\end{itemize}

The third project is related to the internalisation of partially
meta-theoretical results in two-level type theory (homotopy type theory
extended with strict equality). Homotopy type theory is young and a developing
field with a number of open problems. For instance, the notion of weak equality
available in HoTT could lead to the ``infinite coherence'' problem when
defining internally certain structures (such as a type of $n$-restricted
semi-simplicial types, inverse diagrams, and so on). Two-level type theory allows for
approaching this problem. The third question we consider is the following:
\begin{itemize}
\item How can one leverage existing proof assistants to implement two-level type theory,
  and is it possible to use such an implementation for the development of formalisations
  of partially meta-theoretical results?
\end{itemize}

\section{Contributions}
We present three projects concerned with the applications of certified
programming techniques and proof assistants in the area of programming language
theory and in the area of mathematics.

The contributions on certified compilation of financial contracts are as follows:
\begin{itemize}
\item We present an extension of a domain-specific language for financial
  contracts developed by authors of \cite{BahrBertholdElsman}.  The extension
  features contract templates, or \emph{instruments}. We focus on the
  extension allowing for parameterisation of contracts with respect to temporal
  parameters.
\item We develop a payoff intermediate language inspired by traditional payoff
  languages and well-suited for the integration with Monte Carlo simulation techniques.
\item We use the Coq proof assistant to develop a certified compilation
  procedure of contract templates into a parameterised payoff intermediate language.
\item We further parameterise the compiled payoff expressions with the notion of ``current time''
  allowing for capturing the evolution of contracts with the passage of time.
\item We develop the proof of an extended soundness theorem in the Coq proof
  assistant. The theorem establishes a correspondence between the
  time-parameterised compilation scheme and the contract reduction semantics.
\item We argue how the parametric payoff code allows for better performance due to
  avoiding recompilation with the change of parameters.
\end{itemize}

The contributions on the formalisation of a higher-order module system for the
Futhark language are as follows:
\begin{itemize}
\item We develop a formalisation of the static interpretation of a module system
  in the style of \cite{elsman99} in the Coq proof assistant. This is one of the first
  developments in this style in Coq.
\item For implementing the core concept of semantic objects we develop a
  technique that allows for using isomorphic representations of components of
  semantic objects with low proof obligation overhead. We use this technique to
  overcome limitations of the conservative strict positivity check in Coq.
\item To deal with binding in the context of semantic objects, we apply nominal
  techniques. We develop a small library defining nominal sets in a generalised
  setting allowing for sets of variables to be bound at once. We use the developed library
  to define $\alpha$-equivalence of semantic objects.
\end{itemize}

The contributions on formalisation of two-level type theory are as follows:
\begin{itemize}
\item We develop a technique allowing for two-level type theory to be
  implemented in existing proof assistants.
\item We implement two-level type theory in the Lean proof assistant using the
  developed technique.
\item We demonstrate how the fibrant fragment of two-level type theory can be used
  to develop proofs in a similar way as in homotopy type theory.
\item As an application of the implemented type theory, we
  internalise some results on the theory of inverse diagrams in our Lean
  development.
\end{itemize}

\section{Structure of the Dissertation}
Following the outline of the thesis contributions, the main content of the thesis is split
into three chapters.
\begin{itemize}
\item Chapter \ref{chpt:contracts} describes our work on certified compilation of
financial contracts, including a formal semantics, compilation soundness
theorems, and a description of our Coq formalisation.
\item Chapter \ref{chpt:modules} describes our formalisation of a higher-order module system in
Coq, focusing on details of the implementation and the particular techniques applied.
\item Chapter \ref{chpt:tltt} discusses the motivation for two-level type theory and
  describes our approach to implementation, exemplified by a development in
  the Lean proof assistant.
\end{itemize}

\chapter{Certified Compilation of Financial Contracts}\label{chpt:contracts}

\lstset{language=Contracts}
\section{Background and Motivation}

New technologies are emerging that have potential for seriously disrupting the
financial sector. In particular, blockchain technologies, such as Bitcoins
\cite{nakamoto2008bitcoin} and the Ethereum Smart Contract peer-to-peer
platform \cite{ethereumyellow2015}, have entered the realm of the global
financial market and it becomes essential to ask to which degree users can
trust that the underlying implementations are really behaving according to the
specified properties. Unfortunately, the answers are not clear and errors may
result in irreversible high-impact events.

Contract description languages and payoff languages are used in large scale
financial applications \cite{MLFi, SimCorpXpress}, although formalisation of
such languages in proof assistants and certified compilation schemes are less
explored.

The work presented here builds on previous work on specifying financial
contracts
\cite{andersen06sttt,Arnold95analgebraic,Frankau09JFP,hvitved11jlap,SPJ2000}
and in particular on a certified financial contract management engine and its
associated domain-specific contract specification language
\cite{BahrBertholdElsman}.  This framework allows for expressing a wide variety
of financial contracts (a fundamental notion in financial software) and for
reasoning about their functional properties (e.g., horizon and causality).

As in the previous work, the contract language that we consider is equipped
with a denotational semantics, which is independent of stochastic aspects and
depends only on an \emph{external environment} $\type{Env} :\mathbb{N} \times
\type{Label} -> \mathbb{R} \cup \mathbb{B}$, which maps observables (e.g., the
price of a stock on a particular day) to values. We will refer to the contract
language as described in \cite{BahrBertholdElsman} and its extension developed in
this chapter as the CL language.
%
%
As the first contribution of this work, we present a certified compilation
scheme that compiles a contract into a \emph{payoff function}, which aggregates
all cashflows in the contract, after discounting them according to some
model. The result represents a single ``snapshot'' value of the contract. The
payoff language is inspired by traditional payoff languages, and it is well
suited for integration with Monte Carlo simulation techniques for pricing.
It is essentially a small subset of a C-like expression language enriched with
notation for looking up observables in the external environment. We show that
compilation from CL to the payoff language preserves the cashflow
semantics.

The contract language described in \cite{BahrBertholdElsman}, deals
with concrete contracts, such as a one year European call
option on the AAPL (Apple) stock with strike price \$100. The lack of
genericity means that each time a new contract is created (even a
very similar one), the contract management engine needs to compile the
contract into the payoff language and further into a target language
for embedding into the pricing engine. As our second contribution,
we introduce the notion of a \emph{financial instrument},
which allows for templating of contracts and which can be turned into
a concrete contract by instantiating template variables with
particular values. For example, a European call option instrument has
template parameters such as maturity (the end date of the contract),
strike, and the underlying asset that the option is based
on. Compiling such a template once allows the engine to reuse compiled
code, giving various parameter values as input to the pricing engine.

Moreover, an inherent property of contracts is that they evolve over time. This
property is precisely captured by a contract reduction semantics. Each day, a
contract becomes a new ``smaller'' contract, thus, for pricing purposes,
contracts need to be recompiled at each time step, resulting in a dramatic
compile time overhead. As our the third contribution we introduce a mechanism
allowing for avoiding recompilation in relation with the contract evolution. A
payoff expression can be parameterised over the \emph{current time} so that
evaluating the payoff code at time $t$ gives us the same result (upto
discounting) as first advancing the contract to time $t$, then compiling it to
the payoff code, and then evaluating the result. Most of the payoff languages
used in real-world applications require synchronization of the contract and the
payoff code once a contract evolves \cite[Contract State and Pricing
  Synchronization]{MLFi-whitepaper}. But in some cases, as we mentioned
earlier, it is important to capture the reduction semantics in the payoff
language as well. Our result allows for using a single compilation procedure
for both use cases: compiling a contract upfront and synchronizing at each time
step.

The contract analysis and transformation code forms a core code base, which
financial software crucially depends on. A certified programming approach using
the Coq proof assistant allows us to prove various correctness results and
to extract certified executable code.

The rest of the chapter is structured as follows. We describe an extension of
the original contract language \cite{BahrBertholdElsman} with template expressions
for temporal parameters in Section
\ref{subsec:contracts:contracts-syn-sem}. Next, we describe the syntax and the
semantics of the payoff intermediate language designed to capture aspects
relevant for the contract pricing purposes in Section
\ref{sec:contracts:payoffs}.  Section \ref{sec:contracts:compile} describes our
compilation approach, including a novel technique to transfer the contract
evolution behavior to payoff expressions.  We also state and sketch proofs of
soundness theorems with respect to the denotational and the reduction semantics
of contracts.  Formalisation of our contract compilation approach, including
code extraction, along with the example of Haskell code generation, are described in
Section \ref{sec:contracts:coq}.

\section{The Contract Language}
\subsection{Syntax and Semantics}\label{subsec:contracts:contracts-syn-sem}
We assume a countably infinite set of program variables, ranged over by
$\id{v}$. Moreover, we use $\id{n}$, $\id{i}$, $\id{r}$, and $\id{b}$ to range
over natural numbers, integers, reals, and booleans. We use $p$ to range over
parties.
The contract language (CL) that we consider follows the style of \cite{BahrBertholdElsman}
and is extended with template variables (see Figure \ref{fig:contracts:syntax}).
\begin{figure}
\[
\begin{split}
  c ~::=~ & \<zero>~|~\<transfer>(p_1,p_2,a)~|~\<scale>(e,c)~|~ \\
  & \<translate>(t,c) ~|~\<ifWithin>(e,t,c_1,c_2)~|~\<both>(c_1,c_2)\\
  e ~::=~& op(e_1,~e_2, ~\dots,~e_n) ~|~ \<obs>(l,i)~|~\acc(\lambda v .\ e_1, n, e_2)~|~ r ~|~ b \\
  t ~::=~& n ~|~ v \\
  op ~::=~& \<add> ~|~ \<sub> ~|~ \<mult> ~|~ \<lt> ~|~ \<neg> ~|~ \<cond> ~|~\dots
\end{split}
\]
\caption{Syntax of CL.}\label{fig:contracts:syntax}
\end{figure}

Expressions ($e$) may contain \emph{observables}, which are
interpreted in an external environment. \note{The $\acc$ construct allows for accumulating a value over a given number of days $n$.}

A contract ($c$) may be empty($\<zero>$), a transfer of one unit of some asset $a$ ($\<transfer>$),
a scaled contract ($\<scale>$), a translation of a contract into the
future ($\<translate>$), the composition of two contracts ($\<both>$),
or a generalised conditional $\<ifWithin>(cond,t,c_1,c_2)$, which
checks the condition $\id{cond}$ repeatedly during the period given by
$t$ and evaluates to $c_1$ if $\id{cond} = \<true>$ or to $c_2$ if
$\id{cond}$ never evaluates to $\<true>$ during the period $t$.

The main difference between the original version of the contract
language and the version presented here is the introduction of
\emph{template expressions} ($\id{t}$), which, for instance, allows us
to write contract templates with the contract maturity as a parameter.
This feature requires refined reasoning about the temporal properties
of contracts, such as causality. Certain constructs in the original
contract language, such as $\<translate>(n,c)$ and
$\<ifWithin>(\id{cond},n,c_1,c_2)$, are designed such that basic
properties of the contract language, including the property of
causality, are straightforward to reason about. In particular, the
displacement numbers $n$ in the above constructs are constant positive
numbers. For templating, we refine the constructs to support template
expressions in place of positive constants.
One of the consequences of adding template variables is that the semantics of
contracts now depends also on mappings of template variables in a
\emph{template environment} $\type{TEnv}: \type{Var} -> \mathbb{N}$, which is
also the case for many temporal properties of contracts. For example, the type
system for ensuring causality of contracts \cite{BahrBertholdElsman} and the
concept of symbolic contract horizon $\synhor$ are now parameterised by
template environments.  The modified version of $\synhor$ is given in
Figure \ref{fig:synhor}, where $\tSem{t}\delta$ represents the semantics of
template expressions (see Figure \ref{fig:Esem})

\begin{figure}[t]
  \begin{align*}
    \begin{aligned}[t]
      \synhor_\delta(\zero) \\
      \synhor_\delta(\transfer{p}{q}{a})
    \end{aligned}
    &
    \begin{aligned}[t]
      = 0\\
      = 1
    \end{aligned}
    \qquad
    \begin{aligned}[t]
      \synhor_\delta(\scale e c) &= \synhor_\delta(c)\\
    \synhor_\delta(\transl t c) &= \tSem{t}\delta \oplus \synhor_\delta(c)
    \end{aligned}
    \\
    \synhor_\delta(\letc x e c) &= \synhor_\delta(c)\\
    \synhor_\delta(\both {c_1}{c_2}) &= \max(\synhor_\delta(c_1),
                            \synhor_\delta(c_2))\\
    \synhor_\delta(\ifwithin e t {c_1} {c_2}) &= \tSem{t}{\delta} \oplus \max(\synhor_\delta(c_1), \synhor_\delta(c_2))
  \end{align*}
  where \vspace{-1em}
  \[
    a \oplus b =
    \begin{cases}
      0 &\text{if } b = 0\\
      a + b &\text{otherwise}
    \end{cases}
  \]
  \caption{Symbolic horizon.}
  \label{fig:synhor}
\end{figure}

On the other hand, some properties such as \emph{simple} or
\emph{obvious} causality can be verified without information from
a template environment. Although, this property might be too restrictive
for some contracts which are causal, but not obviously
causal (see \cite[Section 3.2]{BahrBertholdElsman}).

Let us consider examples of the contracts written in English and expressed
in CL.
\begin{example}\label{ex:contracts:option}
The definition of an European option contract:
\blockquote[investopedia.com]{European options are contracts that give the
  owner the right, but not the obligation, to buy or sell the underlying
  security at a specific price, known as the strike price, on the option's
  expiration date}.
Let us take: the expiration date to be 90 days into the future
and set the strike at 100 USD.  We can implement the European option contract with
these parameters in CL as follows:
\begin{lstlisting}
      translate(90,
        if(obs(AAPL,0) > 100.0,
          scale(obs(AAPL,0) - 100.0, transfer(you, me, USD)),
          zero))
\end{lstlisting}
\end{example}

\begin{example}
  Three month FX swap for which the payment schedule has been settled:
\begin{lstlisting}
   scale(1.000.000,
         both(all[translate(22, transfer(me, you, EUR)),
                  translate(52, transfer(me, you, EUR)),
                  translate(83, transfer(me, you, EUR))],
              scale(7.21,
                    all[translate(22, transfer(you, me, DKK)),
                        translate(52, transfer(you, me, DKK)),
                        translate(83, transfer(you, me, DKK))])))
\end{lstlisting}

\noindent
In the example, we have written
$\texttt{all[}c_1,\cdots,c_n\texttt{]}$ as an abbreviation for the
contract $\texttt{both(}c_1,\texttt{both(}\cdots,c_n\texttt{)}\texttt{)}$.
We use the $\texttt{all}$ shortcut with the \icode{translate} combinator to
implement a schedule of payments.
\end{example}

Using CL, considered in the present work, we can
abstract some parameters of the contact in Example \ref{ex:contracts:option} to template variables (T for expiration
date, and S for strike)\footnote{In our implementation we focus on contract
  templates allowing for template expressions to represent temporal parameters,
  like maturity. Other parameters, e.g. strike, could be expressed as
  constant observable values.}:
\begin{lstlisting}
      translate(T,
        if(obs(AAPL,0) > S,
          scale(obs(AAPL,0) - S, transfer(you, me, USD)),
          zero))
\end{lstlisting}

Such a parameterisation plays well with a way how users could
interact with a contract management system. Contract templates could be exposed to
users as \emph{instruments} that can be instantiated with concrete values from
users' input.

\begin{figure}[t]
  \fbox{$\Gamma\vdash e : \tau$}
  \begin{gather*}
      \crule{x : \tau \in \Gamma}{\Gamma\vdash x : \tau}
      \quad
      \crule{}{\Gamma\vdash r : \type{Real}}
      \quad
      \crule{}{\Gamma\vdash b : \type{Bool}}
      \quad
      \crule{l \in \type{Label}_\tau}{\Gamma \vdash
        \obs(l,t) : \tau}
      \\
      \crule{
        \begin{aligned}
          \Gamma &\vdash e_i : \tau_i \qquad
          &\vdash op : \tau_1 \times \cdots \times \tau_n \to \tau
        \end{aligned}
}{\Gamma \vdash
        op(e_1,\ldots,e_n) : \tau}
      \\
      \crule{
        \begin{aligned}
          \Gamma,x:\tau &\vdash e_1 : \tau \qquad \Gamma &\vdash e_2 : \tau
        \end{aligned}}{\Gamma \vdash
        \acc(\lambda x .\ e_1, d, e_2) : \tau}
    \end{gather*}
  \\
  \fbox{$\Gamma \vdash c : \type{Contr}$}
  \begin{gather*}
    \crule{}{\Gamma \vdash \zero : \type{Contr}}
    \quad
    \crule{}{\Gamma
      \vdash \transfer{p}{q}{a} : \type{Contr}}
    \\[1em]
    \crule{\Gamma \vdash c : \type{Contr}}{\Gamma \vdash
      \transl d c : \type{Contr}}
    \quad
    \crule{\Gamma \vdash c_i : \type{Contr}}{\Gamma \vdash
      \both{c_1} {c_2} : \type{Contr}}
      \\[1em]
      \crule{\Gamma \vdash e : \type{Real}~~~\Gamma \vdash c :
        \type{Contr}}{\Gamma \vdash \scale e c : \type{Contr}}
      \quad
      \crule{\Gamma \vdash e : \tau~~~\Gamma, x: \tau \vdash c :
        \type{Contr}}{\Gamma \vdash \letc x e c : \type{Contr}}
    \\[1em]
    \crule{\Gamma \vdash e : \type{Bool}~~~\Gamma \vdash c_i :
      \type{Contr}}{\Gamma \vdash
      \ifwithin{e}{d}{c_1}{c_2} : \type{Contr}}
  \end{gather*}
\caption{Typing rules for contracts and expressions of CL.}
\label{fig:simpleType}
\end{figure}

Next, we discuss how adding template expressions to the contract language
affects its semantics. We extend the denotational semantics from
\cite{BahrBertholdElsman} to accommodate the idea of template expressions. The
semantics for the expression sublanguage stays unchanged, since these
expressions do not contain template expressions. That is, the semantics for an
expression $e \in \type{Exp}$ in Figure \ref{fig:contracts:syntax} is given by
the partial function $\eSem{e}{} : \gSem\Gamma \times \type{Env} \pto
\tysem\tau$.  On the other hand, we modify the semantic function for contacts
by adding a template environment as an argument:
\begin{align*}
  \cSem{c}{} &: \gSem{\Gamma} \times \type{Env} \times \type{TEnv} \pto \type{Trace}\\
  \type{Trace} &= \mathbb{N} \to \type{Trans}\\
  \type{Trans} &= \type{Party}\times\type{Party}\times\type{Asset} \to \mathbb{R}
\end{align*}
As the original contract semantics, it depends on the external environment
$\type{Env}:\mathbb{N} \times \type{Label} -> \mathbb{R} \cup \mathbb{B}$ and variable
assignments that map each free variable of type $\tau$ to a value in
$\tysem\tau$.
Where
\begin{equation}\label{eq:contracts:type-interp}
  \begin{aligned}
    \tysem{\type{Real}} = \reals\\
    \tysem{\type{Bool}} = \bools
  \end{aligned}
\end{equation}
Given a typing environment $\Gamma$, the set of \emph{variable
  assignments} in $\Gamma$, written $\gSem\Gamma$, is the set of all partial
mappings $\gamma$ from variable names to $\reals\cup\bools$ such that
$\gamma(x)\in\tysem\tau$ iff $x : \tau \in \Gamma$. The typing rules also remain
the same for expressions and for contracts. These rules are
presented in Figure \ref{fig:simpleType}, and the typing of expression operators
is given in Figure \ref{fig:typeOp}.

\begin{figure}
  \fbox{$\vdash op : \tau_1 \times \cdots \times \tau_n \to \tau $}
\begin{align*}
    &\vdash \oplus : \type{Real} \times \type{Real} \to \type{Real}
  &&\text{for } \oplus \in \{ +, -, \cdot, /, \mathit{max},\mathit{min}\}\\
    &\vdash \oplus : \type{Real} \times \type{Real} \to \type{Bool}
  &&\text{for } \oplus \in \{  \le, <, =, \ge, > \}
    \\
    &\vdash \oplus : \type{Bool} \times \type{Bool} \to \type{Bool} \hspace{8pt}
  &&\text{for } \oplus \in \{  \land, \lor \}
    \\
    &\vdash \neg : \type{Bool} \to \type{Bool}
    \\
    &\vdash \mathbf{if} : \type{Bool} \times \tau \times \tau \to \tau \quad
  &&\text{for } \tau \in \{\type{Real},\type{Bool}\}
  \end{align*}
  \caption{Typing of expression operators.}
  \label{fig:typeOp}
\end{figure}

By the result on the semantics of contracts \cite[Proposition
  3]{BahrBertholdElsman}, for well-typed expressions and well-typed closed
contracts semantic functions $\eSem{-}{}$ and $\cSem{-}{}$ are total. The
semantics for expressions and contracts is given in Figure \ref{fig:Esem}.

\begin{figure}
  \fbox{$\tSem{t}{} : \type{TEnv} \to \nats$}
  \[ \tSem{n}\delta = n \qquad \qquad \tSem{v}\delta = \delta(v) \]
  \fbox{$\eSem{e}{} : \gSem\Gamma \times \type{Env} \to \tysem\tau$}
\begin{align*}
    \eSem{r}{\gamma,\rho} &= r;\quad\!\eSem{b}{\gamma,\rho} = b;\quad\!\eSem{x}{\gamma,\rho} = \gamma(x)\\
    \eSem{\obs(l,t)}{\gamma,\rho} &= \rho(l,t)\\
    \eSem{op(e_1,\ldots,e_n)}{\gamma,\rho} &=
                                             \opsem{op}(\eSem{e_1}{\gamma,\rho},
                                             \ldots,
                                             \eSem{e_n}{\gamma,\rho})\\
  \eSem{\acc(\lambda x.\ e_1,d,e_2)}{\gamma,\rho} &=
  \begin{cases}
    \eSem{e_2} {\gamma,\rho}&\text{if } d = 0\\
    \eSem{e_1} {\gamma[x\mapsto v],\rho} & \text{if } d > 0
  \end{cases}
  \\
  \text{where } \qquad &v =
  \eSem{\acc(e_1,d-1,e_2)}{\gamma,\advEnv{\rho}{-1}}
  \end{align*}
\\[1em]
\fbox{$\cSem{c}{} : \gSem{\Gamma} \times \type{Env} \times \type{TEnv} \to \mathbb{N} \to \type{Trans}$}
\begin{align*}
  \type{Trans} &= \type{Party}\times\type{Party}\times\type{Asset} \to \mathbb{R}\\
 \cSem{\zero}{\gamma,\rho,\delta} &=
\lambda n . \lambda t . 0
\\
\cSem{\scale e c}{\gamma,\rho,\delta} &=
\lambda n . \lambda (p,q,a) . \eSem{e}{\gamma,\rho,\delta} \cdot \cSem{c}{\gamma,\rho,\delta} (n) (p,q,a)
\quad \text{where}\\
&(- \cdot -) : \reals -> \type{Trace} -> \type{Trace}~\\
&\text{denotes trace multiplication defined pointwise.}\\
\cSem{\both{c_1}{c_2}}{\gamma,\rho,\delta} &= \lambda n . \lambda t
                                   . \cSem{c_1}{\gamma,\rho,\delta}(n)(t) +
                                   \cSem{c_2}{\gamma,\rho,\delta}(n)(t)
\quad \text{where}\\
&(- + -) : \type{Trace} -> \type{Trace} -> \type{Trace}~\\
&\text{denotes trace addition defined pointwise.}\\
\cSem{\transl t c}{\gamma,\rho,\delta} &= \mathit{delay(\tSem{T}{\delta},\cSem c
  {\gamma,\rho,\delta})}, \quad \text{where}\\
\mathit{delay}(d,f) &=\lambda n .
\begin{cases}
f (n-d)&\text{if } n \ge d \\
\lambda x . 0 &\text{otherwise}
\end{cases}
  \end{align*}

  \begin{align*}
&\cSem{\transfer{p}{q}{a}}{\gamma,\rho,\delta} =
\begin{cases}
  \lambda n . \lambda t . 0 &\text{if } p = q\\
  \mathit{unit}_{a,p,q} &\text{otherwise,}\quad \text {where}
\end{cases}
\\
&\mathit{unit}_{a,p,q}(n)(p',q',b) =
\begin{cases}
1 &\text{if } b=a, p = p', q = q', n = 0\\
-1 &\text{if } b=a, p = q', q = p', n = 0\\
0 &\text{otherwise}
\end{cases}
\\
&\cSem{\letc x e c}{\gamma,\rho,\delta} = \cSem{c}{\gamma[x \mapsto
                                  v],\rho},\text{ where } v = \eSem{e}{\gamma,\rho,\delta}
\\
&\cSem{\ifwithin e t {c_1}{c_2}}{\gamma,\rho,\delta}
 = \mathit{iter}(\tSem{t}{\delta},\rho),\text{ where}\\
&\mathit{iter}(i,\rho')=\\
&\qquad\begin{cases}
\cSem{c_1}{\gamma,\rho'} &\text{if }
\eSem{e}{\gamma,\rho'} = \mathit{true} \\
\cSem{c_2}{\gamma,\rho'} &\text{if } \eSem{e}{\gamma,\rho'}
=\mathit{false} \land i = 0 \\
\mathit{delay}(1, \mathit{iter}(i-1,\advEnv{\rho'}{1})) &\text{if }
\eSem{e}{\gamma,\rho'}=\mathit{false} \land i > 0
\end{cases}
\end{align*}
  \caption{Denotational semantics of expressions and contracts of CL.}
\label{fig:Esem}
\end{figure}

We define an \emph{instantiation function} that takes a contract and a template environment
containing values for template variables, and produces another contract that does not
contain template variables by replacing all occurrences of template variables with corresponding
values from the template environment.
\begin{defn}[Instantiation function]\label{eq:contratcs:inst-func}
\begin{align*}
  \instdec &: \type{Cont}r \times \type{TEnv} -> \type{Contr}\\
  \inst {\zero} {\delta} &= \zero\\
  \inst {\letA~e~\letB~c} {\delta} &= \inst c \delta\\
  \inst {\transfer{p_1}{p_2}a} \delta &= {\transfer{p_1}{p_2}{a}}\\
  \inst {\scale e c} \delta &= \inst c \delta\\
  \inst {\transl t c} \delta &= \transl{\tSem{t}{\delta}}{\inst{c}{\delta}}\\
  \inst {\both{c_1}{c_2}} \delta &= \both{\inst{c_1}{\delta}}{\inst{c_2}{\delta}}\\
  \inst {\ifwithin e t {c_1} {c_2}} \delta &= \ifwithin{e}{\tSem{t}{\delta}}
        {\inst{c_1}\delta} {\inst{c_2}\delta}
\end{align*}
\end{defn}

We define an inductive predicate that holds only for contract expressions
without template variables (Figure \ref{fig:contracts:template-closed}). We call such
contracts \emph{template-closed}.
\begin{figure}
  \fbox{$\TC c$}
  \begin{mathpar}
  \inferrule{ }{\TC \zero}
  \and
  \inferrule{\TC c}{\TC {\letA~e~\letB~c}}
  \and
  \inferrule{ }{\TC{\transfer{p_1}{p_2}{a}}}
  \and
  \inferrule{\TC c}{\TC {\scale e c}}
  \and
  \inferrule{n~\text{is a numeral} \\ \TC c}{{\transl n c}}
  \and
  \inferrule{\TC{c_1} \\ \TC{c_2}}{\TC {\both{c_1}{c_2}}}
  \and
  \inferrule{n~\text{is a numeral} \\ \TC{c_1} \\ \TC{c_2}}
            {\TC {\ifwithin e n {c_1} {c_2}}}
\end{mathpar}
\caption{Template-closed contracts.}\label{fig:contracts:template-closed}
\end{figure}

\begin{lemma}
  It is straightforward to establish the following fact:
  for any contract $c$ and template environment $\delta$, application
  of the instantiation function gives a template-closed contract:
  \[\TC{\inst c \delta}\]
\end{lemma}
\begin{proof}
  By induction on the structure of $c$.
\end{proof}

\begin{lemma}[Instantiation soundness]
  For any contract $c$, template environments $\delta$ and $\delta'$, external environment
  $\rho$, and any value environment $\gamma$, the contract $c$ and
  $\inst c \delta$ are semantically equivalent.
  That is, $\cSem{c}{\gamma,\rho,\delta} = \cSem{\inst{c}{\delta}}{\gamma,\rho,\delta'}$.
\end{lemma}
\begin{proof}
  By induction on the structure of $c$. The case of $\ifwithin e t {c_1} {c_2}$ requires
  inner induction on $n = \tSem{t}{\delta}$.

  Notice also, that the semantic function on the right hand side takes
  arbitrary template environments $\delta'$, since after instantiation, the template
  environment does not affect the result.
\end{proof}

The reduction semantics of the contract language presented in \cite{BahrBertholdElsman}
remains the same, although, we make additional assumption that the contract expressions is
closed wrt. template variables. We provide the reduction semantics in Figure \ref{fig:redSem}
for completeness of the presentation.

The functions $\smartTransl n c$, $\smartBoth{c_1}{c_2}$, and $\smartScale e c$
represent corresponding smart constructors, which perform some simplifications
before constructing a corresponding contract. Expressions built using smart
constructors are semantically equivalent to the corresponding expressions that
use ordinary constructors. The $\mathsf{sp_E}$ function denotes contract
specialisation (see \cite[Section 4.1]{BahrBertholdElsman}).

\begin{figure}[!h]
  \fbox{$c \cRed {\gamma,\rho}{T}{c'}$}
  \begin{gather*}
    \crule{c \cRed {\gamma,\rho} T {c'}}{\transl 0 c \cRed {\gamma,\rho}
      {T} {c'}}
    \qquad
    \crule{}{\zero \cRed{\gamma,\rho}{T_0}{\zero}}
    \qquad
    \crule{}{\transfer{p}{q}{a} \cRed {\gamma,\rho}
      {T_{p,q,a}} {\zero}}
    \\[.5em]
    \crule{d > 0}{\transl d c \cRed \rho {T_0}
      {\smartTransl{d-1} c}}
    \quad
    \crule{c \cRed {\gamma,\rho} T {c'}\quad r =
      \mathsf{sp_E}(e,\gamma,\rho) \quad r \in \reals}{\scale e c\cRed{\gamma,\rho}{r
      * T}
      {\smartScale r {c'}}}
    \\[.5em]
    \crule{c_i \cRed{\gamma,\rho} {T_i} {c'_i}}{\both{c_1}{c_2}
      \cRed{\gamma,\rho} {T_1 + T_2} {\smartBoth{c'_1 }{c'_2}}}
    \qquad
    \crule{c \cRed{\gamma,\rho}{T_0}{c'}\quad e' =
      \mathsf{sp_E}(e,\gamma,\rho)}{\scale e c\cRed{\gamma,\rho}{T_0}
      {\smartScale{\promote{-1} {e'}}{c'}}}
    \\[.5em]
    \crule{
      \begin{gathered}
        \mathsf{sp_E}(e,\gamma, \rho) = e'\\
        c \cRed{\gamma',\rho}T  c'
      \end{gathered}
      \qquad \gamma' =
      \begin{cases}
        \gamma[x\mapsto e'] & \text{if } e' \in \reals \cup \bools\\
        \gamma&\text{otherwise}
      \end{cases}
}{\letc x e c
      \cRed{\gamma,\rho} T \smartLetc x {\promote{-1}{e'}} {c}}
    \\[.5em]
    \crule{\mathsf{sp_E}(e,\gamma, \rho) =
      \mathit{false}\quad c_2\cRed{\gamma,\rho}T{c'}}{\ifwithin e 0
      {c_1}{c_2}\cRed{\gamma,\rho} T c'}
    \\[.5em]
    \crule{\mathsf{sp_E}(e,\gamma, \rho) =
      \mathit{true}\quad c_2\cRed{\gamma,\rho}T{c'}}{\ifwithin e d
      {c_1}{c_2}\cRed{\gamma,\rho} T c'}
    \\[.5em]
    \crule{\mathsf{sp_E}(e,\gamma, \rho) =
      \mathit{false}\quad d>0}{
      \begin{aligned}
        &\ifwithin e d {c_1} {c_2}
      \end{aligned}\;
      \cRed{\gamma,\rho}{T_0}
      \;
      \begin{aligned}
        &\ifwithin e {d-1} {c_1} {c_2}
      \end{aligned}}
  \end{gather*}
  \begin{align*}
    \text{where}\qquad
    T_0&= \lambda t. 0 \qquad   r*T=\lambda t. r\cdot T(t)\\
    T_1+T_2 &= \lambda t. T_1(t) + T_2(t)\\
    T_{p,q,a} &= \lambda (p',q',a').
    \begin{cases}
      1 &\text{if } (p',q',a') = (p,q,a)\\
      -1 &\text{if } (p',q',a') = (q,p,a)\\
      0 &\text{otherwise}
    \end{cases}
  \end{align*}
  \caption{Contract reduction semantics assuming $\TC c$.}
  \label{fig:redSem}
\end{figure}

\begin{example}
  Let us consider a simple example of the contract reduction.  We take the
  following contract containing two transfers: one transfer is scheduled on the
  current day (no translation into the future), and another transfer is
  scheduled on the following day.
  \begin{lstlisting}
    $c$ $\defeq$ both(transfer(you,me),
                translate(1,transfer(you,me)))
  \end{lstlisting}
  We get the following derivation tree for a one-step reduction of $c$:
  \[
  \inferrule{\inferrule{ }{\transfer{you}{me}{USD}} \cRed{}{T_{you,me}} \zero\\\\
    \\\\
    \inferrule{1 > 0} {\transl{1}{\transfer{you}{me}{USD}}} \cRed{}{T_0}
    \transfer{you}{me}{USD}}
              {c \cRed{}{T_{you,me}} \transfer{you}{me}{USD}}
  \]
  There is an implicit simplification in the result of reduction due to the usage of
  smart constructors:
  \[\smartBoth{\zero}{\smartTransl{0}{\transfer{you}{me}{USD}}} = \transfer{you}{me}{USD}\]
\end{example}

\subsection{Traces as a Vector Space}\label{subsec:contracts:vector-space}
In Section 2.4 of the paper on CL \cite{BahrBertholdElsman} it
was mentioned that the set $\type{Trans}$ of transfers between parties forms a
vector space.  In this section we will make this precise and go a bit further,
discussing the set of traces $\type{Trace}$ and the $\id{delay}$ operation on
traces.

First, we recall the definition of a vector space.
\begin{defn}[Vector Space]
  A \emph{vector space} over a field $\mathbb{F}$ is a set $V$ equipped with
  two operations:
  \begin{itemize}
  \item vector addition $- + - : V \times V -> V$
  \item scalar multiplication $- \cdot - : \mathbb{F} \times V -> V$
  \end{itemize}
  These operations satisfy the following axioms.
  \begin{itemize}
  \item associativity of vector addition:
    $\forall u,v,w \in V, ~ (u + v) + w = u + (v + w)$;
  \item commutativity of vector addition: $\forall u,v \in V. u + v = v + u$;
  \item identity of vector addition: there exists a \emph{zero vector} $\mathbf{0}$,
    s.t $\forall v \in V. ~ v + \mathbf{0} = v$;
  \item inverse of vector addition : for every $v \in V$ there exists an \emph{additive inverse}
    $-v$, s.t. $v + (-v) = \mathbf{0}$;
  \item compatibility of scalar multiplication with field multiplication:
    $\forall a,b \in \mathbb{F},~v \in V. ~(ab) \cdot v = a \cdot (b \cdot v)$
  \item identity of scalar multiplication: $\forall v \in V.~1 \cdot u = u$
  \item distributivity of scalar multiplication over vector addition:
    $\forall a \in \mathbb{F}, v,u \in V. ~a \cdot (v + u) = a \cdot v + a \cdot u$;
  \item distributivity of scalar multiplication over field addition:
    $\forall a,b \in \mathbb{F},~v \in V.~ (a + b) \cdot v = a \cdot v + b \cdot v$
  \end{itemize}
\end{defn}

Transfers are defined as functions to real numbers:
\[\type{Trans} = \type{Party}\times\type{Party}\times\type{Asset} \to \mathbb{R}\]
It is a well-known fact that functions to a field form a vectors space with operations
given pointwise.

We shall spell out explicitly how transfers form a vector space. We use $T$ to denote
elements of $\type{Trans}$. Operations on elements of $\type{Trans}$ are the following:
\begin{itemize}
\item $T_1 + T_2 = \lambda p.~T_1(p) + T_2(p)$,
  where $p : \type{Party}\times\type{Party}\times\type{Asset}$
\item $r \cdot T = \lambda p.~r \cdot T(p)$
\end{itemize}
Let us now check that these two operations satisfy the axioms of a vector space.
\begin{itemize}
\item associativity and commutativity of transfer addition follow from the
  properties of addition on real numbers;
\item the zero vector is just a constantly zero transfer $\mathbf{0} = \lambda p.~0$, and its property
   follows from addition with zero in
  $\reals$: $(T + \mathbf{0})(p) = T(p) + 0 = T(p)$;
\item the inverse is also just a negation from $\reals$: $(-T) = \lambda p. (-T(p))$;
\item identity of scalar multiplication follows from multiplication of real numbers;
\item both distributivity laws follow from the fact that $\reals$ is a field.
\end{itemize}
As one can see, we have never used any specific properties of $\type{Trans}$ to
define the operations, or in our reasoning about vector space axioms.  The fact
that $\type{Trans}$ is a vector space is indeed an instance of a more general
result as we could have used any function to $\reals$ and show that it
satisfies vector space axioms.

Now we consider the type of the semantics function for contracts. We define a \emph{trace}
to be a function from the set of natural numbers $\nats$ to $\type{Trans}$
\[\type{Trace} = \mathbb{N} \to \type{Trans} \]
The semantics of well-types contracts is defined in terms of traces
\[ \cSem{c}{} : \gSem{\Gamma} \times \type{Env} \times \type{TEnv} \to \type{Trace} \]

Traces are functions to $\type{Trans}$ and we have shown that $\type{Trans}$ is
a vector space. We can show that $\type{Trace}$ with operations defined pointwise
also a vector space. The argument is essentially the same as for transfers, but with all
the properties of real numbers replaced by the properties of $\type{Trans}$.

\begin{remark}
  In addition to the usual vector space structure, transfers satisfy an
  additional anti-symmetry axiom.  That is, for any parties $p_1$, $p_2$, and
  an asset $a$
  \[ T(p_1,p_2,a) = -T(p_2,p_1,a) \]
  This property extends pointwise to traces as well. This means that CL is
  interpreted into a subspace of the vector space $\type{Trace}$ given by all
  traces with anti-symmetry property.
\end{remark}

Contract combinators, like $\<scale>$ and $\<both>$ are interpreted as
operations of a vector space $\type{Trace}$: scalar multiplication and vector
addition respectively. In addition to these operations, the $\<translate>$ combinator of
the contract language could be interpreted in terms of vectors spaces. Let us
first recall a definition of a linear map.
\begin{thm}[Linear map]
  For two vector spaces $V$ and $W$ over the same field $\mathbb{F}$ we call
  the function $f : V -> W$ a linear map, if it preserves vector addition and
  vector multiplication. That is, it satisfies the following conditions (for
  any $v \in V$, $w \in W$ and $a \in \mathbb{F}$)
  \begin{align*}
    f (v + w) &= f(v) + f(w)\\
    f (a \cdot v) &= a \cdot f(v)
  \end{align*}
\end{thm}

Let as consider the delay operation on contracts (see \ref{fig:Esem}):
\begin{equation}\label{eq:contracts:delay}
  \begin{aligned}
    \delay - - &: \nats \times \type{Trace} -> \type{Trace}\\
    \mathit{delay}(d,f) &=\lambda n . \begin{cases}
      f (n-d)&\text{if } n \ge d \\
      \lambda x . 0 &\text{otherwise}
    \end{cases}
  \end{aligned}
\end{equation}
For some fixed $n \in \nats$ we have
$\delay n - : \type{Trace} -> \type{Trace}$.
We know that $\type{Trace}$ is a vector space, so we can ask if
the $\mathit{delay}$ function is a linear map.
\begin{lemma}
  The $\mathit{delay}$ function defined by equation \ref{eq:contracts:delay} is
  a linear map.
\end{lemma}
\begin{proof}
  First, we show that $\mathit{delay}$ preserves vector addition.
  Fix some $n \in \nats$, we write $\mathit{delay}_n$ for $\delay n -$, and
  $\overline{T}$ for elements on $\type{Trace}$.
  By functional extensionality, for any $t \in \nats$, we have to show
  \[ \mathit{delay}_n(\overline{T}_1 + \overline{T}_2)(t) =
  \mathit{delay}_n(\overline{T}_1)(t) + \mathit{delay}_n(\overline{T}_2)(t) \]
  We proceed by cases of $t \ge n$.
  \begin{itemize}
    \item $t \ge n$:
    \begin{align*}
      \mathit{delay}_n(\overline{T}_1 + \overline{T}_2)(t)
      & = \text{by definition of $\mathit{delay}$ for $t \ge n$} \\
      &= (\overline{T}_1 + \overline{T}_2)(t-n) \\
      &= \text{by definition of pointwise addition of traces}\\
      &= \overline{T}_1(t-n) + \overline{T}_2(t-n)\\
      &= \text{by definition of $\mathit{delay}$ for $\overline{T}_1$ and $\overline{T}_2$}\\
      &= \mathit{delay}_n(\overline{T}_1)(t) + \mathit{delay}_n(\overline{T}_2)(t)
    \end{align*}
  \item otherwise:
    \[ \mathbf{0} = \mathbf{0} + \mathbf{0} \]
    where $\mathbf{0} = \lambda x.0$ is the constantly zero transfer.
    This holds by the identity of vector addition axiom.
  \end{itemize}

  Now we show that $\mathit{delay}$ preserves scalar multiplication. Fix some $n
  \in \nats$. By functional extensionality, for any $t \in \nats$ we have to
  show
  \[ \mathit{delay}_n(r \cdot \overline{T})(t) = r \cdot \mathit{delay}_n(\overline{T})(t) \]
  We proceed by cases of $t \ge n$.
  \begin{itemize}
    \item $t \ge n$:
    \begin{align*}
      \mathit{delay}_n(r \cdot \overline{T})(t)
      & = \text{by definition of $\mathit{delay}$ for $t \ge n$} \\
      &= (r \cdot \overline{T})(t-n) \\
      &= \text{by definition of pointwise scalar multiplication of traces}\\
      &= r \cdot \overline{T}(t-n)\\
      &= \text{by definition of $\mathit{delay}$ for $\overline{T}$}\\
      &= r \cdot \mathit{delay}_n(\overline{T})(t)
    \end{align*}
  \item otherwise:
    \[ \mathbf{0} = r \cdot \mathbf{0} \]
    where $\mathbf{0} = \lambda x.0$.  This is one of the properties of
    vectors, which can be derived from the basic vector space axioms.
  \end{itemize}
\end{proof}

Knowing that $\type{Trace}$ is a vector space gives us the important set of
properties, which are used in the proofs involving the contract semantics. We will
see some examples in the proof of compilation soundness in Section
\ref{sec:contracts:compile}. Moreover, a many contract equivalences are direct
reflections of vector space axioms.

\section{The Payoff Intermediate Language}\label{sec:contracts:payoffs}
\subsection{Motivation}
The contract language allows for capturing different aspects of financial
contracts.  We consider a particular use case for the contract language, where
one wants to calculate an estimated price of a contract according to some
stochastic model by performing simulations. Simulations is often implemented
using Monte Carlo techniques, for instance, by evaluating a contract price at
current time for randomly generated possible market scenarios and discounting
the outcome according to some model. A software component that implements such a procedure
is called a \emph{pricing engine} and aims to be very efficient in performing
large amount of calculations by exploiting the parallelism
\cite{Andreetta:2016:FPF:2952301.2898354}. For this use case, one has to
 take the following aspects into account:
\begin{itemize}
\item Contracts should be represented as simple functions that take prices of
  assets involved in the contract (randomly generated by a pricing engine) and
  return one value corresponding to the overall outcome of the contract.
\item The resulting value of the contract should be discounted according to a
  given discount function.
\end{itemize}
One way of achieving this would be to implement an interpreter for the contract
language as part of a pricing engine. Although this approach is quite general,
interpreting a contract in the process of pricing will cause significant
performance overhead.  Moreover, it will be harder to reason about correctness
of the interpreter, since it could require non-trivial encoding in languages
targeting GPGPU devices. For that reason we take another
approach: translating a contract from CL to an intermediate
representation and, eventually, to a function in the pricing engine
implementation language.

The main motivation behind the payoff language is to bridge the gap between
CL and programming languages usually used to implement pricing
engines. The payoff language should be relatively easy to compile to various
target languages such as Haskell, Futhark \cite{henriksen2014size}, or
OpenCL. We would like to consider a language containing fewer domain-specific
features and being closer to a subset of some general purpose language, making a
mapping from the payoff language to a target language
straightforward. Moreover, we would like to parameterise payoff expressions
with template expressions, like in our extended CL. In addition to
template expressions we want payoff expressions to capture the dynamic nature of
contracts: reduction with the passage of time.  This feature is usually not
present in most of payoff languages, but it is important for efficient
interaction with a pricing engine. Since our target languages include
high-performance languages for GPGPU computing, and payoff functions are
usually relatively small pieces of code, for efficient execution one often
needs to inline these functions. If our payoff code was not parametric, it would
require recompiling of big portions of the pricing engine code base when contracts
evolve in time. Sometimes it is also necessary to estimate the sensitivity of
a contract to the passage of time. In this case one wants the time parameter to
be a part of the payoff expression.

\subsection{Syntax and Semantics}\label{subsec:contracts:payoffs-syn-sem}
The payoff intermediate language is an expression language
($\id{il}\in\type{ILExpr}$) with binary and
unary operations, extended with conditionals and generalised
conditionals $\<loopif>$, behaving similarly to $\<ifWithin>$.
Template expressions ($t\in\type{TExprZ}$) in this language are extensions of the
template expressions of the contract language with integer literals
and addition. We will often refer to this language as \emph{payoff expressions}.
\[
\begin{split}
  \id{il} ~::=~ & \<now> ~|~\<model>(l,t) ~|~ \<if>(\id{il}_1,\id{il}_2,\id{il}_3) ~|~\\
                & \<loopif>(\id{il},\id{il},\id{il},t)~|~\<payoff>(t,p,p) ~|~ \\
                & \id{unop}(\id{il}) ~|~ \id{binop}(\id{il}_1,\id{il}_2) ~|~ t^e\\
  \id{unop} ~::=~  & \<neg> ~|~ \<not>\\
  \id{binop} ~::=~ &\<add>~|~\<mult>~|~\<sub>~|~\<lt>~|~\<and> ~|~\<or> ~|~ \<ltn> ~|~ \dots \\
  t^e ~::=~ & n ~|~ i ~|~ v ~|~ \<tplus>(t^e_1,t^e_2)
\end{split}
\]

The semantics of payoff expressions (Figure \ref{fig:contracts:payoffs-sem})
depends on environments $\rho \in \type{Env}$ and $\delta \in \type{TEnv}$
similarly to the semantics of the contract language. Payoff expressions can
evaluate to a value of type $\nats$, $\reals$, or $\bools$. We add $\nats$ to
the semantic domain, because we need to interpret the $\<now>$ construct, which
represents the ``current time'' parameter and template expressions $t^e$.
The semantics also depends on a \emph{discount function}
$d : \mathbb{N} -> \mathbb{R}$.  The $t_0 \in \mathbb{N}$ parameter is used to add
relative time shifts introduced by the semantics of $\<loopif>$; $t$ is a current
time, which will be important later, when we introduce a mechanism to cut
payoffs before a certain point in time.

The semantics for unary and binary operations is a straightforward mapping to corresponding
arithmetic and logical operations, provided that the arguments have
appropriate types. For example, $\sem{\<add>}(v_1, v_2) = v_1 + v_2$, if $v_1, v_2 \in \reals$.

\begin{figure}
\[
\begin{split}
\begin{aligned}
  \boxed{\ilSem{il}{} :
  \type{Env} \times \type{TEnv} \times \mathbb{N} \times \mathbb{N} \times
  (\mathbb{N} -> \mathbb{R}) \times \type{Party} \times \type{Party}
  \pto \nats \cup \reals \cup \bools}
\end{aligned} \\
\def\arraystretch{2}
\begin{array}{rl}
  \ilSem{t^e}{\rho,\delta,t_0,t,d,p_1,p_2} &= \tSem{t^e}{\delta} + t_0 \\
  \ilSem{unop(il)}{\rho,\delta,t_0,t,d,p_1,p_2} &=
  \sem{unop} (\ilSem{il}{\rho,\delta,t_0,t,d,p_1,p_2})\\
  \ilSem{binop(il_0,il_1)}{\rho,\delta,t_0,t,d,p_1,p_2} &=
  \sem{binop} (\ilSem{il_0}{\rho,\delta,t_0,t,d,p_1,p_2},~
  \ilSem{il_1}{\rho,\delta,t_0,t,d,p_1,p_2})\\
  \ilSem{\<model>(l,t^e)}{\rho,\delta,t_0,t,d,p_1,p_2} &=
  \rho(l,\tSem{t^e}{\delta} + t_0)\\
  \ilSem{\<now>}{\rho,\delta,t_0,t,d,p_1,p_2} &= t\\
  \ilSem{if(il_0,il_1,il_2)}{\rho,\delta,t_0,t,d,p_1,p_2} &=
      \begin{cases}
        \ilSem{il_1}{\rho,\delta,t_0,t,d,p_1,p_2} &
        \text{if } \ilSem{il_0}{\rho,\delta,t_0,t,d,p_1,p_2} = true \\
        \ilSem{il_2,}{\rho,\delta,t_0,t,d,p_1,p_2} &
        \text{if } \ilSem{il_0}{\rho,\delta,t_0,t,d,p_1,p_2} = false \\
      \end{cases}\\[1em]
  \ilSem{\<payoff>(t^e,p_1',p_2')}{\rho,\delta,t_0,t,d,p_1,p_2} &=
      \begin{cases}
        d (\tSem{t^e}{\delta}) &\text{if }  p_1' = p_1, p_2' = p_2\\
        - d (\tSem{t^e}{\delta}) &\text{if } p_1' = p_2, p_2' = p_1\\
        0 &\text{otherwise}
      \end{cases}
\end{array}
\end{split}
\]
\[
\begin{split}
  \ilSem{\<loopif>(il_0,il_1,il_2, t^e)}{\rho,\delta,t_0,t,d,p_1,p_2} =
      iter(\tSem{t^e}{\delta},t_0), where\\[1em]
  iter~n~t_0 =
  \begin{cases}
    \ilSem{il_1}{\rho,\delta,t_0,t,d,p_1,p_2} &
    \text{if } \ilSem{il_0}{\rho,\delta,t_0,t,d,p_1,p_2} = true \\
    \ilSem{il_2}{\rho,\delta,t_0,t,d,p_1,p_2} &
    \text{if } \ilSem{il_0}{\rho,\delta,t_0,t,d,p_1,p_2} = false \wedge n = 0 \\
    iter (n-1) (t_0+1) &
    \text{if } \ilSem{il_0}{\rho,\delta,t_0,t,d,p_1,p_2} = false \wedge i > 0
  \end{cases}
\end{split}
\]
  \caption{Semantics of payoff expressions.}
  \label{fig:contracts:payoffs-sem}
\end{figure}

The semantics for $\<loopif>$ is very similar to the semantics of
$\mathtt{ifWithin}$, although we do not ``advance'' external
environments. Instead, we increment parameter $t_0$, which is added to the time
shift when looking up a value in the semantics for $\<model>$.  The semantic
function $\ilSem{-}{}$ considers only payoffs between two parties $p_1$ and
$p_2$, which are given as the last two parameters. More precisely, it considers
payoffs from party $p_1$ to party $p_2$ as positive and as negative, if payoffs
go in the opposite direction.


Another way of defining the semantics could
be a \emph{bilateral view} on payoffs. In this case only cashflows to or from
one fixed party to any other party are considered. The semantics for the $\<payoff>$
then would be defined as follows:
\begin{gather*}
    \begin{aligned}[t]
      \ilSem{\<payoff>(t,p_1',p_2')}{\rho,\delta,t_0,t,d,p} &=
      \begin{cases}
        d (\tSem{t}{\delta}) &\text{if }  p_2' = p\\
        - d (\tSem{t}{\delta}) &\text{if } p_1' = p\\
        0 &\text{otherwise}
      \end{cases}
    \end{aligned}
\end{gather*}
\section{Compiling Contracts to Payoffs}\label{sec:contracts:compile}
The contract language (Figure \ref{fig:contracts:syntax}) consist of two levels,
namely constructors to build contracts ($c$) and constructors to build
expressions ($e$), which are used in the constructors for contracts ($\<scale>$,
$\<ifWithin>$, etc.).  We compile both levels into a single payoff language. The
compilation functions $\compileExp{-}{} : \type{Expr} \times \type{TExprZ} \pto
\type{ILExpr}$ and $\compileContr{-}{} : \type{Contr} \times \type{TExprZ} \pto
\type{ILExpr}$ are recursively defined on the syntax of expressions and
contracts, respectively, taking the starting time $t_0 \in \type{TExprZ}$ as a
parameter.
\[
\begin{aligned}[t]
      \compileExp{\<cond>(b,~e_0,~e_1])}{t_0} &=
      \<if>(\compileExp{b}{t_0}, \compileExp{e_0}{t_0}, \compileExp{e_1}{t_0})\\
      \compileExp{\<obs>(l,i)}{t_0} &= \<model>(l, \<tplus>(t_0,i))\\
      \compileContr{\<transfer>(p_1,p_2,a)}{t_0} &=  \<payoff>(t_0,p_1,p_2)\\
      \compileContr{\<scale>(e,c)}{t_0} &= \<mult>(\compileExp{e}{t_0}, \compileContr{c}{t_0})
 \end{aligned}
\]
\[
\begin{aligned}[t]
  \compileContr{\<zero>}{t_0} &= 0\\
  \compileContr{\<translate>(t,c)}{t_0} &= \compileContr{c}{\smartTplus(t_0, t)},
  \quad \text{where}\\
  &\smartTplus(t_1, t_2) =
  \begin{cases}
    t_1 + t_2 &\text{if } t_1,t_2 ~ \text{ are numerals}\\
    \<tplus>(t_1,t_2) &\text{otherwise}
  \end{cases}\\
  \compileContr{\<both>(c_0,c_1)}{t_0} &= \<add>(\compileContr{c_0}{t_0}, \compileContr{c_1}{t_0})\\
  \compileContr{\<ifWithin>(e,t,c_1,c_2)}{t_0} &=
  \<loopif>(\compileExp{e}{t_0},\compileContr{c_0}{t_0},\compileContr{c_1}{t_0}, t)
\end{aligned}
\]

The important point to note here is that all the relative time shifts in CL
are accumulated to the $t_0$ parameter. The resulting payoff expression
only contains lookups in the external environment where time is given
explicitly, and does not depend on nesting of time shifts as it was in the case
of $\transl{t}{c}$ in CL. Such a representation allows for a more
straightforward evaluation model. We also would like to emphasise that $\acc$
and $\letA$ constructs are not supported by our compilation procedure. On the
supported subset of the contract language, compilation functions
$\compileExp{e}{}$, and $\compileContr{c}{}$ are total.

Let us show an example of the contract compilation.  We consider the code of a
contract in CL extended with template expressions and demonstrate how
nested occurrences of $\mathtt{translate}$ are compiled to a payoff expression.

\begin{example}\label{ex:contracts:compile}
We consider the following contract (\icode{t0} and \icode{t1} denote template
variables): the party ``you'' transfer to the party ``me'' 100 USD in \icode{t0}
days in the future, and after \icode{t1} more days ``you'' transfers to ``me''
an amount equal to the difference between the current price of the AAPL ticker and 100
USD, provided that the price of AAPL is higher then 100 USD
(we use infix notation for arithmetic operations to make code more readable).
\begin{lstlisting}
c $=$
  translate(t0,
  both(scale(100.0, transfer(you,me)),
       translate(t1,
         if(obs(AAPL,0) > 100.0,
            scale(obs(AAPL,0) - 100.0, transfer(you, me)),
            zero)))
\end{lstlisting}

This contract compiles to the following code in the Payoff Intermediate Language:
\begin{lstlisting}
e $=$
  (100.0 * payoff(t0,you,me)) +
  if (model(AAPL,t0+t1) > 100.0,
     (model(AAPL,t0+t1) - 100.0) * payoff(t0+t1,you,me),
     0.0)
\end{lstlisting}
As one can see, all the nested occurrences of \icode{translate} construct were
accumulated from top to bottom. That is, in the \icode{if} case we calculate payoffs
and lookup for values of~``AAPL'' at time \icode{(t0+t1)}.
\end{example}

To be able to reason about soundness of the compilation process, one needs to
make a connection between the semantics of the two languages. For the expression sublanguage of
CL it is simple: we can just compare the values that original expression
and compiled expression evaluates to. In case of the contract language
(denoted by $c$ in Figure \ref{fig:contracts:syntax}) the situation is different, since the semantics
of contracts is given in terms of $\type{Trace}$, and expressions of the payoff
intermediate language evaluate to a single value. On the other hand, we know that
the compiled expression represents the sum of the contract cashflows with discount.

Before we state and sketch the proof of the soundness theorem, let us state some additional
lemmas.
\begin{lemma}[Delay scale]\label{lem:contracts:delay-scale}
  For any $s : \reals$, $t : \nats$ and a trace $\id{tr}$, the following holds:
  \[ \delay{t}{s \cdot \id{tr}} = s \cdot \delay{t}{\id{tr}} \]
\end{lemma}

\begin{lemma}[Delay add]\label{lem:contracts:delay-add}
  For any $t : \nats$ and traces $\id{tr}_1$,$\id{tr}_2$, the following holds:
  \[ \delay{t}{\id{tr}_1 + \id{tr}_2} = \delay{t}{\id{tr}_1} + \delay{t}{\id{tr}_2} \]
\end{lemma}

Lemmas \ref{lem:contracts:delay-scale} and \ref{lem:contracts:delay-add}
correspond to the fact that the $\id{delay}$ function is a linear map (see
Section \ref{subsec:contracts:vector-space}).

\begin{lemma}[Common factor]\label{lem:sum-common-factor}
  For any function $f : \nats -> \reals$, $s : \reals$, and $t_0, n : \nats$, we have the following obvious result:
  \[ \sum_{t=t_0}^{t_0 + n} s \cdot f(t) = s \cdot \sum_{t=t_0}^{t_0 + n} f(t)  \]
\end{lemma}

\begin{lemma}[Split sum]\label{lem:sum-split}
  For any function $f,g : \nats -> \reals$, we have the following:
  \[ \sum_{t=t_0}^{t_0 + n} \big(f(t) + g(t)\big) = \sum_{t=t_0}^{t_0 + n} f(t) + \sum_{t=t_0}^{t_0 + n} g(t)  \]
\end{lemma}

\begin{lemma}[Sum delay]\label{lem:contracts:sum-delay}
  For any trace $\id{tr}$ $t_0, t_1, t_2 : \nats$, discount function $d :
  \nats -> \reals$, and parties $p_1$, $p_2$, we have the following:
  \[ \sum_{t=t_0}^{t_0 + t_1 + t_2} d(t) \cdot \delay{t_0 + t_1}{tr}(t)(p1,p2) =
  \sum_{t=t_0+t_1}^{t_0 + t_1+ t_2} d(t) \cdot \delay{t_0 + t_1}{tr}(t)(p_1,p_2) \]
\end{lemma}
Intuitively, Lemma \ref{lem:contracts:sum-delay} says that summing up the delayed trace before the delay point
does not affect the result.

We assume a function $\id{HOR}: \type{TEnv}\times\type{Contr} -> \mathbb{N}$
that returns a conservative upper bound on the length of a contract. We often write
$\compileExp{e}0 = \id{il}$, or $\compileContr{c}0 = \id{il}$ to emphasise that the compilation
function returns some result. The compilation function satisfies the following properties:

\begin{thm}[Soundness]\label{thm:contracts:compile-sound}
  Assume parties $p_1$ and $p_2$ and discount function $d : \mathbb{N} -> \mathbb{R}$, environments
  $\rho \in \type{Env}$ and $\delta \in \type{TEnv}$.
  \begin{enumerate}[(i)]
    \item If ~$\compileExp{e}0 = \id{il}$ and
      $\eSem{e}{\rho,\delta} = v_1$ and $\ilSem{\id{il}}{\rho,\delta,0,0,d,p_1,p_2}
      = v_2$ then $v_1=v_2$.

    \item  If ~$\compileContr{c}0 = \id{il}$ and ~$\cSem{c}{\rho,\delta} = \id{tr}$,
      where $\id{tr} : \mathbb{N} -> \type{Party} \times \type{Party} -> \mathbb{R}$, and
      $\ilSem{\id{il}}{\rho,\delta,0,0,d,p_1,p_2} = v$ then
      \[\sum_{t=0}^{\synhor_\delta(c)} d(t) \times \id{tr}(t)(p_1,p_2) = v\]
  \end{enumerate}
\end{thm}
\begin{proof}
  We will outline the proof here to show which properties we use in particular cases.
  This proof is completely formalised in Coq and all the details can be found in the source code.

  The proof of part (\textit{i}) proceeds by induction on the structure of
  $e$ and mostly straightforward.

  To prove part (\textit{ii}) we first generalise the statement of the
  theorem for an arbitrary template expression $t_0 : \type{TExpZ}$ in place of
  $0$ as the initial value for the contact compilation function.\footnote{We
    need to generalise the initial value for $\<loopif>$ semantics from $0$ to
    some $n$ as well, i.e.  $\ilSem{\id{il}}{\rho,\delta,n,0,d,p_1,p_2} = v$. This can
    be done in the same way as for $t_0$, and we omit this generalisation in
    the proof presented here for readability} This is required because the compilation
  function aggregates all the nested time shift in the contract, and for the
  case of $\transl t c$ the induction hypothesis should be more general. The
  same approach is used to prove properties of tail-recursive functions
  (e.g. \icode{fold_left}). Once the initial for the compilation function
  becomes $t_0$, we have to ``compensate'' this by delaying the trace of the
  contract. One way to do it is to use $\<translate>$ in the theorem statement.
  That is, assumptions become $\compileContr{c}{t_0} = \id{il}$, and
  $\cSem{\transl{\tSem{t_0}{\delta}}{c}}{\rho,\delta} = \id{tr}$, and the
  conclusion becomes
  \[\sum_{t=t_0}^{t_0+\synhor_\delta(c)} d(t) \times \id{tr}(t)(p_1,p_2) = v\]
  The proof proceeds by induction on the structure of $c$. We have the
  following cases to consider.
  \begin{case}[$\zero$]
    We use the fact that the empty trace gives the zero transfer at any point of time.
  \end{case}
  \begin{case}[$\transfer{p_1}{p_2}{a}$]
    By case analysis on decidable equality of parties.
  \end{case}
  \begin{case}[$\transl{t}{c}$]
    We prove this case by application of the induction hypothesis with Lemma
    \ref{lem:contracts:sum-delay}.
  \end{case}
  \begin{case}[$\scale e c$]
    In this case the sum on the left-hand side (let us call it \emph{sigma})
    is equal to a product of two values $v_1 \cdot v_2$,
    where $v_1$ comes from the expression $e$ and $v_2$ from the contract
    $c$. The overall idea is to transform the sum into the product as well, so
    we can split the goal into two independent goals. We achieve that by using
    Lemmas \ref{lem:contracts:delay-scale} ($\id{delay}$ comes from the
    generalisation of the theorem statement for arbitrary $t_0$) and
    \ref{lem:sum-common-factor}.  After that we can prove the two goals using
    soundness of expression compilation (part (\textit{i}) of this theorem) for $v_1$,
    and the induction hypothesis for $v_2$.
  \end{case}
  \begin{case}[$\both{c_1}{c_2}$]
    In this case the sigma on the left-hand side is equal to the sum of two
    values $v_1 \cdot v_2$, where $v_1$ comes from the contract $c_1$ and $v_2$
    from the contract $c_2$. The overall idea is to transform the sigma into
    the sum of two sigmas, so that we can split the goal into two independent
    goals. Again, we achieve that by using Lemmas \ref{lem:contracts:delay-add}
    ($\id{delay}$ comes from the generalisation of the theorem statement for
    arbitrary $t_0$) and \ref{lem:sum-split}. After that, we can use the induction
    hypotheses on $c_1$ and $c_2$ to prove the two goals, but we have to do
    some extra work doing case analysis on $\synhor_\delta({\both{c_1}{c_2}})$,
    since it is defined as the maximum of horizons of the two contracts.
  \end{case}
    \begin{case}[$\ifwithin{e}{t}{c_1}{c_2}$]
      The proof proceeds by nested induction on $n = \tSem{t}{\delta}$, for
      $n \in \nats$.  In the base case and in the inductive step case we perform a case
      analysis on the result of expression evaluation
      $b = \eSem{e}{\rho,\delta}$, for $b \in \bools$. Moreover, in each
      subcase of the case analysis on $b$ we perform a case analysis for
      $\synhor_\delta({\ifwithin{e}{t}{c_1}{c_2}})$ for the same reason as in
      the case for $\both{c_1}{c_2}$.
  \end{case}
\end{proof}

\begin{remark}\label{rem:contracts:templ-expr-compile}
  The first version of a soundness proof was developed for the original
  contract language without template expressions. The proof was somewhat easier,
  since the aggregation of nested time shifts introduced by $\transl n c$
  constructs during compilation was implemented as addition of natural numbers,
  corresponding to time shifts. In the presence of template expressions, the
  compilation function builds a syntactic expression using the $\<tplus>$
  constructor. There are some places in proofs where it was crucial to use
  associativity of addition to prove the goal, but this does not work for
  template expressions. For example, $\<tplus>(\<tplus>(t_1,t_2),t_3)$ is not
  equal to $\<tplus>(t_1,\<tplus>(t_2,t_3))$, because these expressions
  represent different syntactic trees, although semantically
  equivalent. Instead of restating proofs in terms of this semantic equivalence
  (significantly complicating the proofs), we used the following approach.  The
  compilation function uses the \emph{smart constructor} $\smartTplus$ instead
  of just plain construction of the template expression. This allowed us to
  recover the property we needed to complete the soundness proof without
  altering too much of its structure.
\end{remark}

The soundness theorem (Theorem \ref{thm:contracts:compile-sound}) makes an
assumption that the compiled expression evaluates to some value. We do not
develop a type system for our payoff language to ensure this property. Instead,
we show that it is sufficient for a contract to be well-typed to ensure that
the compiled expression always evaluates to some value.

\begin{thm}[Total semantics for compiled contracts]\label{thm:contracts:well-typed-compiled}
  Assume parties $p_1$ and $p_2$ and discount function
  $d : \mathbb{N} -> \mathbb{R}$, well-typed external environment $\rho \in \type{Env}$, template
  environment $\delta \in \type{TEnv}$, and typing context $\Gamma$.  We have the
  following two results:
  \begin{enumerate}[(i)]
  \item for any
    $e \in \type{Exp}$, $t_0 \in \type{TExprZ}$ $t_0' \in \nats$, if
    $\Gamma \vdash e : \tau$
    $\compileExp{e}{t_0} = il$, then
    \[\exists v,~\ilSem{\id{il}}{\rho,\delta,t_0',0,d,p_1,p_2} = v
    \text{, ~and~} v \in \tysem{\tau}\]
  \item for any
    $c \in \type{Contr}$, $t_0$ $t_0'$, if $\Gamma \vdash c$ and $\compileContr{c}{t_0} = il$, then
    \[\exists v,~\ilSem{\id{il}}{\rho,\delta,t_0',0,d,p_1,p_2} = v  \text{, ~and~} v \in \reals\]
  \end{enumerate}
\end{thm}
\begin{proof}
  The proof for the statement (\textit{i}) proceeds by induction on the typing
  derivation for expressions (see Figure \ref{fig:simpleType}). The case for
  operations uses induction hypothesis, which gives a value. The part $ v \in \tysem{\tau}$
  in the conclusion serves as a logical relation allowing us to get typing information
  required for the particular operation. The case for $\obs(l,t)$ uses
  well-typedness of the external environments.

  The proof for statement (\textit{ii}) proceeds by induction on the typing
  derivation for contracts (see Figure \ref{fig:simpleType}), and uses the
  previously proved property of expressions (\textit{i}) for $\scale e c$ and
  $\ifwithin{e}{t}{c_1}{c_2}$.  The case \texttt{ifWithin} also requires a
  nested induction on $n = \tSem{t}\delta$, for $n \in \nats$, and case
  analysis on the result of expression the evaluation $b =
  \eSem{e}{\rho,\delta}$, for $b \in \bools$.
\end{proof}

Notice that Theorem \ref{thm:contracts:well-typed-compiled} holds for any
$t_0 \in \type{TExprZ}$ and $t_0' \in \nats$. These parameters do not affect
totality of the semantics and can be arbitrary, but it is crucial to add
appropriate delays, corresponding to arbitrary $t_0$ and $t_0'$ for the compilation
soundness property (see the proof of part (\textit{ii}) of Theorem
\ref{thm:contracts:compile-sound}).

Theorems \ref{thm:contracts:compile-sound} and
\ref{thm:contracts:well-typed-compiled} ensure that our compilation procedure
produces a payoff expression that evaluates to a value reflecting the aggregated
price of a contract after discounting.

\subsection{Avoiding recompilation}
To avoid recompilation of a contract when time moves forward, we define a
function $\cutPayoff{}$. This function is defined recursively on the syntax of
intermediate language expressions.
\begin{align*}
  \cutPayoff{\<now>} &= \<now>\\
  \cutPayoff{\<model>(l,t)} &= \<model>(l,t)\\
  \cutPayoff{\<if>(\id{il}_1,\id{il}_2,\id{il}_3)} &=
  \<if>(\cutPayoff{\id{il}_1},\cutPayoff{\id{il}_2},\cutPayoff{\id{il}_3})\\
  \cutPayoff{\<payoff>(t,p_1,p_2)} &= \<if>(t < \<now>,0,\<payoff>(t,p_1,p_2))\\
  \cutPayoff{\id{unop}(\id{il})} &= \id{unop}(\cutPayoff{\id{il}}) \\
  \cutPayoff{\id{binop}(\id{il}_1,\id{il}_2)} &=
  \id{binop}(\cutPayoff{\id{il}_1},\cutPayoff{\id{il}_2})
\end{align*}
The most important case is the case for $\<payoff>$.
The function wraps $\<payoff>$ with a condition guarding whether
this payoff affects the resulting value. For the remaining cases, the function recurses
on subexpressions and returns otherwise unmodified expressions.

\begin{example}
  Let us consider Example \ref{ex:contracts:compile} again and apply the $\cutPayoff$
  function to the expression \icode{e}:
  \begin{lstlisting}
cutPayoff(e) $=$
  (100.0 * disc(t0) * if(t0 < now, 0, payoff(you,me)) +
  if (model(AAPL,t0+t1) > 100.0,
     (model(AAPL,t0+t1) - 100.0) * disc(t0+t1) *
       if(t1+t0 < now, 0, payoff(you,me)),
      0.0)
  \end{lstlisting}

  Each \icode{payoff} in the payoff expression is now guarded by the condition,
  comparing the time of the particular payoff with \icode{now}. Notice that the
  templates variables \icode{t0} and \icode{t1} are mapped to concrete values in
  the template environment.
\end{example}

To be able to state a soundness property for the $\cutPayoff$ function we again
need to find a way to connect it to the semantics of CL. Since
$\cutPayoff$ deals with the dynamic behavior of the contract with respect to
time, it seems natural to formulate the soundness property in this case in
terms of contract reduction (Figure \ref{fig:redSem}). The semantics of the
payoff language takes the ``current time'' $t$ as a parameter. We should be
able to connect the $t$ parameter to the step of contract reduction.

\begin{remark}
  In the next lemmas we will implicitly assume parties $p_1$ and $p_2$, discount function
  $d : \mathbb{N} -> \mathbb{R}$, an external environment $\rho \in \type{Env}$, and a template
  environment $\delta \in \type{TEnv}$.
\end{remark}

\begin{defn}\label{def:contracts:equiv-at}
  We say that two payoff expressions $il_1$ and $il_2$ are equivalent at $(t_0,t)$ for
  if
  \[ \forall \rho, \delta, d, p_1, p_2. ~~\ilSem{il_1}{\rho,\delta,t_0,t,d,p_1,p_2}
  = \ilSem{il_2}{\rho,\delta,t_0,t,d,p_1,p_2} \]
  We write $il_1 \simeq_{(t_0,t)} il_2$ for this equivalence.
\end{defn}

Definition \ref{def:contracts:equiv-at} defines an equivalence parameterised by
two parameters: $t_0 \in \nats$ (a ``counter'' of $\<loopif>$ iterations), and
$t$ (time, up to which we want to ignore the payoffs).

Let us first start with a simple property showing that if we take ``current
time'' to be zero, i.e. $t=0$ in the semantic function, then the application of
$\cutPayoff$ should not have any effect.
\begin{lemma}
  Any $\id{il} : \type{ILExpr}$ is equivalent at $(t_0,0)$ to $\cutPayoff{\id{il}}$
  for any $n \in \nats$
  \[  \cutPayoff{\id{il}} \simeq_{(n,0)} \id{il} \]
\end{lemma}

The next lemmas show other properties when the value of a payoff expression
stays the same after application of $\cutPayoff$.

We have the following obvious property for the expression sublanguage of CL.
\begin{lemma}\label{lem:contracts:cutpayoff-compile-exp}
  For any contract expression $e \in \type{Exp}$, payoff expression $\id{il}$,
  and $t_0 \in \type{ILTExprZ}$,
  if $\compileExp{e}{t_0} = il$ then
  \[ \cutPayoff{il} = il \]
\end{lemma}
Notice that for compiled contract expressions $\type{Exp}$ we have a stronger
property, stated using just equality of terms, and not the equivalence. The
reason for this is that there are no $\<payoff>$ constructs in compiled contract
expressions and application of $\cutPayoff$ does not affect the payoff expression.

\begin{lemma}\label{lem:contracts:cutpayoff-compile-init}
  For any contract $c$, payoff expression $\id{il}$,
  $t_0 \in \type{ILTExprZ}$, $n \in \nats$, current time $t_{now}$,
  if $\compileContr{c}{t_0} = il$, and $t_{now} \leq \tSem{t_0}{\delta}$,
  then
  \[ \cutPayoff{\id{il}} \simeq_{(n,t_{now})} \id{il} \]
\end{lemma}

\begin{lemma}\label{lem:contracts:cutpayoff-loopif}
 For any contract $c$, payoff expression $\id{il}$, $n \in \nats$, current time $t_{now}$,
 if $t_{now} \leq n$,
 then \[ \cutPayoff{\id{il}} \simeq_{(n,t_{now})} \id{il} \]
\end{lemma}

Since the contract reduction relation uses smart constructors, we would like to
show how they interact with functions involved in the definition of soundness.
\begin{lemma}\label{lem:contracts:smartscale-hor}
  For any contract $c$, contract expression $e$, and template environment $\delta$,
  if $e$ is not a zero literal, then the following property holds:
  \[ \synhor_\delta(c) = \synhor_\delta(\smartScale{e}{c}) \]
\end{lemma}

\begin{lemma}\label{lem:contracts:smartboth-hor}
  For any contracts $c_1$, $c_2$ and a template environment $\delta$,
  then the following property holds:
  \[ \synhor_\delta(\both{c_1}{c_2}) = \synhor_\delta(\smartBoth{c_1}{c_2}) \]
\end{lemma}

Now, we can state a theorem relating the semantics of the payoff intermediate
language with the contract reduction semantics.

\begin{thm}[Contract compilation soudness wrt. contract reduction]\label{thm:contracts:soundness-red}
  We assume parties $p_1$, $p_2$, discount function $d : \nats -> \reals$.
  For any well-typed and template-closed contact $c$, i.e. we assume $\Gamma |-
  c$, and $\TC{c}$, an external environment $\rho' \in \type{Env}$ extending a
  partial external environment $\rho \in \type{Env_p}$, if $c$ steps to some
  $c'$ by the reduction relation $c ~\cRed{\rho}{T}~c'$, for some transfer
  $T$, such that $\cSem{c'}{\rho'/1} = \id{trace}$, and $\compileContr{c}0=il$,
  then
  \[
  \sum_{t'=0}^{\synhor_\delta(c')} d(t'+1) \times \id{trace}(t') =
  \ilSem{\cutPayoff{P}}{\rho',\delta,0,1,d,p_1,p_2}
  \]
  Where $\rho'/1$ denotes the external environment $\rho$ advanced by one time
  step:
  \[ \rho'/1 = \lambda (l,i).~\rho'(l,i + 1), \quad l \in \type{Label}, i \in \ints \]
\end{thm}
\begin{proof}
  The proof proceeds by induction on the derivation for the contract reduction.
  We group cases according to the shape of the contract $c$ and show only the
  general structure of the proof, highlighting which lemmas we use in
  particular cases.
  \begin{case}[$\zero$]
    Both values are zeros.
  \end{case}
  \begin{case}[$\transfer{p_1}{p_2}a$]
    By case analysis on decidable equality of parties.
    In the case parties are not equal, we observe that condition $0 < t_{now}$ is true, since
    $t_{now} = 1$ and that gives zeros on both sides of the equation.
  \end{case}
  \begin{case}[$\transl{t}{c}$]
    We have two subcases: $t = 0$ and $t = n+1$ for some $n \in \nats$.
    \begin{itemize}
      \item $t = 0$. By induction hypothesis.

      \item $t = n+1$. We observe that contract $c$ is translated $n+1$
        steps into the future.  Ignoring payoffs before time $t_{now}=1$ will not
        affect the result of the evaluation of the corresponding payoff expression,
        since all the potential payoffs can happen only after $n+1$ steps.

        That is, we use soundness of contract compilation (Theorem
        \ref{thm:contracts:compile-sound}, (\textit{ii})) with Lemma
        \ref{lem:contracts:cutpayoff-compile-init} and soundness of contract reduction
        (see \cite[Theorem 11]{BahrBertholdElsman}).
    \end{itemize}
  \end{case}
  \begin{case}[$\scale e c$]
    The smart constructor $\underline{\mathsf{scale}}$ does not preserve the
    symbolic horizon in general, i.e.
    $\synhor_{\delta}(\scale e c) \neq
    \synhor_{\delta}(\smartScale e c)$. This is due to the ``shortcut''
    behavior: if expression specialisation gives the zero literal, then the
    contract collapses to the empty one, giving zero horizon. For that reason we perform case
    analysis on the outcome of the expression specialisation.
    \begin{itemize}
      \item $\mathsf{sp_E}(e, \rho) = 0$. We use soundness of expression
        specialisation (see \cite[Theorem 10]{BahrBertholdElsman}), Theorem
        \ref{thm:contracts:compile-sound}(\textit{i}), and Lemma
        \ref{lem:contracts:cutpayoff-compile-exp} to proof this subcase.
      \item $\mathsf{sp_E}(e, \rho) \neq 0$. We use the same idea as in the
        case of $\<scale>$ in the proof of the compilation soundness theorem to
        split the goal into two cases. We prove the first goal using Theorem
        \ref{thm:contracts:compile-sound}(\textit{i}). The second goal we can prove
        by induction hypothesis and Lemma \ref{lem:contracts:smartscale-hor}, since we know that
        $\mathsf{sp_E}(e, \rho) \neq 0$.
    \end{itemize}
  \end{case}
  \begin{case}[$\both{c_1}{c_2}$]
    We use the same idea to split the goal into two as in the proof of
    Theorem \ref{thm:contracts:compile-sound}, and then use induction hypotheses with
    Lemma \ref{lem:contracts:smartboth-hor}.
  \end{case}
  \begin{case}[$\ifwithin{e}{t}{c_1}{c_2}$]
    This case consists of three subcases.
    \begin{itemize}
    \item $e$ evaluates to $true$.
      We prove this subcase by the induction hypothesis, using
      Theorem \ref{thm:contracts:compile-sound}(\textit{i})
      to show that the corresponding payoff expression also evaluates to $true$.
    \item $e$ evaluates to $false$ and $t = 0$. The proof is similar to the previous subcase.
    \item $e$ evaluates to $false$ and $t = n+1$ for some $n \in \nats$.  We
      observe that the starting value for $\<ifWithin>$ is $n+1$ and we know
      that $e$ evaluates to $false$. This means that all the potential payoffs
      can happen after at least one step, and evaluating the corresponding
      payoff expression at time $t_{now} = 1$ will not affect the result. We
      complete the proof by using the compilation soundness (Theorem
      \ref{thm:contracts:compile-sound}(\textit{ii})) and Lemma
      \ref{lem:contracts:cutpayoff-loopif}.
    \end{itemize}
  \end{case}
\end{proof}

From the contract pricing perspective, the partial external environment $\rho$
contains \emph{historical data} (e.g., historical stock quotes) and the extended
environment $\rho'$ is a union of two environments $\rho$ and $\rho''$, where
$\rho''$ contains \emph{simulated data}, produced by means of simulation in
the pricing engine (e.g., using Monte Carlo techniques).

One also might be interested in the following property. The following two ways
of using our compilation procedure give identical results:
\begin{itemize}
\item first reduce, compile, then evaluate;
\item first compile, apply $\cutPayoff$, and then evaluate, specifying the appropriate
  value for the ``current time'' parameter.
\end{itemize}

Let us introduce some notation first.  We fix the well-typed external environment
$\rho$, the partial environment $\rho'$, which is historically complete
($\rho'(l,i)$ is defined for all labels $l$ and $i \leq 0$), and a discount
function $d : \nats -> \reals$. Next, we assume that contracts are well-typed,
and closed both with respect to variables bound by $\<let>$ and template
variables, the compilation function is applied to supported constructs only, and
that the reduction function, corresponding to the reduction relation is total on
$\rho'$ (see \cite[Theorem 11]{BahrBertholdElsman}). This gives us the
following total functions:
  \begin{align*}
    red_{\rho'} : \type{Contr} -> \type{Contr}\\
    \compileContr{-}0 : \type{Contr} -> \type{ILExpr}\\
  \end{align*}

  These function correspond to the contract reduction function and the contract compilation function.
  We also define an evaluation function for compiled payoff expressions as a shortcut
  for the payoff expression semantics.
  \begin{align*}
    \id{evalAt}_{-} &: \nats -> \type{ILExpr} \times \type{Env} \times \type{Disc} -> \reals + \bools\\
    \id{evalAt}_t (e,\rho,d) &= \ilSem{e}{\rho,\emptyset,0,t,d,p_1,p_2}
  \end{align*}

for some parties $p_1$ and $p_2$. We know by Theorem
\ref{thm:contracts:well-typed-compiled} that $\id{evalAt}$ is total on payoff
expressions produced by the compilation function from well-typed contracts.

We summarise the property by depicting it as a commuting diagram (we give the
theorem here without a proof, but emphasise that we have formalised this
theorem in our Coq development).
\begin{thm}\label{thm:contracts:commutes}
  Given notation and assumptions above, the following diagram commutes:
\begin{center}
\begin{tikzcd}
  \type{Contr} \arrow[rr, "\id{red}_{\rho'}"] \arrow[d, "\mathsf{cutPayoff}~\circ~\compileContr{-}0"]
    &  & \type{Contr} \arrow[d, "\compileContr{-}0"] \\
    \type{ILExpr} \arrow[rd,"\id{evalAt}_1\tuple{-,\rho,d}"]
    & & \type{ILExpr} \arrow[ld,"\id{evalAt}_0\tuple{-,\rho/1,d/1}"] \\
    & \reals &
\end{tikzcd}
\end{center}
Where we write $\rho/1$ and $d/1$ for shifted one step external environment and
discount function, respectively.
\end{thm}
The above diagram gives rise to the following equation:
\[
\id{evalAt}_1\tuple{-,\rho,d} \circ \mathsf{cutPayoff}~\circ~\compileContr{-}0 =
\id{evalAt}_0\tuple{-,\rho/1,d/1} \circ \compileContr{-}0 \circ \id{red}_{\rho'}
\]
 This property shows, that we can use our implementation in both the settings:
 if a contract is compiled upfront with $\cutPayoff$, and if a reduced contract
 is compiled to the payoff for each time. The former use case allows more
 flexibility for users. For example, one can develop a system where users
 define contracts directly in terms of CL working in a specialised
 IDE. The latter case gives performance improvement allowing for the use a set of
 predefined financial instruments (or contract templates).  Adding a new
 instrument is possible, but requires recompilation.

 We also point out that the path (in a diagram given in Theorem \ref{thm:contracts:commutes})
 $\id{evalAt}_1\tuple{-,\rho,d} \circ \mathsf{cutPayoff}~\circ~\compileContr{-}0$,
 requires an external environment containing all historical data from the
 beginning of the contract and up to $t$.  While for the other path it is often
 possible to use only simulated data for pricing.

Avoiding recompilation using contract templates can significantly improve
performance especially on GPGPU devices. On the other hand, additional
conditional expressions are inserted into the code,
which results in a number of additional checks at
runtime. Experiments conducted with ``hand-compiled'' OpenCL code, which was
semantically equivalent to the payoff language code, showed that for the simple
contracts, such as European options, additional conditions, introduced by
$\cutPayoff{}$ do not significantly decrease performance.  The estimated
overhead was around $2.5$ percent, while compilation time is in the order of a
magnitude bigger than the total execution time.

\section{Formalisation in Coq}\label{sec:contracts:coq}
\lstset{language=Coq}
Our formalisation in Coq extends the previous work
\cite{BahrBertholdElsman} by introducing the concept of template expressions
and by developing a certified compilation technique for translating contracts
to payoff expressions. The modified denotational semantics has been presented
in Section \ref{subsec:contracts:contracts-syn-sem}. This modification required
us to propagate changes to all the proofs affected by the change of syntax and
semantics.  We start this section with a description of the original
formalisation, and then continue with modifications and additions made by the
author of this work.

The formalisation described in \cite{BahrBertholdElsman} uses an extrinsic
encoding of CL. That means that syntax is represented using Coq's
inductive data types, and a typing relation on these \emph{raw} terms are
given separately. For example, the type of the expression sublanguage is
defined as follows.
\begin{lstlisting}[language=Coq]
Inductive Exp : Set :=
  OpE (op : Op) (args : list Exp)
| Obs (l : ObsLabel) (i : Z)
| VarE (v : Var)
| Acc (f : Exp) (d : nat) (e : Exp).
\end{lstlisting}

One of the design choices in the definition of \icode{Exp} is to make the
constructor of operations \icode{OpE} take ``code'' for an operation
and the list of arguments. Such an implementation makes adding new operations
somewhat easier. Although, we would like to point out that this definition
is a \emph{nested} inductive definition (see \cite[Section 3.8]{cpdt}). In such
cases Coq cannot automatically derive strong enough induction principle,
and it should be defined manually. In the case of \icode{Exp} it is not hard to
see, that one needs to add a generalised induction hypothesis in case of
\icode{OpE}, saying that some predicate holds for all elements in the arguments
list.

Although the extrinsic encoding requires more work in terms of proving, it has a
big advantage for code extraction, since simple inductive data types are easier
to use in the Haskell wrapper for CL.

One of the consequences of this encoding is that semantic functions for
contracts $\type{Contr}$ and expressions $\type{Exp}$ are partial, since they
are defined on raw terms which might not be well-typed.  This partiality is
implemented with the \icode{Option} type, which is equivalent to Haskell's
\icode{Maybe}. To structure the usage of these partial functions, authors
define the \icode{Option} monad and use monadic binding
\begin{lstlisting}
  bind : forall A B : Type, option A -> (A -> option B) -> option B
\end{lstlisting}
to compose calls of partial functions together. The functions
\begin{lstlisting}
liftM: forall A B : Type,
  (A -> B) -> option A -> option B

liftM2 : forall A B C : Type,
  (A -> B -> C) -> option A -> option B -> option C

liftM3 : forall A B C D : Type,
 (A -> B -> C -> D) -> option A -> option B ->
                         option C -> option D
\end{lstlisting}
allow for a total function of one, two, or three arguments to be lifted to the
\icode{Option} type. The implementation includes poofs of some properties of
\icode{bind} and the lifting functions. These properties include cases for which an
expression evaluates to some value.
\begin{lstlisting}
bind_some : forall (A B : Type) (x : option A)
                  (v : B) (f : A -> option B),
    x >>= f = Some v -> exists x' : A, x = Some x' /\ f x' = Some v
\end{lstlisting}

The similar lemmas could be proved for other functions related to the
\icode{Option} type. To simplify the work with the \icode{Option} monad, the
implementation defines tactics in the \texttt{Ltac} language (part of Coq's
infrastructure). The tactics \icode{option_inv} and \icode{option_inv_auto} use
properties of operations like \icode{bind} and \icode{liftM} to invert
hypotheses like \icode{e = Some v}, where $e$ contains aforementioned
functions. The implementation uses some tactics from
\cite{pierce}. Particularly, the \icode{tryfalse} tactic is widely used.  It
tries to resolve the current goal by looking for contradictions in assumptions,
which conveniently removes impossible cases.

The original formalisation of the contract language was modified by introducing the type of
template expressions
\begin{lstlisting}
Parameter TVar : Set.
Inductive TExpr : Set := Tvar (t : TVar) | Tnum (n : nat).
\end{lstlisting}

We keep the type of variables abstract and do not impose any restrictions on
it. Although, one could add decidability of equality for \icode{TVar}, if
required, but we do not compare template variables in our formalisation.
We modify the definition of the type of contracts \icode{Contr} such that
constructors of expressions related to temporal aspects now accept \icode{TExpr}
instead of \icode{nat} (\icode{If} corresponds to $\<ifWithin>$):
\begin{lstlisting}
  Translate : TExpr -> Contr -> Contr
  If : Exp -> TExpr -> Contr -> Contr -> Contr.
\end{lstlisting}
and leave other constructors unmodified.

We define a template environment as a function type \icode{TEnv := TVar -> nat}
similarly to the definition of the external environment. Such a definition
allows for easier modification of existing code base in comparison with partial
mappings. According to the definitions in Section
\ref{subsec:contracts:contracts-syn-sem} we modify the semantic function for
contracts, and the symbolic horizon function to take an additional parameter of
type \icode{TEnv}. Propagation of these changes was not very problematic and
almost mechanical. Although, the first attempt to parameterise the reduction
relation with a template environment led to some problems, and we decided to
define the reduction relation only for template-closed contracts.  In most
cases it is sufficient to instantiate a contract, containing template variables
using the instantiation function (\ref{eq:contratcs:inst-func}), and then
reduce it to a new contract. Although instantiation requires a template
environment, containing all the mapping for template variables mentioned in the
contract, we do not consider this a big limitation.

The definition of the payoff intermediate language (following Section
\ref{subsec:contracts:payoffs-syn-sem}) also uses an extrinsic encoding to
represent raw terms as an inductive data type. We define one type for the payoff
language expressions \icode{ILExpr}, since there is no such separation as in
CL on contracts and expressions. The definition of template expressions used in the
definition of \icode{ILExpr} is an extended version of the definition of template expressions
\icode{TExpr} used in the contract language definition.
\begin{lstlisting}
Inductive ILTExpr : Set :=
  ILTplus (e1 : ILTExpr) (e2 : ILTExpr)
| ILTexpr (e : TExpr).

Inductive ILTExprZ : Set :=
  ILTplusZ (e1 : ILTExprZ) (e2 : ILTExprZ)
| ILTexprZ (e : ILTExpr)
| ILTnumZ (z : Z).
\end{lstlisting}

Notice that we use two different types of template expressions \icode{ILTExpr}
and \icode{ILTExprZ}. The former extends the definition of \icode{TExpr} with
the addition operation, and the latter extends it further with integer literals
and with the corresponding addition operation (recall that template expressions
used in CL can be either natural number literals or variables).  The
reason why we have to extend \icode{TExpr} with addition is that we want to
accumulate time shifts introduced by \icode{Translate} in one expression using
(syntactic) addition. In the expression sublanguage of CL, observables
can refer to the past by negative time indices. For that reason we introduce
the \icode{ILTExprZ} type.

The full definition of syntax for the payoff intermediate language in our Coq
formalisation looks as follows:
\begin{lstlisting}
Inductive ILExpr : Set :=
| ILIf : ILExpr -> ILExpr -> ILExpr -> ILExpr
| ILFloat : R -> ILExpr
| ILNat : nat -> ILExpr
| ILBool : bool -> ILExpr
| ILtexpr : ILTExpr -> ILExpr
| ILNow  :  ILExpr
| ILModel : ObsLabel -> ILTExprZ -> ILExpr
| ILUnExpr : ILUnOp -> ILExpr -> ILExpr
| ILBinExpr : ILBinOp -> ILExpr -> ILExpr -> ILExpr
| ILLoopIf : ILExpr -> ILExpr -> ILExpr -> TExpr -> ILExpr
| ILPayoff  : ILTExpr -> Party -> Party -> ILExpr.
\end{lstlisting}
Notice that we use template expressions, which could represent negative
numbers (\icode{ILTExprZ}) in the constructor \icode{ILModel}. This
constructor corresponds to observable values in the contract language
and allows for negative time indices corresponding to historical data.

We could have generalised our formalisation to deal with different types of
template variables and add a simple type system on top of the template
expression language, but we decided to keep our implementation simple, since
the main goal was to demonstrate that it is possible to extend the original
contract language to contract templates with temporal variables.

All the theorems and lemmas from Section \ref{sec:contracts:compile} are
completely formalised in our Coq development. We use a limited amount of proof
automation in the soundness proofs. We use proof automation
mainly in the proofs related to compilation of
contract expression sublanguage, since compilation is straightforward and
proofs are relatively easily to automate. Moreover, without the proof automation
one would have to consider a large number of very similar cases leading to code
duplication. In addition to \icode{option_inv_auto} mentioned above, we use a
tactic that helps to get rid of cases when expressions (a source expression in
\icode{Exp} and a target expression in \icode{ILEpxr}) evaluate to values of
different types (denoted by the corresponding constructor).
\begin{lstlisting}
Ltac destruct_vals := repeat (match goal with
                        | [x : Val |- _] => destruct x; tryfalse
                        | [x : ILVal |- _] => destruct x; tryfalse
                        end).
\end{lstlisting}
Where the \icode{Val} and \icode{IVal} types corresponds to values of the contract
expression sublanguage and the payoff expression language respectively.

Another tactic that significantly reduces the complexity of the proofs is the
\icode{omega} tactic from Coq's standard library. This tactic implements a
decision procedure for expressions in Presburger arithmetic. That is, goals can
be equations or inequations of integers, or natural numbers with addition and
multiplication by a constant. The tactic uses assumptions from the current
context to solve the goal automatically.

The principle we use in the organisation of the proofs is to use proof automation
to solve most trivial and tedious goals and to be more explicit about the proof
structure in cases requiring more sophisticated reasoning.

There are a few aspects that introduce complications to the development of proofs
of the compilation properties.
\begin{itemize}
  \item Accumulation of relative time shifts during compilation. Because of
    this we have to generalise our lemmas to any initial time \icode{t0}. The
    same holds for the semantics of \icode{loopif}, since there is an
    additional parameter in the semantics to implement iterative behavior.
  \item Presence of template expressions. The complications we faced with in
    this case are described in Remark
    \ref{rem:contracts:templ-expr-compile}. We resolves these complications with
    smart constructors, but it still adds some overhead.
  \item Conversion between types of numbers. We use integers and natural
    numbers (\icode{nat} and \icode{Z} type from the standard library). In some
    places, including the semantics of template expressions, we use a conversion
    from natural numbers to integers. This conversion makes automation with the
    \icode{omega} tactic more complicated, because it requires first to use
    properties of conversion, which is harder to automate. With the
    accumulation aspect, conversions add even more overheard.
  \item We use contract horizon in the statement of soundness theorems, which
    often leads to additional case analysis in proofs.
\end{itemize}

\subsection{Code Extraction}\label{subsec:contracts-code-extraction}
The Coq proof assistant allows for extracting Coq functions into programs in
some functional languages \cite{coqextract2008}. The implementation described
in \cite{BahrBertholdElsman} supports code extraction of the contract
type checker and contract manipulation functions into the Haskell programming
language. We extend the code extraction part of the implementation with
features related to contract templates and contract compilation. Particularly,
we extract Haskell implementations of the following functions:
\begin{itemize}
\item \icode{inst_contr} function that instantiates a given contract according to given
  template environment;
\item \icode{fromExp} function for compiling the contract expression sublanguage;
\item \icode{fromContr} function for compiling contract language constructs;
\item \icode{cutPayoff} function for parameterising a payoff expression with
  the ``current time''.
\item \icode{ILsem} semantic function for payoff expressions, which can be
  used as an interpreter.
\end{itemize}

We update the Haskell front end, exposing the contract language in convenient to
use form, with combinators for contract templates. We keep the original versions
of extended constructs, such as \icode{translate} and \icode{within} without
changes and add \icode{translateT} and \icode{withinT} combinators supporting
template variables.

Our implementation contains an extended collection of contract examples, examples
of contract compilation, and evaluation of resulting payoff expressions.

\subsection{Code Generation}\label{subsec:contarcts:haksell}
To exemplify how the payoff language can be used to produce a payoff function
in a subset of some general purpose language, we have implemented a code
generation procedure to the Haskell programming language. That is, we have
implemented the following chain of transformations:
\[ \text{CL} -> \text{Payoff Intermediate Language} -> \text{Haskell} \]
We make use of the code extraction mechanism described in Section
\ref{subsec:contracts-code-extraction} to obtain a certified compilation function,
which we use to translate expressions in CL to expressions in the payoff language.

The code generation procedure is (almost) a one-to-one mapping of the payoff
language constructs to Haskell expressions. One primitive, which we could
not map directly to Haskell build-in functions was the \icode{loopif}
construct. We have solved this issue by implementing \icode{loopif} as a
higher-order function. The implementation essentially follows the definition of
the semantics of \icode{loopif} in Coq.

\begin{lstlisting}[language=Haskell]
loopif :: Int -> Int -> (Int -> Bool) -> (Int -> a) -> (Int -> a) -> a
loopif n t0 b e1 e2 = let b' = b t0 in
                      $~$case b' of
                        True -> e1 t0
                        False -> case n of
                                   0 -> e2 t0
                                   _ -> loopif (n$-$1) (t0+1) b e1 e2

\end{lstlisting}

The resulting payoff function has the following signature:
\begin{lstlisting}[language=Haskell]
  payoff :: Map.Map ([Char], Int) Double -> Map.Map [Char] Int
            -> Int -> Party -> Party -> Double
\end{lstlisting}
That is, the function takes an external environment, a template environment,
current time, and two parties. The \icode{payoff} function calls the
\icode{payoffInternal} function, which takes an additional parameter -- an
initial value for the \icode{loopif} function needed for technical reasons.

\begin{example}~\\
  We apply the code generation procedure to the expression \icode{e} from Example
  \ref{ex:contracts:compile}. The result of code generation is given below.
  \begin{lstlisting}[language=Haskell]
module Examples.PayoffFunction where
import qualified Data.Map as Map
import BaseTypes
import Examples.BasePayoff
payoffInternal ext tenv t0 t_now p1 p2 =
  (100.0 * (if  (X== p1 && Y== p2) then 1
            else if (X== p2 && Y== p1) then $-$1 else 0)) +
  (if ((100.0 < (ext Map.! ("AAPL",(0 + (tenv Map.! "t1") +
                                   (tenv Map.! "t0") + 0+ t0)))))
   then ((((ext Map.! ("AAPL",(0 + (tenv Map.! "t1") +
                               (tenv Map.! "t0") + 0+ t0))) * 100.0) *
           (if  (X== p1 && Y== p2) then 1 else
              if  (X== p2 && Y== p1) then $-$1 else 0))) else 0.0)
payoff ext tenv t_now p1 p2=payoffInternal ext tenv 0 t_now p1 p2
  \end{lstlisting}
  As one can see, the external environment and the template environment are
  represented using Haskell's \icode{Data.Map}, and \icode{Map.!} is an infix
  notation for the lookup function. The Haskell code above makes use of the
  \icode{loopif} function with zero as the first argument. It is possible to
  replace it with the regular \icode{if} by a simple optimisation. One could also add
  more optimisations to our Coq implementation along with proofs of soundness.

  A module declaring the \icode{payoff} function can be used as an ordinary
  Haskell module as a part of the development requiring the payoff
  functions. For example, it could be used in the context of the FinPar
  benchmark \cite{Andreetta:2016:FPF:2952301.2898354}, which contains a Haskell
  implementation of pricing among other routines.  Moreover, the $\cutPayoff{}$
  function can be used to obtain a parameterised version of a payoff function
  in Haskell, allowing us to reproduce the contract reduction behavior.
\end{example}

\section{Conclusion}

This work extends the certified contract management system of
\cite{BahrBertholdElsman} with template expressions, which allows for drastic
performance improvements and reusability in terms of the concept of
instruments (i.e. contract templates).  Along with changes to the contract
language, we have developed a formalisation of the payoff intermediate language and
a certified compilation procedure in Coq. Our approach uses an extrinsic encoding,
since we are aiming at using code extraction for obtaining a correct implementation
of the compiler function that translates expressions in CL to payoff
expressions. We have also developed a technique allowing for capturing contract
development over time.

A number of important properties, including soundness of the translation from
CL to the payoff language have been proven in Coq. We have
exemplified how the payoff intermediate language can be used to generate code
in a target language by mapping payoff expressions to a subset of Haskell.

There are number of possibilities for future work:
\begin{itemize}
\item Generalise Theorems \ref{thm:contracts:soundness-red} and
  \ref{thm:contracts:commutes} to n-step contract reduction. We have defined
  some steps of the proofs in this generalised setting, but details still need
  to be worked out.
\item The representation of traces as functions $\nats -> \type{Trans}$ is
  equivalent to infinite streams of transfers. It would be interesting to
  explore this idea of using streams further, since observable values also can
  be naturally represented as streams.
\item Improve the design of the payoff intermediate language to support all
  CL constructs.  Also, adding $\<loopif>$ seems to be somewhat
  ad-hoc. It could be possible to a have more general language construct for iteration
  and compile $\<ifWithin>$ to a combination of iteration and conditions.
\item Implement a translation from the payoff intermediate language to the
  Futhark programming language for data-parallel GPGPU computations
  \cite{Henriksen:Futhark}. Futhark seems to be a natural choice as a
  target language, since there is an implementation of the pricing engine
  in Futhark, and we believe that mapping from payoff expressions to a
  subset of Futhark should be similar to our experience with Haskell.
\item Develop a formalised infrastructure to work with external environment
  representations. That is, instead of finite maps as in our Haskell code
  generator (Section \ref{subsec:contarcts:haksell}), one could use arrays to
  represent external environments. In this case one has to implement some
  reindexing scheme, since a naive translation of external environments could
  result in sparse arrays. Particularly, this will be important for targeting array
  languages like Futhark.
\end{itemize}

\chapter{Formalising Modules}\label{chpt:modules}
In this chapter we present a number of techniques that allow for formal
reasoning with nested and mutually inductive structures built up from finite
maps and sets (also called semantic objects), and at the same time allow for
working with binding structures over sets of variables. The techniques, which
build on the theory of nominal sets combined with the ability to work with
multiple isomorphic representations of finite maps, make it possible to give a
formal treatment, in Coq, of a higher-order module language for Futhark, an
optimising compiler targeting data-parallel architectures, such as GPGPUs
\cite{Henriksen:Futhark}. We want to emphasise that the main focus of this
chapter is on a formal development in the Coq proof assistant: encoding of
semantic objects, formal treatment of issues related to variable binding, and
proof techniques applied.

The rest of this chapter is structured as follows. In Section
\ref{sec:modules-motivation} we present the motivation and background for
development of a module system for the Futhark language.  In Section
\ref{sec:modules-spec} we introduce a formal system for the module language
specification. in Section \ref{subsec:modules-stlc}, we provide a well known
result demonstrating normalisation of the simply typed lambda-calculus (STLC)
using a logical relation argument.  The purpose of this exposition is to
motivate how we can late use a similar technique to prove the normalisation of
static interpretation of the module language.  We give basic definitions from
the theory of nominal sets and discuss motivations and applications of nominal
techniques to the module language formalisation in Section
\ref{sec:modules:var-binding-nominal}. In the same section we develop an
example in a simplified setting to demonstrate how nominal techniques apply to
our formalisation. We discuss our Coq development in Section
\ref{sec:modules-coq}. The aim of this section is to highlight specifics of the
development, and show applications of reasoning techniques developed to solve
issues related to representation of the structures given in the previous section.
Our development presented in Section \ref{sec:modules-coq}
is the first treatment of the static interpretation of modules in the style of
\cite{elsman99} in the Coq proof assistant. \note{Our Coq formalisation covers
  definitions required to state and prove an important property of the static
  interpretation: static interpretation for well-typed modules
  terminates. While another important property - the static interpretation
  procedure preserves the typing of target language expressions - is not
  covered by our implementation and left as future work.}

\section{Motivation}\label{sec:modules-motivation}
Modules in the style of Standard ML and OCaml provide a powerful
abstraction mechanism allowing for writing generic highly parameterised code. A
common issue with an abstraction mechanism is that it can introduce a runtime
overhead. For some application domains it is important to have static
guarantees that module abstractions introduce no overhead. This can be done by
statically interpreting the module system expressions at compile time.  This
technique is similar to the way C++ templates are eliminated at compile time
with the difference that using modules give more static guarantees by means of a
type system. The presented work extends the previous work to higher-order
modules~\cite{elsman99}.

As an application of the abstraction mechanism provided by higher-order modules
we consider a module language implemented on top of the monomorphic,
first-order functional data-parallel language Futhark, which features a number
of polymorphic second-order array combinators (SOACs) with parallel semantics,
such as \kw{map}, \kw{reduce}, \kw{scan}, and \kw{filter}, but has no support
for user-defined polymorphic higher-order functions. The module language allows
for defining certain kind of polymorphic functions in Futhark with the
guarantee that, at compile time, module level language constructs will be
compiled away.  That is, the module language gives rise to highly reusable
components, which, for instance, form the grounds of a Basis Library for
Futhark.

\begin{example}
  For the purpose of demonstrating static interpretation in action,
  consider the (contrived) example Futhark program.
  \begin{lstlisting}[basicstyle=\small,mathescape=true]
  module type MT = {
    module F: (X:{ val b:int } $\rar$ { val f:int$\rar$int })
  }
  module H = funct (M:MT) $\Rar$ M.F { val b = 8 }
  module Main =
    H ( { module F =
              funct (X:{ val b:int }) $\Rar$ { fun f(x:int) = X.b+x }
          })
  fun main (a:int) : int = Main.f a
  \end{lstlisting}
  The program declares a module type \texttt{MT} and a higher-order module
  \texttt{H}, which is applied to a module containing a parameterised module
  \texttt{F}. The result of the module application is a module containing a
  function \texttt{f} of type $\texttt{int}\rar \texttt{int}$. The contained
  function is called in the \texttt{main} function with the input to the
  program. Static interpretation partially evaluates the program to achieve the
  following result.

  \begin{lstlisting}[basicstyle=\small][language=Haskell]
    val b = 8
    fun f (x:int) = b + x
    fun main (a:int) = f a
  \end{lstlisting}

  The code snipped presents monomorphic target code, which can be composed,
  analysed, and compiled without any module language considerations. This
  feature provides the target language implementor with the essential
  meta-level abstraction property that the module language features are
  orthogonal to the domain of the source language.
\end{example}

\section{Normalisation in the Call-by-Value Simply-Typed Lambda Calculus}
\label{sec:modules:stlc-norm}
In this section we present a well-known result that simply-typed
lambda calculus (STLC) is normalising. We use a big-step
semantics with explicit closures, and assume a call-by-value evaluation
strategy for STLC. We use STLC as analogy for the static
interpretation of a module language, which we will present formally in
subsequent sections. For that reason, the argument usually used to
prove normalisation of STLC in some adapted form is also applicable to
the module language. We use Tait's method of logical relations
\cite{tait1967} in the proof we present in this section.

\label{subsec:modules-stlc}
We assume countably infinite set of variables, ranged over by $x$, and $i$ ranges
over integers.
\begin{align}
  \tau \in Ty::&= int ~|~ \tau_1 -> \tau_2
\end{align}

\begin{align}\label{eq:modules:lambda-syntax}
  e \in Lam  ::&= i ~|~ x ~|~ \lambda x.e ~|~ e_1 e_2
\end{align}

A context $\Gamma$ maps variables to types.
\begin{mathpar}
  \inferrule{ }{\Gamma \vdash i : int}~\deflabel{\textsc{ty-int}}\label{rule:ty-int}

  \and

  \inferrule{\Gamma(x) = \tau}{\Gamma \vdash x : \tau}
  ~\deflabel{\textsc{ty-var}}\label{rule:ty-var}

  \and

  \inferrule{\Gamma, x:\tau_1 \vdash e : \tau_2}
            {\Gamma \vdash \lambda x.e : \tau_1 -> \tau_2}
            ~\deflabel{\textsc{ty-lam}}\label{rule:ty-lam}

  \and

  \inferrule{\Gamma \vdash e_1 : \tau_1 -> \tau_2 \\ \Gamma \vdash e_2 : \tau_2}
            {\Gamma \vdash e_1 e_2 :\tau_2}
            ~\deflabel{\textsc{ty-app}}\label{rule:ty-app}

\end{mathpar}

We define a big-step call-be-value operational semantics with explicit
closures. Evaluation contexts $E$ map variables to values. Values can be either
closures, or integer literals.
\begin{align*}
  v \in Val  ::&= \Cl~E~x~e ~|~i
\end{align*}

\begin{mathpar}
  \infer{ }{ E |- i => i}
  ~\deflabel{\textsc{ev-int}}\label{rule:ev-int}
  \and
  \infer{ }{ E |- \lambda x.e ==> \Cl~E~x~e }
    ~\deflabel{\textsc{ev-lam}}\label{rule:ev-lam}
  \and
  \infer{E(x) = v}{E \vdash x ==> v}
  ~\deflabel{\textsc{ev-lam}}\label{rule:ev-var}
  \and
  \infer{
    E |- e_1 ==> \Cl~E_0~x~e_0 \\
    E |- e_2 ==> v_0 \\
    E_0[x \mapsto v_0] |- e_0 ==> v}
        {E |- e_1 e_2 ==> v}
        ~\deflabel{\textsc{ev-app}}\label{rule:ev-app}
\end{mathpar}

We define a logical relation, which we will use in the proof of normalisation
in Figure \ref{fig:modules:stlc-lr}. Similarly, we define a relation on typing
and evaluation contexts:
\begin{equation}
    \infer{dom(E) = dom(\Gamma) \\ \forall x.~E(x) = v \land \Gamma(x) = \tau =>  v |= \tau}
          {E |= \Gamma}
\end{equation}

\begin{figure}
\begin{mathpar}
  \infer{i \in \ints}{i |= int }
  ~\deflabel{\textsc{lr-int}}\label{rule:lr-int}
  \and
  \infer{v = \Cl~E~x~e \and
    (\forall v_1.~v_1 |= \tau_1 => \exists v_2.~E[x \mapsto v_1] |- e_0 => v_2 \land v_2 |= \tau_2)}
        {v |= \tau_1 -> \tau_2}
        ~\deflabel{\textsc{lr-arr}}\label{rule:lr-arr}
\end{mathpar}
\caption{Logical relation.}\label{fig:modules:stlc-lr}
\end{figure}

\begin{lemma}\label{lem:modules:stls-ctx}
  For any typing context $\Gamma$, evaluation context $E$, variable $x$, value $v$ and type $\tau$,
  if $E |= \Gamma$ and $v |= \tau$, then $E[x \mapsto v] |= \Gamma,x : \tau $.
\end{lemma}

\begin{thm}[Normalisation\footnote{We have developed a standalone formalisation
      of this theorem in Coq.\\ See \url{https://annenkov.github.io/stlcnorm/Stlc.stlc.html}}]
  \label{thm:modules:stlc-norm}
  For any typing context $\Gamma$, evaluation context $E$, term $e$ and type
  $\tau$, if $\overset{\mathcal{T}}{\Gamma |- e : \tau}$ and
  $\overset{\mathcal{R}}{\Gamma |= E}$ then there exists a value $v$, such that
  \[ E |- e ==> v \quad \text{and} \quad v |= \tau\]
  \note{$\mathcal{T}$ and $\mathcal{R}$ denote a typing derivation and a logical relation derivation respectively.}
\end{thm}
\begin{proof}
  The proof proceeds by induction on the typing derivation $\mathcal{T}$. We
  consider cases for lambda abstraction and application.
  \begin{case}[\nameref{rule:ty-lam}]
    $\mathcal{T}$ = $\inferrule{\overset{\mathcal{T}_1}{\Gamma, x:\tau_1 \vdash
        e : \tau_2}} {\Gamma \vdash \lambda x.e : \tau_1 -> \tau_2}$

    We take $v = \Cl ~E~x~e$. We get $E |- e ==> \Cl ~E~x~e$ from the
    evaluation rule for lambda-abstraction \nameref{rule:ev-lam}.  We have to
    show $\Cl ~E~x~e |= \tau_1 -> \tau_2$. By the rule of logical relation for
    the arrow type, suffices to show that assuming $v_1$,
    s.t. $\overset{\mathcal{R}_1}{v_1 |= \tau_1}$ we have $\exists v_2.~E[x
      \mapsto v_1] |- e_0 => v_2 \land v_2 |= \tau_2$.

   \noindent From Lemma \ref{lem:modules:stls-ctx}, with $\mathcal{R}$ from assumptions
    and $\mathcal{R}_1$, we get
    $\overset{\mathcal{R}_1}{E[x \mapsto v_1] |= \Gamma,x:\tau_1}$.

    \noindent We complete the proof of this case by using the induction hypothesis on $\mathcal{T}_1$
    with $\mathcal{R}_1$.
  \end{case}
  \begin{case}[\nameref{rule:ty-app}]
    $\mathcal{T}$ =
    $\inferrule{\overset{\mathcal{T}_1}{\Gamma \vdash e_1 : \tau_1 -> \tau_2} \\
      \overset{\mathcal{T}_2}{\Gamma \vdash e_2 : \tau_2}}
    {\Gamma \vdash e_1 e_2 :\tau_2}$

    \noindent By using the induction hypothesis on $\mathcal{T}_1$ with $\mathcal{R}$, we get:
    there exists some $v_1$, s.t.
    \[ \overset{\mathcal{E}_1}{E |- e_1 ==> v_1}~\text{and}~
    \overset{\mathcal{R}_1}{v_1 |= \tau_1 -> \tau_2} \]
    By using the induction hypothesis on $\mathcal{T}_2$ with $\mathcal{R}$, we get:
    there exists some $v_2$, s.t.
    \[ \overset{\mathcal{E}_2}{E |- e_2 ==> v_2}~\text{and}~
    \overset{\mathcal{R}_2}{v_2 |= \tau_1} \]

    \noindent From $\mathcal{R}_1$, we get $x$, $e_0$, $E_0$, s.t.
    \[ v= \Cl~E_0~x~e_0 \]
      and
      \begin{equation}
        \forall v'.~v' |= \tau_1 => \exists v''.~E[x \mapsto v'] |- e_0 =>
        v'' \land v'' |= \tau_2)\label{eq:modules-stlc-1}
      \end{equation}
      From (\ref{eq:modules-stlc-1}) with $v_2$ and $\mathcal{R}_2$, we get: there
      exists some $v_3$, s.t.
      \[\overset{\mathcal{E}_3}{E[x \mapsto v_2] |- e_0 =>  v_3} ~\text{and}~
      \overset{\mathcal{R}_3}{v_3 |= \tau_2} \]

      \noindent Now, take $v=v_3$ and use the rule for application
      (\nameref{rule:ev-app}) to construct the required derivation.
      \[ \infer{
        \overset{\mathcal{E}_1}{E |- e_1 ==> \Cl~E_0~x~e_0} \\
        \overset{\mathcal{E}_2}{E |- e_2 ==> v_2} \\
        E_0[x \mapsto v_2] |- e_0 ==> v_3}
         {E |- e_1 e_2 ==> v_3} \]
         From from $\mathcal{R}_3$, we can conclude $v_3 |= \tau_2$ as required.
  \end{case}
\end{proof}

\note{Theorem \ref{thm:modules:stlc-norm} can be used to show that a well-typed closed term always evaluates to some value.
For example, we have the following corollary.}
\begin{corollary}
  \note{Every closed term of type $int$ evaluates to an integer literal. That is, if $\{\}|- t : int$ then
    there exists $i : int$, s.t. $\{\} |- e ==> i$. Here, $\{\}$ is the empty context}
\end{corollary}
\begin{proof}
  \note{In Theorem \ref{thm:modules:stlc-norm}, take $\Gamma$ and $E$ to be the empty context $\{\}$.
  Then the claim follows immediately, since $\{\} |= \{\}$ is trivially satisfied.}
\end{proof}

\section{Formal Specification}\label{sec:modules-spec}
The module language can be considered parameterised over a core language,
which, for the purpose of the presentation, is a simple functional language.
We assume countably infinite sets of \emph{type identifiers} ($\id{tid}$),
\emph{value identifiers} ($\id{vid}$), and \emph{module identifiers}
($\id{mid}$). For each of the above identifier sets $X$, we define the
associated set of \emph{long identifiers} \textrm{Long$X$}, inductively with
$X$ as the base set and \id{mid}.\id{long}$x$ as the inductive case with
\id{long}$x$ $\in$ Long$X$ and \id{mid} being a module identifier. For the
module language, we also assume a denumerably infinite set of \emph{module type
  identifiers} ($\id{mtid}$). Long identifiers, such as $x.y.z$, allow users to
use traditional dot-notation for accessing components deep within modules and
the separation of identifier classes makes it clear in what syntactic category
an identifier belongs.

The simple core language is defined by notions of \emph{type
  expressions} ($\id{ty}$), \emph{core language expressions}
($\id{exp}$), and \emph{core language declarations} ($\id{dec}$):

\[
\begin{array}{lcl}
  \id{ty} & ::= & \id{longtid} \SEP \id{ty}_1 \rar \id{ty}_2\\
  \id{exp} & ::= & \id{longvid} \SEP \lambda \id{vid} \rar \id{exp} \SEP \id{exp}_1~\id{exp}_2 \SEP \id{exp}~\kw{:}~\id{ty} \\
  \id{dec} & ::= & \kw{val} ~ \id{vid}~ \kw{=}~ \id{exp}
\end{array}
\]

\begin{figure}
\begin{minipage}{0.49\columnwidth}
\[
\begin{array}{lcl}
  \id{mty}  & ::= & \cbra{~\id{spec}~} \\
            & | & \id{mtid} \\
            & | & \id{mid}\,\kw{:}\,\id{mty}_1 \rar \id{mty}_2  \\
            & | & \id{mty}~\kw{with}~\id{longtid}~\kw{=}~\id{ty} \\
  \id{spec} & ::= & \kw{val}~\id{vid}~\kw{:}~\id{ty} \\
            & | & \kw{type}~\id{tid} \\
            & | & \kw{module}~\id{mid}~\kw{:}~\id{mty} \\
            & | & \kw{include}~\id{mty} \\
            & | & \id{spec}_1~\id{spec}_2 \SEP \eps
\end{array}
\]
\end{minipage}
\begin{minipage}{0.49\columnwidth}
\[
\begin{array}{lcl}
  \id{mexp} & ::= & \cbra{~\id{mdec}~} \\
            & | & \id{mid} \SEP \id{mexp}\,\kw{.}\,\id{mid} \\
            & | & \kw{funct}~\id{mid}\,\kw{:}\,\id{mty}~\RAR~\id{mexp} \\
            & | & \id{longmid}\,\para{\id{mexp}} \\
  \id{mdec} & ::= & \id{dec} \\
            & | & \kw{type}~ \id{tid}~\kw{=}~ \id{ty} \\
            & | & \kw{module}~\id{mid}~ \kw{=} ~ \id{mexp} \\
            & | & \kw{module type}~\id{mtid}~\kw{=}~\id{mty} \\
            & | & \kw{open}~\id{mexp} \\
            & | & \id{mdec}_1~\id{mdec}_2 \SEP \eps
\end{array}
\]
\end{minipage}
\caption{Grammar for the module language excluding derived forms.}
\label{grammar.fig}
\end{figure}

The core language can be understood entirely in isolation
from the module language except that long identifiers may be used to
access values and types in modules.

The grammar for the module language is given in
Figure~\ref{grammar.fig}.
The module language is separated into a language for specifying
module types ($\id{mty}$) and a language for declaring modules
($\id{mdec}$). The language for module types is a two-level language
with sub-languages for specifying module components and for expressing
module types. Similarly, the language for declaring (i.e., defining)
modules is a two-level language for declaring module components and
for expressing module manipulations.
At the very toplevel, a program is simply a module declaration,
possibly consisting of a sequence of module declarations where later
declarations may depend on earlier declarations.

As will become apparent from the typing rules, in declarations of the
form $\id{mdec}_1~\id{mdec}_2$, identifiers declared by $\id{mdec}_1$
are considered bound in $\id{mdec}_2$ (similar considerations hold for
composing specifications and programs).

\subsection{Semantic objects}\label{subsec:modules:sem-obj}

For the static semantics, we assume a countably infinite set $\TSet$ of
\emph{type variables} ($\id{t}$).
A \emph{semantic type} (or simply a type), ranged over by $\tau$, takes
the form:
\[
\begin{array}{lcl}
  \tau & ::= & t \SEP \tau_1 \rar \tau_2
\end{array}
\]

\noindent
Types relate straightforwardly to syntactic types with
the difference that syntactic types contain type identifiers and
semantic types contain type variables. This difference is essential
in that it enables the support for type parameterisation and type
abstraction.

At the core level, a \emph{value environment} ($\VE$) maps value
identifiers ($\id{vid}$) to types and a \emph{type environment}
($\TE$) maps type identifiers ($\id{tid}$) to types.

\begin{figure}\[
  \begin{array}{rcl}
    E = (\TE,\VE,\ME,G) & \in & \Env = \TEnv \times \VEnv \times \MEnv \times \MTEnv \\
    \ME & \in & \MEnv = \Mid \finmap \Mod \\
    M & \in & \Mod = \Env \cup \FunSig \\
    F = \forall T.(E,\Sigma) & \in & \FunSig = \Fin(\TSet) \times \Env \times \MTy \\
    \Sigma = \exists T.M & \in & \MTy = \Fin(\TSet) \times \Mod \\
    G & \in & \MTEnv = \MTid \finmap \MTy
  \end{array}
  \] \caption{Module language semantic objects. Parameterised module types ($F$) and module types ($\Sigma$) are parameterised over finite sets of type variables (written $\Fin(\TSet)$), ranged over by $T$.}
  \label{fig:modules:semobjects}
\end{figure}

The module language semantic objects are shown in
Figure~\ref{fig:modules:semobjects}. The semantic objects constitute a number of
mutually dependent inductive definitions.
An \emph{environment} ($E$) is a quadruple $(\TE,\VE,\ME,G)$ of a type
environment $\TE$, a variable environment $\VE$, a \emph{module environment}
($\ME$), which maps module identifiers to modules, and a \emph{module type
  environment} ($G$), which maps module type identifiers to module types. A
\emph{module} is either an environment $E$, representing a non-parameterised
module, or a \emph{parameterised module type} $F$, which is an object $\forall
T.(E,\Sigma)$, for which the type variables in $T$ are considered bound.
A \emph{module type} ($\Sigma$) is a pair, written $\exists T.M$, of a
set of type variables $T$ and a module $M$. In a module type
$\exists T.M$, type variables in $T$ are considered bound and we
consider module types identical up-to renaming of bound variables and
removal of type variables that do not appear in $M$. When $T$ is
empty, we often write $M$ instead of $\exists \emptyset . M$.
We consider module function types $\forall T.(E,\Sigma)$ identical up-to
renaming of bound type variables and removal of type variables in $T$ that do
not occur free in $(E,\Sigma)$.

When $X$ is some tuple and when $x$ is some identifier, we shall often
write $X(x)$ for the result of looking up $x$ in the appropriate
projected finite map in $X$. Moreover, when long$x$ is some long
identifier, we write $X(\textrm{long}x)$ to denote the lookup in $X$,
possibly inductively through module environments.

\begin{defn}\label{def:modules:env-modification}
When $X$ and $Y$ are finite maps, the \emph{modification} of $X$ by
$Y$, written $X+Y$, is the map with $\Dom(X+Y) = \Dom~X \cup \Dom~Y$
and values
\[
(X + Y)(x) = \left \{ \begin{array}{ll} Y(x) & \textrm{if}~x \in \Dom~Y \\ X(x) & \textrm{otherwise} \end{array} \right .
\]
\end{defn}
The notion of modification is extended pointwise to tuples, as are operations
such as $\Dom$, $\cap$, and $\cup$.

\begin{defn}\label{def:modules:env-ext}
A finite map $X$ \emph{extends} another finite map $X'$, written $X \sqsupseteq
X'$, if $\Dom~X \supseteq \Dom~X'$ and $X(x) = X'(x)$ for all $x \in \Dom~X'$.
\end{defn}

Given a particular kind of environment, such as a module environment $\ME$, we
shall often be implicit about its injection $(\{\},\{\},\ME,\{\})$ into
environments of type $\Env$. Moreover, given an identifier, such as \id{tid},
its class specifies exactly that, given some type $\tau$,
$\{\id{tid}\mapsto\tau\}$ denotes a type environment of type $\TE$, which again,
by the above convention, can be injected implicitly into an environment of type
$\Env$.

As an example, if \texttt{t} is a type identifier, \texttt{a} and \texttt{b} are
value identifiers, and \texttt{A} is a module identifier, we can write
$\{\texttt{t}\mapsto t\} + \{\texttt{A}\mapsto \{\texttt{a}\mapsto t\}\}$ for
specifying the environment $E = (\{\texttt{t}\mapsto t\},\{\},
\{\texttt{A}\mapsto E'\},\{\})$, where $E' = (\{\},\{\texttt{a}\mapsto t\},
\{\},\{\})$ and where $E \sqsupseteq \{\texttt{t}\mapsto t\}$. Moreover, looking
up the long identifier \texttt{A.a} in $E$, written $E(\texttt{A.a})$, yields
$t$.

\subsection{Elaboration}
The elaboration rules for the core language (Figure \ref{fig:modules:elab-core})
illustrate the interaction between the module language and the core language
through the concept of long identifiers.

\begin{figure}

  \titlesembox{Type Expressions}{E \vd \id{ty} : \tau}

  \twopart{
    \fraccn{E(\id{longtid}) = \tau}
           {E \vd \id{longtid} : \tau}
  }{
    \fraccn{E \vd \id{ty}_i : \tau_i  \SP i = [1,2]}
           {E \vd \id{ty}_1 \rar \id{ty}_2 : \tau_1 \rar \tau_2}
  }

  \titlesembox{Core language expressions}{E \vd \id{exp} : \tau}

  \twopart{
    \fraccn{E(\id{longvid}) = \tau}
           {E \vd \id{longvid} : \tau}
  }{
    \fraccn{E \vd \id{exp} : \tau \LSP E \vd \id{ty} : \tau}
           {E \vd \id{exp}~\kw{:}~\id{ty} : \tau}
  }

  \twopart{
    \fraccn{E + \{\id{vid}\mapsto\tau\} \vd \id{exp} : \tau'}
           {E \vd \lambda \id{vid} \rar \id{exp} : \tau \rar \tau'}
  }{
    \fraccn{E \vd \id{exp}_1 : \tau \rar \tau' \LSP E \vd \id{exp}_2 : \tau}
           {E \vd \id{exp}_1~\id{exp}_2 : \tau'}
  }
  \caption{Elaboration rules for the core language.}
  \label{fig:modules:elab-core}
\end{figure}

\begin{figure}
  \titlesembox{Module Types}{E \vd \id{mty} : \Sigma}

  \twopart{
    \fraccn{E(\id{mtid}) = \Sigma}
           {E \vd \id{mtid} : \Sigma}
  }{
    \fraccn{\label{elab.with.rule}E \vd \id{ty} : \tau \LSP E'(\id{longtid}) = t   \LSP   t \in T \\
      E \vd \id{mty} : \exists T.E' \LSP \Sigma = \exists(T \setminus \{t\}).(E'[\tau/t])}
           {E \vd \id{mty} ~\kw{with}~\id{longtid}~\kw{=}~\id{ty} : \Sigma}
  }

  \twopart{
    \fraccn{E \vd \id{spec} : \Sigma}
           {E \vd \cbra{~\id{spec}~} : \Sigma}
  }{
    \fraccn{\label{elab.mty.funct.rule}
      E \vd \id{mty}_1 : \exists T.E' \LSP T \cap (\textrm{tvs}(E) \cup T') = \emptyset \\
      E + \{\id{mid}\mapsto E'\} \vd \id{mty}_2 : \exists T'.M}
           {E \vd \id{mid}~\kw{:}~\id{mty}_1 \rar~\id{mty}_2 : \forall T.(E',\exists T'.M)}
  }

  \titlesembox{Module Specifications}{E \vd \id{spec} : \exists T.E'}

  \twopart{
    \fraccn{}
           {E \vd \kw{type}~\id{tid} : \exists \{\id{t}\} .\{\id{tid} \mapsto \id{t}\}}
  }{
    \fraccn{E \vd \id{ty} : \tau}
          {E \vd \kw{val}~\id{vid}~\kw{:}~\id{ty} : \{\id{vid} \mapsto \tau\}}
  }

  \twopart{
    \fraccn{\label{elab.spec.module.rule}
          E \vd \id{mty} : \exists T.M}
          {E \vd \kw{module}~\id{mid}~\kw{:}~\id{mty} : \exists T .\{\id{mid} \mapsto M\}}
  }{
    \fraccn{E \vd \id{mty} : \exists T.E'}
          {E \vd \kw{include}~\id{mty} : \exists T.E'}
  }

  \twopart{
    \fraccn{\label{elab.spec.seq.rule}
            E \vd \id{spec}_1 : \exists T_1.E_1 \LSP E + E_1 \vd \id{spec}_2 : \exists T_2.E_2 \\
            T_1 \cap (\textrm{tvs}(E) \cup T_2) = \emptyset \LSP \Dom~E_1 \cap \Dom~E_2 = \emptyset}
          {E \vd \id{spec}_1 ~\id{spec}_2 : \exists(T_1 \cup T_2).(E_1 + E_2)}
  }{
    \fraccn{}
           {E \vd \eps : \{\}}
  }

  \caption{Elaboration rules for module types and module
    specifications. This sub-language does not directly depend
    on the rules for module expressions and module
    declaration.}
  \label{fig:modules:elab-mty}
\end{figure}

Elaboration of module types and specifications is defined as a mutual
inductive relation allowing inferences among sentences of the forms $E
\vdash \id{mty}:\Sigma$ and $E \vdash \id{spec}:\exists T.E'$. The rules are
presented in Figure~\ref{fig:modules:elab-mty}.
There is a subtle difference between module type expressions
($\id{mty}$) and specifications ($\id{spec}$). Whereas module type
expressions may elaborate to parameterised module types,
specifications only elaborate to non-parameterised module types, which
may, however, contain parameterised modules inside. Thus,
in the specification rule for including module types, we require that the included
module is a non-parameterised module type.

An essential aspects of the semantic technique is that of requiring, for
instance, that the sets $T_1$ and $T_2$ in Rule \Ref{elab.spec.seq.rule} are
disjoint. This property can always be satisfied by $\alpha$-renaming, which is
applied often (and in the Coq implementation, explicitly) when proving
properties of the language (see Example \ref{ex:modules:ftor-app-alpha-equiv}
and Remark \ref{rem:modules:alpha-equiv} in Section \ref{subsec:modules:nominal-in-coq}).

\begin{figure}
\titlesembox{Module Expressions}{E \vd \id{mexp} : \Sigma}

\twopart{
\fraccn{E \vd \id{mdec} : \Sigma}{E \vd \cbra{~\id{mdec}~} : \Sigma}
}{
\fraccn{E \vd \id{mexp} : \exists T.E' \\ E'(\id{mid}) = E''}
       {E \vd \id{mexp}\,\kw{.}\,\id{mid} : \exists T.E''}
}

\twopart{
\fraccn{\label{elab.mexp.funct.rule}
  E \vd \id{mty} : \exists T.E' \\ E + \{\id{mid} \mapsto E'\} \vd \id{mexp} : \Sigma \\
  T \cap \textrm{tvs}(E) = \emptyset \LSP F = \forall T. (E',\Sigma)}
      {E \vd \kw{funct}~\id{mid}~\kw{:}~\id{mty}~\RAR~\id{mexp} : \exists \emptyset.F}
}{
\fraccn{E(\id{mid}) = E'}{E \vd \id{mid} : E'}
}

\onepart
{
  \fraccn{\label{elab.mexp.app.rule}
    E \vd \id{mexp} : \exists T.E' \LSP T \cap T' = \emptyset \\
    E(\id{longmid}) \geq (E'',\exists T'.E''') \\ E' \succ E'' \LSP (T\cup T') \cap \textrm{tvs}(E) = \emptyset}
         {E \vd \id{longmid}\,\para{~\id{mexp}~} : \exists (T \cup T').E'''}
}

\titlesembox{Module Declarations}{E \vd \id{mdec} : \exists T.E'}

\twopart{
  \fraccn{\label{elab.mdec.dec.rule}\id{mdec} = \id{dec} \LSP E \vd \id{dec} : E'}{E \vd \id{mdec} : E'}
}{
  \fraccn{E \vd \id{ty} :\tau}
      {E \vd \kw{type}~\id{tid}~\kw{=}~\id{ty} : \{\id{tid} \mapsto \tau\}}
}

\onepart{
  \fraccn{\label{elab.mdec.module.rule}
         E \vd \id{mexp} : \exists T.M}
         {E \vd \kw{module}~\id{mid}~\kw{=}~\id{mexp} : \exists T.\{\id{mid} \mapsto M\}}
}

\onepart{
  \fraccn{E \vd \id{mexp} : \Sigma}
         {E \vd \kw{open}~\id{mexp} : \Sigma}
}

\twopart{
  \fraccn{\label{elab.mexp.moduletype.rule}
         E \vd \id{mty} : \Sigma}
         {E \vd \kw{module}~\kw{type}~\id{mtid}~\kw{=}~\id{mty} : \exists \emptyset.\{\id{mtid} \mapsto \Sigma\}}
}{
  \fraccn{}
         {E \vd \eps : \{\}}
}

\onepart{\label{rule:modules:elab-mdec-seq}
  \fraccn{T_1 \cap (\textrm{tvs}(E) \cup T_2) = \emptyset \\
          E \vd \id{mdec}_1 : \exists T_1.E_1 \LSP E + E_1 \vd \id{mdec}_2 : \exists T_2.E_2}
         {E \vd \id{mdec}_1 ~\id{mdec}_2 : \exists(T_1 \cup T_2).(E_1 + E_2)}
}

\caption{Elaboration rules for module language expressions and declarations.}
\label{fig:module:elab-mdec}
\end{figure}

The elaboration rules for module language expressions and declarations
are given in Figure~\ref{fig:module:elab-mdec} and allow inferences among
sentences of the forms $E \vdash \id{mdec} : \Sigma$ and $E \vdash \id{mexp}
: \exists T.E$. The rules make use of the previously introduced rules
for module type expressions and core language declarations and
types. Similarly to the elaboration difference between module type
expressions and specifications, module expressions may elaborate to
general module types, of the form $\exists T.M$, whereas module
declarations elaborate to non-parameterised module types of the form
$\exists T.E$.

The by far most complicated rule is Rule~\Ref{elab.mexp.app.rule},
the rule for application of a parameterised module. The rule looks up a
parameterised module type $\forall T_0.(E_0,\Sigma_0)$ for the long module
identifier in the environment and seeks to match the parameter module type
$\exists T_0.E_0$ against a cut-down version (according to the enrichment
relation) of the module type resulting from elaborating the argument module
expression. The result of elaborating the application is the result module type
perhaps with additional abstract type variables stemming from elaborating the
argument module expression. The need for also quantifying over the type set $T$ in
the result module type comes from the desire to prove a property that if $E
\vdash \id{mexp} : \exists T.E'$ then ${\textrm{tvs}}(E') \subseteq
       {\textrm{tvs}}(E) \cup T$.

\subsection{Enrichment}
Next, we introduce a notion of \emph{enrichment}. Intuitively, the enrichment
relation for semantic objects is a generalised version of environment extension
(Definition \ref{def:modules:env-ext}). We write $E' \succ E$ for $E'$ enriches
$E$ meaning that $E'$ contains the same elements as or more elements then $E$.
The formal specification of the enrichment relation is given in Figure
\ref{fig:modules:enrich}.

Notice that enrichment for parameterised modules is contravariant in parameter
environments. Notice also the special treatment of module type
environments. Because a module type cannot specify bindings of module types, we
can safely require that when $E' \succ E$, the module type environment in $E$ is
empty.
\begin{figure}
  \begin{equation*}
    \begin{gathered}
    \inferrule{E' = (\TE',\VE',\ME',G') \\ E = (\TE,\VE,\ME,\{\}) \\\\
               \VE' \sqsupseteq \VE \\ \TE' \sqsupseteq \TE  \\ \ME' \succ \ME)}
              {E' \succ E}\\    \inferrule{\Dom~\ME' \supseteq \Dom~\ME \\
                \forall \id{mid} \in \Dom~\ME.~ \ME'(\id{mid}) \succ \ME(\id{mid})}
              {\ME' \succ \ME} \\
   \inferrule{M' = E' \\ M = E \\ E' \succ E \\ }
            {M'\succ M}\\
  \inferrule{M' = \forall T'.(E',\Sigma') \\ M = \forall T.(E,\Sigma) \\\\
             T'=T \\ E \succ E' \\ \Sigma' \succ \Sigma}
            {M' \succ M} \\
  \inferrule{M' \succ M \\ \Sigma' = \exists T.M' \\ \Sigma = \exists T.M}
            {\Sigma' \succ \Sigma}
    \end{gathered}
  \end{equation*}
  \caption{The enrichment relation.}
  \label{fig:modules:enrich}
\end{figure}

\subsection{Target Language}

We assume a denumerably infinite set $\LSet$ of \emph{labels}, ranged over by
$l$. Target expressions are basically identical to core
level expressions with the modification that value identifiers are
replaced with labels. For the simple core language that we are
considering, \emph{target expressions} ($\id{ex}$) and \emph{target
  code} ($c$) take the form:
\[
\begin{array}{lcl}
  \id{ex} & ::= & l \SEP \lambda l \rar \id{ex} \SEP \id{ex}_1~\id{ex}_2 \\
  c & ::= & \kw{val}~l~\kw{=}~\id{ex} \SEP c_1~\kw{;}~c_2 \SEP \eps
\end{array}
\]

The type system for the target language is simple (for the purpose of
this work) and allows inferences among sentences of the forms $\Gamma
\vd \id{ex} : \tau$ and $\Gamma \vd c : \Gamma'$, which are read: ``In
the context $\Gamma$, the expression $\id{ex}$ has type $\tau$'' and
``in the context $\Gamma$, the target code $c$ declares the context
$\Gamma'$. Contexts $\Gamma$ map labels to types. The type system for
the target language is presented in Figure~\ref{target.fig}.

\begin{figure}
  \twopart{
    \fraccn{\Gamma \vd c_1 : \Gamma_1 \LSP \Gamma + \Gamma_1 \vd c_2 : \Gamma_2}
           {\Gamma \vd c_1 ~\kw{;}~c_2 : \Gamma_1 + \Gamma_2}
  } {
    \fraccn{}{\Gamma \vd \eps : \{\}}
  }

  \onepart
  {
    \fraccn{\Gamma \vd \id{ex} : \tau }{\Gamma \vd \kw{val}~l~\kw{=}~\id{ex} : \{l \mapsto \tau\}}
  }
  \caption{Type rules for the target language. For the purpose of
    the presentation, the target language is simple and mimics closely
    the source language with the difference that long identifiers are
    replaced with labels for referring to previously defined value
    declarations.}
  \label{target.fig}
\end{figure}

\subsection{Interpretation Objects}

In the following, we shall use the term \emph{name} to refer to either
a type variable $t$ or a label $l$. We write $\NSet$ to refer to the
disjoint union of $\TSet$ and $\LSet$. Moreover, we use $N$ to range
over finite subsets of $\NSet$.

An \emph{interpretation value environment} ($\VEc$) maps
value identifiers to a label and an associated type.
An \emph{interpretation environment} ($\Ec$) is a quadruple
$(\TE,\VEc,\MEc,G)$ of a type environment, an interpretation value
environment, an interpretation module environment, and a module type
environment. An \emph{interpretation module environment} ($\MEc$) maps
module identifiers to module interpretations. A \emph{module
  interpretation} ($\Mc$) is either an interpretation environment
$\Ec$ or a functor closure $\Phi$. A \emph{functor closure} ($\Phi$)
is a triple $(\Ec,F,\lambda\id{mid}\Rar\id{mexp})$ of an
interpretation environment, a parameterised module type, and a
representation of a parameterised module expression. Finally, an
\emph{interpretation target object} ($\exists N.(\Ec,c)$) is a triple
of a name set, an interpretation environment, and a target code
object.

\subsection{Interpretation Erasure}

For establishing a link between interpretation objects and elaboration
objects, we introduce the concept of interpretation erasure.  Given an
interpretation object $O$, we define the \emph{interpretation erasure}
of $O$, written $\overline{O}$, as follows:
\[
\begin{array}{rcl}
  \overline{(\TE,\VEc,\MEc,G)} & = & (\TE,\overline{\VEc},\overline{\MEc},G) \\
  \overline{(\Ec,F,\lambda\id{mid}\Rar\id{mexp})} & = & F \\
  \overline{\{\id{vid}_i\mapsto l_i:\tau_i\}^n} & = & \{\id{vid}_i\mapsto\tau_i\}^n
\end{array}
\]
\[
\begin{array}{rcl}
  \overline{\{\id{mid}_i\mapsto \Mc_i\}^n} & = & \{\id{mid}_i\mapsto\overline{\Mc_i}\}^n \\
  \overline{\exists N.(\Ec,c)} & = & \exists (\TSet \cap N).\overline{\Ec}
\end{array}
\]

\subsection{Core Language Compilation}

Core language expressions and declarations are compiled into target
language expressions and declarations, respectively. The rules
specifying the compilation allow inferences among sentences of the
forms (1) $\Ec \vdash \id{exp} \Rar \id{ex},\tau$ and (2) $\Ec \vdash
\id{dec} \Rar \exists N.(\Ec',c)$. The rules are given in
Figure~\ref{core.comp.fig}.

\begin{figure}
\titlesembox{Compiling Expressions}{\Ec \vd \id{exp} \Rar \id{ex},\tau}

\twopart{
  \fraccn{\Ec(\id{longvid}) = (l,\tau)}
         {\Ec \vd \id{longvid} \Rar l,\tau}
}{
  \fraccn{\Ec + \{\id{vid} \mapsto (l,\tau)\} \vd \id{exp} \Rar \id{ex}, \tau'}
         {\Ec \vd \lambda \id{vid} \rar \id{exp} \Rar \lambda l \rar \id{ex}, \tau \rar \tau'}
}

\onepart{
  \fraccn{\Ec \vd \id{exp}_1 \Rar \id{ex}_1, \tau\rar\tau' \LSP \Ec \vd \id{exp}_2 \Rar \id{ex}_2, \tau}
         {\Ec \vd \id{exp}_1 \id{exp}_2 \Rar \id{ex}_1~\id{ex}_2, \tau'}
}

\onepart{
  \fraccn{\Ec \vd \id{exp} \Rar \id{ex}, \tau \LSP \overline{\Ec} \vd \id{ty} : \tau}
         {\Ec \vd \id{exp} ~\kw{:}~ \id{ty} \Rar \id{ex}, \tau}
}

\titlesembox{Compiling Declarations}{\Ec \vd \id{dec} \Rar \exists N.(\Ec',c)}

\onepart{\label{rule:modules:comp-decl}
  \fraccn{\Ec \vd \id{exp} \Rar \id{ex},\tau \LSP l \not \in \textrm{names}(\Ec)}
         {\Ec \vd \kw{val}~\id{vid}~\kw{=}~\id{exp} \Rar
           \exists \{l\}.(\{\id{vid}\mapsto(l,\tau)\},\kw{val}~l~\kw{=}~\id{ex})}
}
\caption{Core language compilation.}
\label{core.comp.fig}
\end{figure}

The rules track type information and it is
straightforward to establish the following property of the
compilation:

\begin{lemma}
  If $\Ec \vd \id{dec} \Rar \exists N .(\Ec',c)$ then $\overline{\Ec}
  \vd \id{dec} : \overline{\exists N.(\Ec',c)}$.
\end{lemma}

\subsection{Environment Filtering}

Corresponding to the notion of enrichment for elaboration, we
introduce a notion of filtering for the purpose of static
interpretation, which filters interpretation environments to contain
components as specified by an elaboration environment. Filtering is
essential to the interpretation rule for applications of parameterised
modules and is defined mutual inductively based on the structure of
elaboration environments and elaboration module environments.

More formally, the \emph{filtering} of an interpretation environment $\Ec$ to an
elaboration environment $E$ results in another interpretation environment $\Ec'$
with only elements from $\Ec$ that are also present in $E$. The filtering
relation is defined by a number of inference rules that allow inferences among
sentences of the forms (1) $\vd \Ec :: E \Rar \Ec'$, (2) $\vd \VEc :: \VE \Rar
\VEc'$, (3) $\vd \MEc :: \ME \Rar \MEc'$, and (4) $\vd \Mc :: M \Rar \Mc'$. The
inference rules for filtering are presented in Figure~\ref{filtering.fig}.  It
is a straightforward to establish the connection between filtering and enrichment.
\begin{lemma}[Filtering to enrichment]\label{lem:modules:filt-enriches}
  It is a straightforward exercise to demonstrate that if
  $\vd \Ec :: E \Rar \Ec'$ then it holds that $\overline{\Ec} \succ E$.
\end{lemma}
\begin{figure}

  \titlesembox{Environments}{\vd \Ec :: E \Rar \Ec'}

  \onepart{
    \fraccn{\vd \VEc :: \VE \Rar \VEc' \\
      \vd \MEc::\ME \Rar \MEc' \LSP \TE \succ \TE'}
           {\vd (\TE,\VEc,\MEc,\{\}) :: (\TE',\VE,\ME,\{\}) \Rar (\TE',\VEc',\MEc',\{\})}
  }

  \titlesembox{Value Environments}{\vd \VEc :: \VE \Rar \VEc'}

  \onepart{
    \fraccn{m \geq n}
           {\vd \{\id{vid}_i\mapsto l_i:\tau_i\}^m :: \{\id{vid}_i\mapsto\tau_i\}^n \Rar \{\id{vid}_i\mapsto l_i:\tau_i\}^n}
  }

  \titlesembox{Module Environments}{\vd \MEc :: \ME \Rar \MEc'}

  \onepart{
    \fraccn{m \geq n \LSP \vd \Mc_i :: M_i \Rar \Mc_i' \SP i=1..n}
          {\vd \{\id{mid}_i\mapsto \Mc_i\}^m :: \{\id{mid}_i\mapsto M_i\}^n \Rar \{\id{mid}_i\mapsto \Mc_i'\}^n}
  }

  \titlesembox{Module Interpretations}{\vd \Mc :: M \Rar \Mc'}

  \[
    \fraccn{\Mc = \Ec \LSP \vd \Ec :: E \Rar \Ec'}{\vd \Mc :: E \Rar \Ec'}
  \]
  \[
    \fraccn{\Phi = (\Ec,F',\lambda \id{mid} \Rar \id{mexp}) \LSP F' \succ F}
           {\vd \Phi :: F \Rar (\Ec,F,\lambda \id{mid} \Rar \id{mexp})}
  \]
  \caption{Filtering relation specifying how an interpretation environment can be constrained by an elaboration environment to form a restricted interpretation environment.}
  \label{filtering.fig}
\end{figure}

\subsection{Static Interpretation Rules}

Static interpretation of the module language is defined by a number
of mutually inductive inference rules allowing inferences among
sentences of the forms (1) $\Ec \vd \id{mexp} \Rar \Psi$ and (2) $\Ec
\vd \id{mdec} \Rar \exists N.(\Ec',c)$, which state that in an
interpretation environment $\Ec$, static interpretation of a module
expression \id{mexp} results in an interpretation target object
$\Psi$, and static interpretation of a module declaration \id{mdec}
results in an interpretation target object $\exists N.(\Ec',c)$. The
rules for static interpretation are presented in
Figure~\ref{interp.fig}.

\begin{figure}

  \titlesembox{Module Expressions}{\Ec \vd \id{mexp} \Rar \Psi}

  \twopart{
    \fraccn{\Ec \vd \id{mdec} \Rar \exists N.(\Ec',c)}
           {\Ec \vd \cbra{~\id{mdec}~} \Rar \exists N.(\Ec',c)}\label{rule:modules-sint-mdec}
  }{
    \fraccn{\Ec(\id{mid}) = \Ec'}
           {\Ec \vd \id{mid} \Rar \exists \emptyset .(\Ec',\eps)}\label{rule:modules-sint-mid}
  }

  \onepart{
    \fraccn{\Ec'(\id{mid}) = \Ec'' \\ \Ec \vd \id{mexp} \Rar \exists N.(\Ec',c)}
           {\Ec \vd \id{mexp}\,\kw{.}\,\id{mid} \Rar \exists N .(\Ec'',c)}
           \label{rule:modules-sint-proj}
  }

  \onepart{
    \fraccn{\overline{\Ec} \vd \id{mty} : \exists T.E \LSP T \cap \textrm{names}(\Ec) = \emptyset \\
      \overline{\Ec} + \{\id{mid} \mapsto E\} \vd \id{mexp} : \Sigma \LSP F = \forall T. (E,\Sigma) \\
      \Phi = (\Ec,F,\lambda \id{mid} \Rar\id{mexp})}
           {\Ec \vd \kw{funct}~\id{mid}~\kw{:}~\id{mty}~\RAR~\id{mexp} \Rar \exists \emptyset .\Phi}
           \label{rule:modules-sint-funct}
  }

  \onepart{
    \fraccn{\Ec \vd \id{mexp} \Rar \exists N.(\Ec',c) \LSP (N\cup N') \cap \textrm{names}(\Ec) = \emptyset \LSP N \cap N' = \emptyset \\
      \Ec(\id{longmid}) = (\Ec_0,F,\lambda \id{mid}\Rar\id{mexp}') \LSP
      F \geq (E',\exists T'.E'') \LSP T' \subseteq N' \\
      \vd \Ec' :: E' \Rar \Ec'' \LSP
      \Ec_0 + \{\id{mid}\mapsto\Ec''\}\vd \id{mexp}' \Rar \exists N'.(\Ec''',c')}
           {\Ec \vd \id{longmid}\para{~\id{mexp}~} \Rar \exists (N \cup N').(\Ec''',c~\kw{;}~c')}
           \label{rule:modules-sint-funct-app}
  }

  \titlesembox{Module declarations}{\Ec \vd \id{mdec} \Rar \exists N.(\Ec',c)}

  \onepart{
    \fraccn{\id{mdec} = \id{dec} \LSP \Ec \vd \id{dec} \Rar \exists N.(\Ec',c)}
           {\Ec \vd \id{mdec} \Rar \exists N .(\Ec',c)}
  }\label{rule:modules-sint-mdec}

  \onepart{
    \fraccn{\overline{\Ec}\vd \id{ty} :\tau}
           {\Ec \vd \kw{type}~\id{tid}~\kw{=}~\id{ty} \Rar \exists \emptyset . (\{\id{tid} \mapsto \tau\},\eps)}\label{rule:modules-sint-type}
  }

  \onepart{
    \fraccn{\Ec \vd \id{mexp} \Rar \exists N.(\Phi,c)}
          {\Ec \vd \kw{module}~\id{mid} ~\kw{=}~\id{mexp} \Rar \exists N.(\{\id{mid} \mapsto \Phi\},c)}
  }\label{rule:modules-sint-mod}

  \onepart{
    \fraccn{\overline{\Ec} \vd \id{mty} : \Sigma \LSP \Ec' = \{\id{mtid} \mapsto \Sigma\}}
           {\Ec \vd \kw{module}~\kw{type}~\id{mtid}~\kw{=}~\id{mty} \Rar \exists \emptyset .(\Ec',\eps)}\label{rule:modules-sint-mod-type}
  }

  \onepart{
    \fraccn{\Ec \vd \id{mexp} \Rar \Psi}
           {\Ec \vd \kw{open}~\id{mexp} \Rar \Psi}
  }\label{rule:modules-sint-open}

  \onepart{
    \fraccn{\Ec \vd \id{mdec}_1 \Rar \exists N_1.(\Ec_1,c_1) \\
      \Ec + \Ec_1 \vd \id{mdec}_2 \Rar \exists N_2.(\Ec_2,c_2) \\
      N_1 \cap (\textrm{names}(\Ec) \cup N_2) = \emptyset}
           {\Ec \vd \id{mdec}_1 ~\id{mdec}_2 \Rar \exists(N_1 \cup N_2).(\Ec_1 + \Ec_2, c_1~\kw{;}~c_2)}
  }\label{rule:modules-sint-seq}

  \onepart{
    \fraccn{}
           {\Ec \vd \eps \Rar \exists \emptyset.(\{\},\eps)}
  }\label{rule:modules-sint-empty}

  \caption{Static interpretation rules for module expressions and module declarations.}
  \label{interp.fig}
\end{figure}

\subsection{Static Interpretation Normalisation}\label{subset:modules:static-int-norm}
The proof technique we use to prove the static interpretation normalisation
property is similar to the one we showed in Section
\ref{sec:modules:stlc-norm}.  That is, we first define an appropriate logical
relation (Figure \ref{fig:modules:consistency}), which we call the
\emph{consistency relation}, and based on this relation, we can state a
property establishing that static interpretation is possible and
terminates for all elaborating programs. We consider only the termination
property of static interpretation, since in the current Coq formalisation we have
implemented it in this form. In general, one would like to add the type
soundness property, which states that target programs are appropriately typed. We leave
this property for future extensions of our formalisation.

\begin{figure}
  \begin{equation*}
    \begin{gathered}
    \inferrule{\Ec = (\TE,\VEc,\MEc,G) \\
            \VE \models \VEc \\ \ME \models \MEc}
              {(\TE,\VE,\ME,G) \models \Ec}\\
    \inferrule{\Dom \ME = \Dom \MEc \\
      \forall \id{mid} \in \Dom \ME, \ME(\id{mid}) \models \MEc(\id{mid})}
      {\ME \models \MEc}\\
    \inferrule{
            \Dom \VE = \Dom \VEc \\
      \forall x \in \Dom \VE, \tau = \VE(x) \wedge (l,\tau) = \VEc(x)}
              {\VE \models \VEc}\\
    \inferrule{M = E \\ \Mc = \Ec \\ E \models \Ec}
           {M \models \Mc}\\
    \inferrule{
      M = \forall T.(E,\exists T'.M') \\
      \Mc = (\Ec,\forall T.(E,\exists T'.M'),\lambda{mid} =>\id{mexp}) \\
      (\forall \Ec',
      E \models \Ec' ~~\Longrightarrow~~
        \exists N', \Mc', c.~~
        (\Ec + \{\id{mid}\mapsto\Ec'\}) \vdash \id{mexp} \Rightarrow \exists N'.(\Mc',c)
        \wedge M' \models \Mc'}
           {M \models_S \Mc}
    \end{gathered}
  \end{equation*}
  \caption{Type consistency logical relation.}
  \label{fig:modules:consistency}
\end{figure}

Before we state the theorem and sketch its proof, let us formulate auxiliary
lemmas, related to properties of relations and operations involved in the
definition of the type consistency relation.

\begin{lemma}[Lookup consistency]\label{lem:modules:consistent-lookup}
  If looking up a module for some $mid$ in any semantic object $E$ gives a
  module $M$, and $E$ is related to some $\Ec$ by the consistency relation (Figure
  \ref{fig:modules:consistency}), then looking up for the same $mid$ in the
  interpretation environment $\Ec$ must give some module interpretation $\Mc$,
  such that $M$ is consistent with $\Mc$.

  That is, if $E \models \Ec$, and $E(mid)$, then
  $\exists \Mc.~\Ec(mid) = \Mc \wedge M \models \Mc$.
\end{lemma}

\begin{lemma}[Uniformness]\label{lem:modules:uniform-modules}
  This lemma can be seen as some an inversion principle for the consistency
  relation: if we know something about the shape of a module, we know what
  the shape of a corresponding module interpretation is.

  That is, if $M$ is a module, and $M$ is consistent with some module
  interpretation $\Mc$, then:

  \begin{enumerate}[(i)]
  \item if $M$ is a non-parameterised module, then $\Mc$ is also a
    non-parameterised module interpretation for some interpretation environment
    $\Ec$. That is, if $M = E$, and $M \models \Mc$, then $\exists \Ec. \Mc =
    \Ec$.
  \item if $M$ is a functor, then $\Mc$ is a functor closure, which is
    consistent with the module $M$. That is,
    if $M = \forall T.(E,\exists T'.M')$ and $M \models \Mc$, then
    $\Mc = (\Ec,\forall T.(E,\Sigma),\lambda mid => mexp)$, and
    the corresponding consistency condition holds (we write $\Longrightarrow$ for implication ):
    \begin{align*}
    \forall \Ec',
      E \models \Ec' ~~&\Longrightarrow~~\\
        \exists N', \Mc', c.&~~
        ((\Ec + \{\id{mid}\mapsto\Ec'\}) \vdash \id{mexp} \Rightarrow \exists N'.(\Mc',c))
        \wedge M' \models \Mc'
    \end{align*}
  \end{enumerate}
\end{lemma}

\begin{lemma}[Type consistency to erasure]\label{lem:modules:consistent-erasure}
  If some semantic object $E$ is consistent with some interpretation environment
  $\Ec$, then erasure of $\Ec$ gives us exactly $E$.

  That is, if $E \models \Ec$ then $\overline{\Ec} = E$.
\end{lemma}

\begin{lemma}[Consistency and environment extension (cf. Lemma \ref{lem:modules:stls-ctx})]
  \label{lem:modules:consistent-extend}
  If some semantic object $E$ is consistent with some interpretation environment
  $\Ec$, and some module $M$ is consistent with some module interpretation $\Mc$,
  then for some module identifier $mid$, the extension of $E$ with a mapping $mid \mapsto M$
  is consistent with $\Ec$ extended with the mapping $mid \mapsto \Mc$.

  That is, if $E \models \Ec$ and $M \models \Mc$ then
  $(E + \{mid \mapsto M\}) \models (\Ec +  \{mid \mapsto \Mc\})$
\end{lemma}

\begin{lemma}[Consistency of environment modification]\label{lem:modules:consistent-env-plus}
  For any semantic objects $E_1$, $E_2$, and interpretation environments $\Ec_1$, $\Ec_2$,
  if $E_1 \models \Ec_1$ and $E_2 \models \Ec_2$, then $(E_1 + E_2) \models (\Ec_1 + \Ec_2)$.
\end{lemma}

\begin{lemma}[Termination of the declarations compilation]\label{lem:modules:term-dec-comp}
  For any environments $E$ and $E'$, interpretation environment $\Ec$, and declaration
  $\id{dec}$, if $E \vdash \id{dec} : E'$, and $E \models \Ec$,
  then there exist $N$, $\Ec'$, $c$, s.t. $\Ec \vdash \id{dec} => \exists N.(\Ec',c)$ and
  $E' \models \Ec'$.
\end{lemma}

\begin{lemma}[Consistency of filtering]\label{lem:modules:consistent-enrich-to-filtering}
  For any environments $E'$ and $E$, interpretation environment $\Ec'$, if
  $E' \succ E$, and $E' \models \Ec'$ then there exist an interpretation
  environment $\Ec$, s.t. $\Ec' :: E => \Ec$ and $E \models \Ec$.
\end{lemma}

The idea of the proof of Lemma \ref{lem:modules:consistent-enrich-to-filtering}
is to \emph{construct} the environment $\Ec$. The definition of the filtering
relation does not define a particular algorithm for constructing a filtered
environment, but instead serves as a specification. One can define a function
that actually implements filtering and show that this function satisfies the
specification give by the definition in Figure \ref{filtering.fig}. Essentially,
we take this approach in our Coq development.

\begin{thm}(Static Interpretation Normalization)
  \label {norm.prop}~\\
  If $\overset{\mathcal D_{mexp}}{E \vdash \id{mexp} : \exists T.M}$ and
  $\overset{\mathcal C}{E \models \Ec}$
  then there exists $N$, $\Mc$, $c$ such that
  \[\Ec \vdash \id{mdec} \Rar \exists N.(\Mc,c), ~\text{and}~ M \models \Mc\]
  and (mutually)\\
  if $\overset{\mathcal D_{mdec}}{E \vdash \id{mdec} : \exists T.E'}$ and
  $\overset{\mathcal C}{E \models \Ec}$
  then there exists $N$, $\Ec'$, $c$ such that
  \[\Ec \vdash \id{mdec} \Rar \exists N.(\Ec',c), ~\text{and}~ E' \models \Ec'\]
\end{thm}
\begin{proof}
  The proof sketch presented here is following our development in the Coq
  proof assistant. We omit some details, since they worked out
  fully in the formalisation, and we aim here to show the overall structure of the proof
  and refer to auxiliary lemmas, which are crucial to use in particular cases.
  We also emphasise the simplifying assumptions we made for our development.

  The proof proceeds by mutual induction over derivations $\mathcal D_{mdec}$ and
  $\mathcal D_{mexp}$.
  \begin{case}[$\inferrule{\overset{\mathcal D_{mdec}}{E \vdash \id{mdec} : \Sigma}}{E \vdash \cbra{~\id{mdec}~} : \Sigma}$]~\\

    By induction hypothesis on $\mathcal D_{mdec}$  with $\mathcal C$,
    we get $N$, $\Ec$ $c$, and we take $\Mc = \Ec$ and apply Rule \Ref{rule:modules-sint-mdec}.
  \end{case}
  \begin{case}[$\inferrule{E(\id{mid}) = E'}{E \vdash \id{mid} : E'}$]~\\
    By lemma \ref{lem:modules:consistent-lookup}, we get $\Mc'$, s.t.  $M
    \models \Mc'$. Take $N=\emptyset$, $\Mc = \Mc'$, and $c = \eps$. We get the
    derivation for the interpretation by applying Rule \Ref{rule:modules-sint-mid}, and we
    already have $M \models \Mc'$ from having applied Lemma
    \ref{lem:modules:consistent-lookup} earlier.
  \end{case}

  \begin{case}[$\inferrule{
        \overset{\mathcal D_{mexp}}{E \vdash \id{mexp} : \exists T.E'} \\ E'(\id{mid}) = E''}
       {E \vdash \id{mexp}\,\kw{.}\,\id{mid} : \exists T.E''}$]~\\
    By induction hypothesis on $\mathcal D_{mexp}$ with $\mathcal C$, we get $N$, $\Mc$, $c$, s.t.
    $\Ec \vdash \id{mexp} \Rar \exists N.(\Mc,c)$, and $M \models \Mc$, where
    $M$ is a non-parameterised module $E'$. By Lemma
    \ref{lem:modules:uniform-modules}, we get $E' \models \Ec'$ for some $\Ec'$.

    Now, by applying the lookup consistency lemma (Lemma \ref{lem:modules:consistent-lookup}),
    we get all the required pieces to construct a derivation of static interpretation
    using Rule \Ref{rule:modules-sint-proj}.~

    We already have $M \models \Mc'$ from the induction hypothesis.
  \end{case}
  \begin{case}[
      \inferrule
          {E \vdash \id{mty} : \exists \emptyset.E' \\
            \overset{\mathcal D_{mexp}}{E + \{\id{mid} \mapsto E'\} \vdash \id{mexp} : \Sigma} \\
            F = \forall \emptyset. (E',\Sigma)}
      {E \vdash \kw{funct}~\id{mid}~\kw{:}~\id{mty}~\RAR~\id{mexp} : \exists \emptyset.F}]~\\
    Since we are in the simplified setting here, we instantiate universally and
    existentially quantified variables with empty sets in the rule.

    Take $N=\emptyset$, $\Mc = (\Ec, F, \lambda => \id{mexp})$ (a functor
    closure), $c = \eps$. From assumptions we have $\overset{\mathcal C}{E \models \Ec}$.
    From $\mathcal C$ with the consistency to erasure lemma
    (Lemma \ref{lem:modules:consistent-erasure}), we
    get $\overline{\Ec} = E$. Now, we apply Rule \Ref{rule:modules-sint-funct}.~

    Proving $M \models \Mc$ in this case requires some work. We get a required
    premise for the corresponding rule of the consistency relation from the
    induction hypothesis with the consistency and the environment extension
    lemma (Lemma \ref{lem:modules:consistent-extend})
    (cf. Theorem \ref{thm:modules:stlc-norm}, Case \nameref{rule:ty-lam}).
  \end{case}
  \begin{case}[
      \inferrule{
          \overset{\mathcal D_{mexp}}{E_0 \vdash \id{mexp} : \exists \emptyset.E} \\
             E_0(\id{longmid}) \geq (E',\exists \emptyset.E'') \\ E \succ E'}
          {E_0 \vdash \id{longmid}\,\para{~\id{mexp}~} : \exists \emptyset.E''}]~\\
    This is the most complicated case of our proof. Again, we consider a simplified setting
    where sets of variables in binding positions are empty.
    Moreover, instead of instantiation $E_0(\id{longmid}) \geq (E',\exists \emptyset.E'')$
    we use simple equality $E_0(\id{longmid}) = (E',\exists \emptyset.E'')$.

    First, we use the lookup consistency lemma (Lemma \ref{lem:modules:consistent-lookup}) to get
    $\Mc_0$, s.t. $\Ec(longid)=\Mc_0$ and $M \models \Mc_0$.

    From the uniformness lemma (Lemma \ref{lem:modules:uniform-modules}(\textit{ii})) we
    know that $\Mc_0$ is a functor closure $(\Ec_1,M,\lambda mid =>
    \id{mexp}_1)$ for some $\Ec_1$ and $\id{mexp}_1$.  From this lemma we also get
    the consistency condition; we will denote it as $\mathcal{H}_c$ and use it
    later.

    Next, from the induction hypothesis on $\mathcal{D}_{mexp}$ with
    $\mathcal{C}$, we get $N_1$, $\Mc_1$, $c_1$, s.t.
    $\Ec \vdash mexp => \exists N_1.(\Mc_1, c_1)$ and $E \models \Mc_1$.
    From the uniformness lemma \ref{lem:modules:uniform-modules}(\emph{i}) we know that
    $\Mc_1$ is also some non-parameterised module interpretation $\Ec'$, and so
    $\overset{\mathcal{H}_e}{E \models \Ec'}$.

    From the consistency of filtering lemma (Lemma
    \ref{lem:modules:consistent-enrich-to-filtering}), we get $\Ec''$, s.t.  $\Ec'
    :: E' => \Ec''$ and $\overset{\mathcal{C}'}{E' \models \Ec''}$.  We can use
    the consistency condition $\mathcal{H}_c$ with $\mathcal{C}'$ to get $N_2$,
    $\Mc_2$, $c_2$, s.t.  $(\Ec' + \{mid \mapsto \Ec''\}) \vdash mexp_1 =>
    \exists N_2.(\Mc_2,c_2)$ and $\overset{\mathcal{C}''}{E'' \models \Mc_2}$.

    Now, we take $N = N_1 \cup N_2$, $\Mc = \Mc_2$, $c = c_1 ; c_2$ in our goal.
    We construct the required derivation by applying Rule \Ref{rule:modules-sint-funct-app}.~

    The consistency part follows from $\mathcal{C}''$.
    (cf. Theorem \ref{thm:modules:stlc-norm}, Case \nameref{rule:ty-app}.
    Although in the present proof we do not have two induction hypotheses and
    instead of the first induction hypothesis we look up for the $longid$).
  \end{case}
  \begin{case}[
      \inferrule{\id{mdec} = \id{dec} \\ E \vdash \id{dec} : E'}
                {E \vdash \id{mdec} : E'}]~\\

    By the consistency to erasure lemma (Lemma \ref{lem:modules:consistent-erasure})
    with $\mathcal{C}$, we get $\overline{\Ec} = E$. By the termination of the
    declarations compilation lemma \ref{lem:modules:term-dec-comp}, we get $N$,$\Ec'$,$c$, s.t.
    \begin{equation}\label{eq:modules:interp:comp-1}
      \Ec \vdash \id{dec} => \exists N.(\Ec',c)
    \end{equation}
    \begin{equation}\label{eq:modules:interp:consist-1}
      E \models \Ec'
    \end{equation}
    We construct the required derivation by applying Rule \Ref{rule:modules-sint-mdec} with
    \reff{eq:modules:interp:comp-1}, and we have $E \models \Ec'$ from
    \reff{eq:modules:interp:consist-1}.
  \end{case}

  \begin{case}[\inferrule{E \vdash \id{ty} :\tau}
      {E \vdash \kw{type}~\id{tid}~\kw{=}~\id{ty} : \{\id{tid} \mapsto \tau\}}]~\\
    By the consistency to erasure lemma \ref{lem:modules:consistent-erasure}
    with $\mathcal{C}$, we get $\overline{\Ec} = E$. Take $N = \emptyset$,
    $\Ec' = \{tid \mapsto \tau\}$, $c=\eps$.

    \noindent We construct the required derivation by applying Rule
    \Ref{rule:modules-sint-type}.  We get $\{tid \mapsto \tau\} \models \{tid
    \mapsto \tau\}$, since the environments are equal.
  \end{case}

  \begin{case}[\inferrule{
         E \vdash \id{mexp} : \exists T.M}
         {E \vdash \kw{module}~\id{mid}~\kw{=}~\id{mexp} : \exists T.\{\id{mid} \mapsto M\}}]~\\
    By induction hypothesis on $\mathcal D_{mexp}$ with $\mathcal C$, we get $N$,
    $\Mc$, $c$, s.t.  $\Ec \vdash \id{mexp} \Rar \exists N.(\Mc,c)$ and $M
    \models \Mc$.  Take $\Ec' = \{mid \mapsto \Mc\}$. We construct the required
    derivation by applying Rule \Ref{rule:modules-sint-mod}.~

    We get $\{mid \mapsto M \} \models \{tid \mapsto \Mc\}$ by the consistency
    and environment extension lemma (Lemma \ref{lem:modules:consistent-extend}) with
    $E = \Ec = \{\}$, and the fact that $\{\} \models \{\}$.
  \end{case}
  \begin{case}[
      \inferrule{\overset{\mathcal D_{mexp}}{E \vdash \id{mexp} : \Sigma}}
                {E \vdash \kw{open}~\id{mexp} : \Sigma}]~\\
    By induction hypothesis on $\mathcal D_{mexp}$ with $\mathcal C$, we get
    required pieces to construct the derivation using Rule \Ref{rule:modules-sint-open}.
  \end{case}
  \begin{case}[\inferrule{
         E \vdash \id{mty} : \Sigma}
         {E \vdash \kw{module}~\kw{type}~\id{mtid}~\kw{=}~\id{mty} : \exists \emptyset.\{\id{mtid} \mapsto \Sigma\}}]~\\
    Take $N = \emptyset$, $\Ec' = \{mtid \mapsto \Sigma\}$, $c = \eps$. By the
    consistency to erasure lemma (Lemma \ref{lem:modules:consistent-erasure}) with
    $\mathcal{C}$, we get $\overline{\Ec} = E$.
    We construct the required derivation by applying Rule \Ref{rule:modules-sint-mod-type}.~

    We get $\{mtid \mapsto \Sigma\} \models \{mtid \mapsto \Sigma\}$, since the
    environments are equal.
  \end{case}
  \begin{case}[\inferrule{T_1 \cap (\textrm{tvs}(E) \cup T_2) = \emptyset \\\\
        \overset{\mathcal D_{mdec_1}}{E \vdash \id{mdec}_1 : \exists T_1.E_1} \\
        \overset{\mathcal D_{mdec_2}}{E + E_1 \vdash \id{mdec}_2 : \exists T_2.E_2}}
         {E \vdash \id{mdec}_1 ~\id{mdec}_2 : \exists(T_1 \cup T_2).(E_1 + E_2)}]~\\
    \noindent In this case we do not make simplifying assumptions and leave
    sets $T_1$ and $T_2$ non-empty along with corresponding sets in the static
    interpretation rule to show where we need to apply the bound variable
    convention (see Remark \ref{rem:modules:alpha-equiv}).

    By induction hypotheses on $\mathcal D_{mdec_1}$ and $\mathcal D_{mdec_2}$
    with $\mathcal C$, we get
    \begin{align*}
      N_1, \Ec_1, c_1, \text{~s.t.~} \Ec \vdash mdec_1 => \exists N_1.(\Ec_1, c_1)
      \text{~and~} E_1 \models \Ec_1 ,\\
      N_2,\Ec_2, c_2, \text{~s.t.~} \Ec \vdash mdec_2 => \exists N_2.(\Ec_2, c_2)
      \text{~and~} E_2 \models \Ec_2.
    \end{align*}

    \noindent We can always $\alpha$-rename $\exists N_1.(\Ec_1, c_1)$ and
    $\exists N_2.(\Ec_2, c_2)$ in such a way that $N_1$ and $N_2$ will
    satisfy the disjointness condition in Rule \Ref{rule:modules-sint-seq}.~

    \noindent Take $N = N_1 \cup N_2$, $\Ec' = \Ec_1 + \Ec_2$, $c = c_1;c_2$.
    \noindent We construct the required derivation by applying Rule \Ref{rule:modules-sint-seq}.~

    \noindent We get $(E_1 + E_2) \models (\Ec_1 + \Ec_2)$ from the consistency of
    environment modification lemma (Lemma \ref{lem:modules:consistent-env-plus}).
  \end{case}
  \begin{case}[\inferrule{\empty}
      {E \vdash \eps : \{\}}]~\\
    Take $N = \emptyset$, $\Ec' = \{\}$, $c=\eps$. We construct the required
    derivation by applying Rule \Ref{rule:modules-sint-empty}.~

    We get $\{\} \models \{\}$ trivially.
  \end{case}
\end{proof}

\section{Variable Binding and Nominal Techniques}\label{sec:modules:var-binding-nominal}
In this section we are going to give a general introduction to nominal sets
\cite{Gabbay2002,Pitts:2013:NSN:2512979}, and we give motivations why we have decided to
use nominal techniques in our formalisation. We outline definitions of relevant
concepts such as atoms, permutation of atoms, freshness relation, equivariance, and
so on.

Such aspects as freshness of variables and $\alpha$-renaming are often left
implicit in pen-and-paper formalisations, but in proof assistants one has to
pay attention to all the details related to these aspects. Moreover, it is a
well-known fact, that an implementation of variable binding conventions in proof
assistants often requires significant efforts. There are multiple ways to
approach variable binding and related issues in proof assistants, such as de
Bruijn indices~\cite{DEBRUIJN1972381}, locally named and locally nameless
representation \cite{McKinna1993,Gordon1994,Chargueraud2012}, parametric
higher-order abstract syntax \cite{Chlipala:PHOAS}, and so on. One of the goals of
our formalisation of the module system in Coq is to keep it close to the
presentation given in Section \ref{sec:modules-spec}. Since this representation uses
names, changing to another way of treating bound variables (like de Bruijn
indices) will require us to deviate from this representation. Moreover, since
our setting is different from well-studied systems such as the lambda-calculus
(mutual inductive definitions, sets of variables in binding positions), we
would like to use a first-order representation as a more flexible approach
applicable for various structures with variable binding. The approach based on
nominal sets seems to be a good fit in our situation.

Broadly speaking, the theory of nominal sets is a theory of names involved in
different data structures, covering such aspects as variable binding, scope,
and freshness of variables. Nominal sets offer a solid foundation for
expressing \emph{independence} of data structures on the particular choices of bound
variables \cite{PittsAM:nomt}. The theory of nominal sets based on the idea of
permutations of variables and notion of finite support.\footnote{Similarly to
  the development in Agda \cite{Choudhury:nominal-agda} we do not consider the
  notion of \emph{unique} smallest finite support. This allows us to develop
  required notions constructively. For the detailed discussion about nominal sets
  in constructive set theory see \cite{Swan2017arXiv}.} The theory gives a uniform approach to deal
with bound variables and allows for generalisation of binders to various
structures \cite{Clouston2013} We are interested in the particular
generalisation where one can bind a set of variables at once. Moreover, since
nominal sets offer a uniform approach it can be applied to various
structures even if they are quite different from well-studied systems such as the
lambda calculus. Our goal is to define a nominal set semantic objects (see
Section \ref{subsec:modules:sem-obj}) and use nominal reasoning techniques in
our formalisation.

First, we give basic definitions from the theory of nominal sets and show how
these notions can be implemented in the proof assistant Coq (see Section
\ref{subsec:modules:nominal-in-coq}).

\begin{defn}[Atoms]\label{def:modules:atoms}
  Let $\Atom$ be a set of \emph{atoms}, if it satisfies two axioms:
  \begin{itemize}
    \item elements of $\Atom$ has decidable equality, i.e.
      $\forall a,b \in \Atom. (a = b) \lor (a \neq b)$;
    \item $\Atom$ is countably infinite, i.e. for any finite set of atoms $F$,
      $\exists b \in \Atom. b \notin F$.
  \end{itemize}
\end{defn}

\begin{defn}[Permutation]\label{def:modules:perm}
  We write $\Perm$ for the set of all finitely supported permutations $\pi : \Atom -> \Atom$.
  That is, $\pi$ is a bijection on $\Atom$ with the finite support property: there exists
  a finite set of atoms $F$, s.t. $\forall a \notin F. \pi a = a$. We call $F$ the \emph{support}
  of the permutation $\pi$.
\end{defn}
Intuitively, Definition \ref{def:modules:perm} says that $\pi$ is a bijection which
permutes only elements of some finite subset of $\Atom$ and leaves other elements
outside of $F$ untouched.

The set $\Perm$ has a group structure with function composition as group
operation, the identity function (treated as a permutation) $\id{id}$ as a
neutral element, and the inverse function (since permutations are bijections)
as a group inverse.

\begin{defn}[Transposition]\label{def:modules:transposition}
  There is an elementary permutation which we call a transposition:
  \[ \swap{a}{b} c \defeq
  \begin{cases}
    a, \text{if } b = c\\
    b, \text{if } a = c\\
    c, \text{otherwise}
  \end{cases} \]
\end{defn}

The transposition $\swap a b$, being a permutation, has support set $\{a,b\}$.
Any permutation can be non-uniquely factored through a sequence of
transpositions.  Transpositions are involutions, that is, being applied twice
they ``cancel out'' each other:
$\swap a b \circ \swap a b = id$. Transpositions can be generalised from single
atoms to n-tuples of atoms.

\begin{defn}[Generalised transposition]\label{def:modules:gen-transposition}
  Let $\overrightarrow{a} \defeq (a_1,a_2,\dots,a_n)$, $\overrightarrow{b} \defeq (b_1, b_2,\dots,b_n)$, then we define
  a generalised transposition of $\overrightarrow{a}$ and $\overrightarrow{b}$ as a composition of
  simple transpositions:
  \[ \swap{\overrightarrow{a}}{\overrightarrow{b}} \defeq \swap{a_1}{b_1} \circ \swap{a_2}{b_2}\circ \dots
     \circ \swap{a_n}{b_n}\]
\end{defn}

\begin{defn}[Action]
  We define the \emph{action} of a $\Perm$ on a set $X$ as a binary operation
  $\action - -: \Perm -> X -> X$ with the following properties:
  \begin{itemize}
  \item for any $x\in X$, $\action{\id{id}} x = x$
  \item for any $x\in X$, $\pi_1,\pi_2 \in \Perm$, $\action{\pi_1}{(\action{\pi_2}}x) = \action{(\pi_1 \circ \pi_2)}x$
  \end{itemize}
\end{defn}

An action of $\pi$ on some $x$ basically allows one to ``apply'' a permutation to
occurrences of atoms ``inside'' $x$.

\begin{defn}[Support]
  Consider a finite subset $F$ of atoms $\Atom$, and a subgroup of permutations
  $\SubPerm{F}$ satisfying $\forall a \in F. \pi a = a$. We say that $F$ \emph{supports}
  an element $x \in X$ if for any permutation $\pi$ in $\SubPerm{F}$ we have
  $\action \pi x = x$. We write $\supp x$ for the support of $x$.
\end{defn}

We can also give an alternative characterisation of support using transpositions:
if for any two elements \emph{outside} of the a finite subset $F$ of $\Atom$,
it holds that if $\action{\swap{a}{b}} x = x$ for an element $x \in X$, then $F$ supports
x.

The characterisation in terms of transposition gives an intuitive understanding
of support as the finite set of atoms that may occur ``inside'' $x$, since if
we pick two elements outside of the support of $x$, then the action of the
transposition of these elements will have no effect on $x$.\footnote{Since we
  are not defining a \emph{smallest} support, $\supp x$ may contain more
  elements than the set of all atoms occurring in $x$}

\begin{defn}[Nominal set]\label{def:modules:nominal-set}
  A nominal set $\mathbf{X}$ is a set $X$ together with the action $\action - - : \Atom -> X -> X$,
  such that each element of $x \in X$ there \note{exists a finite subset $S$ of atoms $\Atom$ supporting $x$.}
\end{defn}

\begin{defn}[Freshness]\label{def:modules:freshness}
  We say that an atom $a \in \Atom$ is fresh for an element of a nominal set $x \in X$
  \note{if there exists some finite support $S$ of $x$ such that $a \notin S$}.
\end{defn}
Definition \ref{def:modules:freshness} is the usual notion of freshness for a
single atom. We can also define a generalised notion of freshness following \cite{Clouston2013},
which will be useful in the context of our formalisation.
\begin{defn}[Generalised freshness]\label{def:modules:gen-freshness}
  We say that the set of atoms of $x \in X$ is fresh for the set of atoms of $y \in Y$
  (we also can say that $y$ is fresh for $x$) if \note{for some finite support $S_1$ of $x$ and some finite support $S_2$ of $y$
  the following holds: $S_1 \cap S_2 = \emptyset$. That is, $S_1$ and $S_2$ are disjoint.}
\end{defn}

We will use the following notation for the freshness relation: $x \# y$, which one
can read as ``$x$ is fresh for $y$''. We will use the same notation for both freshness
and generalised freshness relations, and say explicitly, which notation is used
in case of ambiguity.

\begin{remark}
  \note{In our Coq development, each implementation of a nominal set module type
    contains a $supp$ function that accepts an element of the nominal set and returns its finite
    support.  Since we are not requiring for a nominal set to have the smallest
    support, we cannot prove some general properties as we would be able to do with
    the smallest support requirement. For example, we cannot show that any $supp$
    function is equivariant, because this is not true in general. Instead, for
    all the definitions of nominal sets used in our formalisation we choose the
    definitions of $supp$ in such a way that they satisfy the
    equivariance property. For every such definition of $supp$ we provide a
    separate proof of equivariance. The same applies to the freshness relation
    (since it depends on $supp$):
    for each nominal set in our formalisation we have to show that it is
    equivariant.}
\end{remark}

\begin{example}\label{ex:modules:lambda-nom}
  The set of lambda terms $\id{Lam}$ (see \reff{eq:modules:lambda-syntax} in
  Section \ref{sec:modules:stlc-norm}, assuming that variables $v \in \Atom$) is
  a nominal set $\mathbf{Lam}$ with the action:
  \begin{align*}
    \action \pi i &\defeq i\\
    \action \pi v &\defeq \pi v\\
    \action \pi {(\lambda x.e)} &\defeq \lambda (\pi x. \action{\pi}{e})\\
    \action \pi {(e_1 e_2)} & \defeq (\action \pi {e_1}) (\action \pi {e_2})
  \end{align*}
  The action of $\pi$ is applied uniformly to all occurrences of variables:
  binding, bound and free.
  Support is defined as follows:
  \begin{align*}
    \supp i &\defeq \emptyset\\
    \supp v &\defeq \{v\}\\
    \supp {(\lambda x.e)} &\defeq \{x\} \cup \supp e)\\
    \supp (e_1 e_2) & \defeq (\supp {e_1}) \cup (\supp {e_2})
  \end{align*}
  \begin{align*}
  \end{align*}
  That is, the support of a lambda term is the set of \emph{all} variables occurring
  in the term.
\end{example}

\begin{example}
  The set of atoms $\Atom$ is a nominal set $\mathbf{A}$. In this case the
  action on $a : \Atom$ is just an application of a permutation:
  \[ \action \pi a \defeq \pi a \]
  The support of $a : \Atom$ is the singleton set $\{a\}$.
\end{example}

\begin{example}\label{ex:modules:finset-nom}
  The set of finite sets of atoms $\Finset$ is a nominal set
  $\mathbf{Fin}_\Atom$. In this case the action of a permutation $\pi$
  on $X \in \Finset$ can be defined as an application of $\pi$ to each
  element of the set:
  \[ \action \pi X  \defeq \{\pi a ~|~ a \in X\} \]
  The support of $X \in \Finset$ is the finite set $X$ itself.
\end{example}

As a running example, which is relevant in our setting of a module system
formalisation, let us consider a simplified version of semantic objects (Figure
\ref{fig:modules:simpl-sem-obj}). We consider only module specifications with
a ``flat'' environment structure, avoiding mutual definitions for simplicity.
\begin{figure}
  \begin{align*}
    t & \in \Atom\\
    \tau \in \mathcal{T}  ::&= t ~|~ \tau_1 -> \tau_2\\
    E \in \Env &= \id{tid} \finmap \tau\\
    \Sigma = \exists T.E \in \MTy &= \Finset(\TSet) \times \Env\\
  \end{align*}
  \titlesembox{Type Expressions}{E \vdash \id{ty} : \tau}
  \[\id{ty}  ::= tid ~|~ \type{Arr}~\id{ty}_1~\id{ty}_2\]
  \begin{mathpar}
    \inferrule{E(\id{tid}) = \tau}
              {E \vdash \id{tid} : \tau}
              \quad\deflabel{\textsc{ty-tid}}\label{rule:modules:ty-tid}

              \and

    \inferrule{E \vdash \id{ty}_i : \tau_i \\ i = [1,2]}
              {E \vdash \type{Arr}~\id{ty}_1~\id{ty}_2 : \tau_1 -> \tau_2}
              \quad\deflabel{\textsc{ty-arr}}\label{rule:modules:ty-arr}
  \end{mathpar}
  \titlesembox{Module Specifications}{E \vdash \id{spec} : \exists T.E'}
  \[\id{spec}::=\kw{type}~tid~|~\kw{type}~tid = \id{ty}~|~\id{spec}_1;\id{spec}_2\]
  \begin{mathpar}
    \inferrule{\empty}
              {E \vdash \kw{type}~\id{tid} : \exists \{\id{t}\} .\{\id{tid} \mapsto \id{t}\}}
              \quad\deflabel{\textsc{spec-type}}\label{rule:modules:mini-spec-type}
              \and
    \inferrule{E \vdash \id{ty} : \tau}
              {E \vdash \kw{type}~\id{tid} = \tau:
                \exists \{\} .\{\id{tid} \mapsto \tau\}}
              \quad\deflabel{\textsc{spec-type-assgn}}\label{rule:modules:mini-spec-type-assgn}
              \and
    \inferrule{E \vdash \id{spec}_1 : \exists T_1.E_1 \\
               E + E_1 \vdash \id{spec}_2 : \exists T_2.E_2 \\
               T_1 \#(E,T_2) \\ \Dom~E_1 \cap \Dom~E_2 = \emptyset}
              {E \vdash \id{spec}_1 ;\id{spec}_2 : \exists(T_1 \cup T_2).(E_1 + E_2)}
              \quad\deflabel{\textsc{spec-seq}}\label{rule:modules:mini-spec-seq}
  \end{mathpar}
  \caption{Simplified semantic objects and elaboration rules for type
    expressions and module specifications.}\label{fig:modules:simpl-sem-obj}
\end{figure}

\begin{example}
  Types $\mathcal{T}$ form a nominal set $\mathbf{T}$ with the action given by
  \begin{align*}
    \action\pi{(\tau_1 -> \tau_2) & \defeq (\action \pi {\tau_1}) -> (\action \pi {\tau_2})}\\
    \action\pi t & \defeq \pi t
  \end{align*}
  And support given by
  \begin{align*}
    \supp{(\tau_1 -> \tau_2)} & \defeq (\supp{\tau_1}) \cup (\supp{\tau_2})\\
    \supp t & \defeq \{t\}
  \end{align*}
\end{example}

\begin{example}\label{ex:modules:env-nom}
  Environments (or finite maps) $\Env$, mapping type identifiers to types
  $\tau$ form a nominal set $\mathbf{Env}$ with the action given by
  \begin{align*}
    \action\pi e & \defeq \{(k \mapsto \aaction{\mathcal{T}}{\pi}{\tau})~ |~ (k \mapsto \tau) \in e\}\\
  \end{align*}
  We get a new finite map by applying the permutation action to
  every element in the codomain. Subscript $\mathcal{T}$ on the action of permutation
  $\pi$ means that we apply the action on nominal set $\mathbf{T}$.

  The support is defined as follows:
  \begin{align*}
    \supp e & \defeq \bigcup(\{\supp{\tau}~ |~ (k \mapsto \tau) \in e\})
  \end{align*}
  That is, the support of $e$ is a union of sets given by supports of all elements in the codomain.
\end{example}

From the examples above we can define a nominal set for the type of module signatures $\Sigma$,
which are pairs of a finite set of existentially quantified variables $T$ and an environment
$E$ where these variables may occur.
\begin{example}
  We define a nominal set $\mathbf{MTy}$ of module signatures $\MTy$ equipped with
  the action defined in terms of actions on nominal sets $\mathbf{Fin}_\Atom$
  and $\mathbf{Env}$
  \[
    \action\pi (\exists T.E) \defeq \exists (\aaction \Finset \pi T).(\aaction \Env \pi E)\\
  \]
  The same applies to the definition of support:
  \[
   \supp (\exists T.E) \defeq (\asupp \Finset T) \cup (\asupp \Env E)\\
  \]
\end{example}

From now on we will omit explicit annotations on the action of a permutation
and on the support, since it can be determined from the type of argument.

An important notion in the theory of nominal sets is the notion of
\emph{equivariance}.
\begin{defn}[Equivariant functions]\label{def:modules:equivar-fn}
  For two nominal sets $\mathbf{X}$ and $\mathbf{Y}$, a function between the carrier sets
  $f : X -> Y$ is called equivariant if it has the following property:
  \begin{align*}
    \forall x \in X, \pi \in \Perm.~ f ~(\action \pi x) = \action \pi {(f~x)}
  \end{align*}
  That is, $f$ preserves the action of a permutation $\pi$.
\end{defn}

\begin{defn}[The category of nominal sets]
  Nominal sets form a category $\Nom$ where objects are nominal sets and
  morphisms are equivariant functions.
\end{defn}

The $\Nom$ category is Cartesian closed. Particularly, our example of $\mathbf{MTy}$
being a nominal set boils down to the fact that $\Nom$ has finite products.

The same way as we defined the notion of equivariant functions above
(definition \ref{def:modules:equivar-fn}), one can define the notion of
equivariant relation.
\begin{defn}[Equivariant relations]
  For two nominal sets $\mathbf{X}$ and $\mathbf{Y}$ a relation
  $\mathcal{R} \subseteq X \times Y$ is equivariant if
  \begin{align*}
    \forall x \in X, y \in Y, \pi \in \Perm.~
    x \mathcal{R} y => (\action \pi {x}) \mathcal{R} (\action \pi {y})
  \end{align*}
\end{defn}

\begin{lemma}
  The following operations are equivariant:
  \begin{itemize}
  \item the union of two finite sets $X$ and $Y$:
    \[\action \pi {X \cup Y} = (\action{\pi}{X}) \cup  (\action{\pi}{Y})\]
  \item the intersection of two finite sets $X$ and $Y$:
    \[\action \pi {X \cap Y} = (\action{\pi}{X}) \cap  (\action{\pi}{Y})\]
  \item the modification of environment $E_1$ by environment $E_2$ (see
    Definition \ref{def:modules:env-modification}):
    \[\action \pi {(E_1 + E_2)} = (\action{\pi}{E_1}) +  (\action{\pi}{E_2})\]
  \end{itemize}
\end{lemma}

\begin{lemma}\label{lem:modules:gen-fresh-equivar}
  \note{The generalised freshness relation for nominal sets $\mathbf{Fin}_\Atom$ and
    $\mathbf{Env}$, with finite supports for the elements given by the $supp$ function
    defined in Examples \ref{ex:modules:finset-nom} and \ref{ex:modules:env-nom}, is equivariant. That is,
  for elements $x$ and $y$ of these nominal sets we have the following:}
  \[ x \# y => (\action{\pi}{x}) \# (\action{\pi}{y}) \]
\end{lemma}

\begin{remark}
  We write $X = Y$, where $X$ and $Y$ are two finite sets, for extensional equality of sets,
  i.e. $x \in X <=> y \in Y$. The same holds for environments.
\end{remark}

Let us consider proofs of equivariance of elaboration relations
(see Figure \ref{fig:modules:simpl-sem-obj}. We present detailed proofs of two lemmas
below.

\begin{lemma}\label{lem:modules:nominal-minielab-equivar}
  The elaboration relation for type expressions $E \vdash \id{ty} : \tau$
  (see Figure \ref{fig:modules:simpl-sem-obj}) is equivariant.
  That is, for any $\pi$, $E$, $\id{ty}$ and $\tau$
  \[ \overset{\mathcal{T}}{E \vdash \id{ty} : \tau} =>
  (\action{\pi}{E}) \vdash \id{ty} : (\action{\pi}{\tau}) \]
\end{lemma}
\begin{proof}
  By induction on derivation $\mathcal{T}$.
  \begin{case}
    $\mathcal{T} =$
    $\inferrule{E(\id{tid}) = \tau}
               {E \vdash \id{tid} : \tau}\quad\nameref{rule:modules:ty-tid}$\\
    From the assumption $E(k) = v$,
    and the definition of the action on environments, we get
    $(\action{\pi}{E})(tid) = \action{\pi}{\tau}$.

    Now, we can construct the required derivation
    \[\inferrule{(\action{\pi}{E})(tid) = \action{\pi}{\tau}}
                {(\action{\pi}{E}) \vdash \id{tid} : \action{\pi}{\tau}}
                \quad\deflabel{\textsc{ty-tid}}\]
  \end{case}
  \begin{case}
    $\mathcal{T} =$
    $\inferrule{\overset{\mathcal{T}_i}{E \vdash \id{ty}_i : \tau_i} \\ i = [1,2]}
    {E \vdash \type{Arr}~\id{ty}_1~\id{ty}_2 : \tau_1 -> \tau_2}\quad\nameref{rule:modules:ty-arr}$

    By induction hypotheses, we get two derivations:

    $\mathcal{T}_1' =$ $(\action{\pi}{E}) \vdash \id{tid} : \action{\pi}{\tau_1}$

    $\mathcal{T}_2' =$ $(\action{\pi}{E}) \vdash \id{tid} : \action{\pi}{\tau_2}$

    We construct the required derivation from $\mathcal{T}_1'$ and $\mathcal{T}_2'$ using the rule
    \nameref{rule:modules:ty-arr}.
  \end{case}
\end{proof}

\begin{lemma}
  The elaboration relation for module specifications ${E \vdash \id{spec} : \exists T.E'}$
  (see Figure \ref{fig:modules:simpl-sem-obj}) is equivariant.
  That is, for any $\pi$, $E$, $E'$, $\id{spec}$ and $T$
  \[ \overset{\mathcal{T}}{E \vdash \id{spec} : \exists T.E'} =>
    (\action{\pi}{E}) \vdash \id{spec} : \exists (\action{\pi}{T}).(\action{\pi}{E}) \]
\end{lemma}
\begin{proof}
  By induction on the derivation $\mathcal E$.
  \begin{case}
    $\mathcal{E}=$
    $\inferrule{\empty}{E \vdash \kw{type}~\id{tid} : \exists \{\id{t}\} .\{\id{tid} \mapsto \id{t}\}}\quad\nameref{rule:modules:mini-spec-type}$

    We use the fact that $\action \pi {\{x\}} = \{\action{\pi}{x}\}$.
    The same holds for the singleton environment
    $\action \pi {\{k \mapsto v\}} = \{k \mapsto \action{\pi}{x}\}$.
    We construct the required derivation using these facts and the rule
    \nameref{rule:modules:mini-spec-type}.
  \end{case}
  \begin{case}
    $\mathcal{E}=$
    $\inferrule{\overset{\mathcal T}{E \vdash \id{ty} : \tau}}
              {E \vdash \kw{type}~\id{tid} = \tau:
                \exists \{\} .\{\id{tid} \mapsto \tau\}}\quad\nameref{rule:modules:mini-spec-type-assgn}.$

    We know that $\action{\pi}{\{\}} = \{\}$, and
    $\action \pi {\{k \mapsto v\}} = \{k \mapsto \action{\pi}{x}\}$. From
    equivariance of the elaboration relation for type expressions (Lemma
    \ref{lem:modules:nominal-minielab-equivar}) with $\mathcal{T}$, we get
    $\mathcal{T'}=$ $( \action \pi E ) \vdash \id{ty} : ( \action \pi \tau )$.

    Now, we can construct the required derivation using the rule
    \nameref{rule:modules:mini-spec-type-assgn}.
  \end{case}
  \begin{case}
    $\mathcal{E}=$
    $\inferrule{\overset{\mathcal{E}_1}{E \vdash \id{spec}_1 : \exists T_1.E_1} \\
               \overset{\mathcal{E}_2}{E + E_1 \vdash \id{spec}_2 : \exists T_2.E_2} \\\\
               T_1 \#(E,T_2) \\ \Dom~E_1 \cap \Dom~E_2 = \emptyset}
    {E \vdash \id{spec}_1 ;\id{spec}_2 : \exists(T_1 \cup T_2).(E_1 + E_2)}\quad\nameref{rule:modules:mini-spec-seq}$

    We have to construct a derivation for
    \[( \action \pi E ) \vdash \id{spec}_1;\id{spec}_2 :
    \action{\pi}{(\exists(T_1 \cup T_2).(E_1 + E_2))}\]

    Using the definition of the action of $\pi$ on module signatures $\MTy$, and the facts that
    the union operation on sets and the modification operation on environments are equivariant,
    we transform our goal to the following form:
    \[( \action \pi E ) \vdash \id{spec}_1;\id{spec}_2 :
    (\exists(\action{\pi}{T_1} \cup \action{\pi}{T_2}).(\action{\pi}{E_1} + \action{\pi}{E_2}))\]
    \note{Next, by equivariance of the freshness relation for nominal sets $\mathbf{Fin}_\Atom$ and
    $\mathbf{Env}$ (Lemma \ref{lem:modules:gen-fresh-equivar})}, we get
    \begin{equation}\label{eq:modules:spec-elab-equivar-1}
      (\action\pi{T_1}) \# (\action\pi{E},\action\pi{T_2})
    \end{equation}
    we also get
    \begin{equation}\label{eq:modules:spec-elab-equivar-2}
      \Dom~(\action\pi{E_1}) \cap \Dom~~(\action\pi{E_2}) = \emptyset
    \end{equation}
    since the permutation action on environments does not affect the domain.
    By induction hypothesis on $\mathcal{E}_1$, we get
    \[\mathcal{E}_1' = (\action\pi E) \vdash \id{spec}_1 :
    \exists (\action\pi{T_1}).(\action\pi{E_1})\]
    By induction hypothesis on $\mathcal{E}_2$, we get
    \[\mathcal{E}_2' = (\action\pi{E + E_1}) \vdash \id{spec}_2 :
    (\action\pi{T_2}).(\action\pi{E_2})\]
    Using the equivariance of the modification operation on environments, we get
    \[\mathcal{E}_2'' = (\action\pi{E} + \action\pi{E_1}) \vdash \id{spec}_2 :
    (\action\pi{T_2}).(\action\pi{E_2})\]
    We get the required derivation using Rule \nameref{rule:modules:mini-spec-seq}
    with $\mathcal{E}_1'$, $\mathcal{E}_2''$, (\ref{eq:modules:spec-elab-equivar-1}),
    and (\ref{eq:modules:spec-elab-equivar-2}).~

    We observe, that the proof of equivariance boils down to showing that operations
    and relations involved in the relation definition are themselves equivariant (see \cite[Section 7.3]{Pitts:2013:NSN:2512979}).
  \end{case}
\end{proof}

The notion of $\alpha$-equivalence is often used to say that two terms are equal ``up to''
renaming of bound variables. Using  $\alpha$-equivalence one can express independence
of particular choices for names of bound variables.
Using nominal techniques, we can give a definition of  $\alpha$-equivalence just in
terms of freshness and permutation of variables.

\begin{example}
  First, let us consider a traditional setting of the lambda calculus. The following
  inductively defined relation specifies $\alpha$-equivalence of lambda terms \cite{Gabbay2002}:
  \begin{mathpar}
    \inferrule{\empty}
              {a =_\alpha a}

   \and

    \inferrule{t_1 =_\alpha t_1' \\
               t_2 =_\alpha t_2'}
              {t_1 t_2 =_\alpha t_1' t_2'}
   \and

   \inferrule{\action{\swap{a_1}{b}}{t_1} =_\alpha \action{\swap{a_2}{b}}{t_2} \\
               b \# (a_1,a_2,t_1,t_2)}
              {\lambda a_1 . t_1 =_\alpha \lambda a_2 . t_2}
  \end{mathpar}
  It has been shown in \cite{Gabbay2002} that the definition of
  $\alpha$-equivalence in terms of permutations corresponds to the usual notion
  of $\alpha$-equivalence on lambda terms. Moreover, the support of
  $t \in \id{Lam/_{=_\alpha}}$ (where $Lam/_{=_\alpha}$ is a quotient set) is
  exactly the set of free variables of $t$.
\end{example}

For our running example of simplified semantic objects given in Figure
\ref{fig:modules:simpl-sem-obj}, we could consider a characterisation
of $\alpha$-equivalence in terms of generalised transpositions
(Definition \ref{def:modules:gen-transposition}). In this case, we would
have to fix some (arbitrary) order for the sets of variables in binding positions.
We write $\overrightarrow{T}$ for the set of variables $T$ ordered according to some
fixed order. We can define an $\alpha$-equivalence relation on $\Env$ as follows:
\begin{mathpar}
  \inferrule{\action{\swap{\overrightarrow{T_1}}{\overrightarrow T}}{E_1}\\
             \action{\swap{\overrightarrow{T_2}}{\overrightarrow T}}{E_2}\\
             T \# (T_1, T_2, E_1, E_2)\\\\
             \card{T_1} = \card{T_2} = \card T
  }
  {\exists T_1. E_1 =_\alpha \exists T_2.E_2}
\end{mathpar}
In this definition we assume that cardinalities of the finite sets of variables
$T_1$, $T_2$, and $T$ are equal.

Unfortunately, since we have to fix some arbitrary order on sets of variables,
some properties of generalised transpositions does not play well with the idea
of sets of variables to be unordered. For example, if we have sets $T = T_1 \cup T_2$
and $T' = T_1' \cup T_2'$, where $T_1$ is disjoint from $T_2$, and $T_1'$ is
disjoint from $T_2'$ this property fails to hold:
\[ \swap{\overrightarrow{(T_1 \cup T_2)}}{\overrightarrow{(T_1' \cup T_2')}}
\neq \swap{\overrightarrow{T_1}}{\overrightarrow{T_1'}} \circ
\swap{\overrightarrow{T_2}}{\overrightarrow{T_2'}} \]

Instead of defining $\alpha$-equivalence in terms of generalised transpositions,
we give a more general definition that imposes constraints on the permutation applied,
instead of using transpositions explicitly.

\begin{align}\label{eq:modules:MTy-alpha-equiv-perm}
  \inferrule{\action \pi {(\exists T_1.E_1) = \exists T_2.E_2}\\
             \forall a. a \in (\supp E_1 - T_1) => \pi a = a}
  {\exists T_1. E_1 =_\alpha \exists T_2.E_2}
\end{align}
The constraint on the permutation $\pi$ says that the permutation is an
identity on free variables of $E_1$, that is, it affects only bound variables
in $E_1$.  Since
$\action \pi {(\exists T_1.E_1) = \exists T_2.E_2} = \exists(\action{\pi}{T_1}).(\action{\pi}{E_1})$,

Rule \reff{eq:modules:MTy-alpha-equiv-perm} says that two module signatures are
$\alpha$-equivalent for any permutation $\pi$ such that it affects only bound
variables of $E_1$ and makes components of the signatures equal. This
definition does not depend on the particular way of constructing a permutation
$\pi$. Particularly, in certain situations, when we need such a permutation, we
can use a generalised transposition to construct it.

All the examples related to the simplified setting of module specifications
given in Figure \ref{fig:modules:simpl-sem-obj} generalises to the full setting
of mutually dependent definitions of semantic objects given in Section
\ref{subsec:modules:sem-obj}.  We will provide some details of how nominal
techniques apply to the full setting when discussing our formalisation in the
Coq proof assistant (see section \ref{subsec:modules:nominal-in-coq}).

In the conclusion of this section, we would like to mention the following. The
use of nominal techniques could provide us with the mechanism for convenient
reasoning about structures with binders using an induction principle that
incorporates the idea of $\alpha$-conversion. As it has been shown
(see \cite{PittsAM:alpsri} and \cite{Urban2007}), for equivariant relations one can derive a stronger induction
principle with an additional finitely supported \emph{induction context}
$C$. Instantiating the induction context appropriately, one can obtain the required
freshness conditions for cases where it is required.

Even for certain formulations of the simply-typed lambda calculus, to be able to
prove the weakening property in a proof assistant, one has to use
$\alpha$-conversion explicitly. That is, if we defined the typing rule for the
case of lambda abstraction in the following form
\[
\inferrule{\Gamma, x:\tau_1 \vdash e : \tau_2 \\ x \notin \id{dom}(\Gamma)}
            {\Gamma \vdash \lambda x.e : \tau_1 -> \tau_2}
\]
then in the proof of the weakening lemma we would have to show that
$x \notin \id{dom}(\Gamma')$ for $\Gamma \subseteq \Gamma'$, knowing that
$x \notin \id{dom}(\Gamma)$. This is usually done informally: by variable
convention we always can pick such $x$, but in the formalisation one has to
make precise why such a renaming is possible. Since the typing relation is
equivariant (and relations used in side conditions as well), one can use a
strengthened induction principle. The induction principle, adapting the
definition from \cite{Urban2007} looks as follows. For any predicate
$P~C~\Gamma~e~\tau$, where $C$ is an additional induction context, we have:
\begin{align*}
  \forall C ~\Gamma~n,&~ P~C~ \Gamma~ n~ Int \\[0.5em]
  \forall C ~\Gamma~ x~ \tau, &~\Gamma(x) = \tau => P ~C ~\Gamma ~x~\tau \\[0.5em]
  \forall C ~ \Gamma~ x~ e ~\tau_1~ \tau_2,&~x \# \Gamma => x \# C =>\\
  &~\Gamma,x:\tau_1 |- e : \tau_2 => (\forall C, P~C~(\Gamma,x:\tau_1)~e~\tau_1) =>\\
  ~&P~C~\Gamma~ (\lambda x. e) (\tau_1 -> \tau_2) \\[0.5em]
  \forall C~\Gamma~ e_1~ e_2~\tau_1~ \tau_2,
  &~\Gamma |- e_1 : \tau_1 -> \tau_2 =>(\forall C, P~C~\Gamma~ e_1 ~(\tau_1 -> \tau_2)) =>\\
  &~\Gamma |- e_2 : \tau_1 => (\forall C, P ~C~ \Gamma~ e_2~ \tau_1) =>\\
  &~ P~C~\Gamma~(e_1~e_2)~\tau_2\\
    \hline\\[-1em]
    \forall ~C~\Gamma~ e~\tau,&~ \Gamma |- e : \tau => P~C~ \Gamma~ e~ \tau
\end{align*}
Notice an additional condition $x \# C$ in the case of lambda
abstraction. Using this induction principle, the case for $\lambda x.e$ could
be proved by instantiating the context $C$ with $\Gamma'$.

\section{Formalisation in Coq}\label{sec:modules-coq}
We have formalised, in the Coq proof assistant, essential parts of the
definitions given in Section \ref{sec:modules-spec} along with the proof
of static interpretation normalisation. We have taken an extrinsic approach
\cite{BentonHKM12}, as opposed to an intrinsic one, to the representation of
the core language, the module language, and the target language, which keeps
our implementation close to the approach presented in the paper.  The extrinsic
encoding has an advantage of being more suitable for code extraction to obtain
a certified implementation.

\note{We use two extra axioms in our Coq formalisation: the axiom of
functional extensionality and the axiom of proof irrelevance. Both
axioms are consistent with the theoretical foundation of the
Coq proof assistant - Calculus of Inductive Constructions (CIC). See
Remark \ref{rem:modules:proof-irrel} for the details about the proof
irrelevance axiom.}

That is, we have implemented the abstract syntax as simple inductive data types
and given separate inductive definitions for relations such as elaboration,
typing, and so on. The semantic objects from Section
\ref{subsec:modules:sem-obj} have been implemented as mutually defined
inductive types using Coq's \texttt{with} clause. The same approach is used in
definitions of relations on semantic objects.

For proof automation, we make use of the \kt{crush} tactic from \cite{cpdt} and
some tactics from \cite{pierce}. We have striven to keep the structure of
proofs explicit, using automation to resolve only the most tedious and
repetitive cases. When proving goals involving dependent types (like vectors)
we use tactics \icode{dependent destruction} and \icode{dependent induction}
from the \icode{Program.Equality} module.

In the following, we are going to discuss issues related to some details of the
implementation in Coq. Particularly, we will discuss the definitions of
environments (finite maps) with properties that simplify proof development,
some issues related to the conservative strict positivity checks in Coq for the
definition of semantic objects, and induction principles used in proofs. We
describe the nominal sets implementation and provide details of definitions of
particular nominal sets relevant for our development. We also present an
implementation of the logical relation and highlight some details related to
the proof of static interpretation normalisation (Theorem \ref{norm.prop}).

We will use Coq's syntax for most of the definitions in the reminder of this chapter,
including lemmas and theorems, which directly correspond to our
implementation. In the proof sketches that follow, we are aiming at
showing the overall ideas and intuitions, which will serve as a guide to our Coq
implementation.

\subsection{Semantic Objects Representation}
As described in Section \ref{sec:modules-spec}, semantic objects are
represented using finite maps (which we also call \emph{environments}) and sets
(see Figure \ref{fig:modules:semobjects}). Indeed, the implementation makes
use of Coq's standard library implementations of such objects. For
environments we define a module type, which specifies operations on
environments used in our implementation following our naming convention. We
define a module, corresponding to this module type by mapping operations from
our custom module type to the operations from the standard
library. Specifically, we use the \texttt{FMapList} implementations of the
\texttt{FMap} interface. In the \texttt{FMapList} implementation, the
underlying data structure is a list of pairs \icode{list (Key * A)}, where
\icode{Key} is the type of keys (a domain) and \icode{A} is the type of values (a
codomain).  This list is equipped with the property of being ordered according
to a strict order of keys (we assume that the type \icode{Key} has a strict
order).  The list and the property are packed together using Coq's records,
which can be seen as a generalisation of $\Sigma$-types.  We use the definition
from the \texttt{FMapList} module from the Coq standard library:
\begin{defn}\label{def:slist}
\begin{lstlisting}
Record slist (A : Type) :=
  {this :> list A; sorted : sort (ltk A) this}.
\end{lstlisting}
\end{defn}
Here, the \icode{this} field denotes an underlying list, and the \icode{sorted}
field represents the property of the list being ordered with respect to
a strict order on keys:
\begin{lstlisting}
Definition ltk : forall A : Type, Key * A -> Key * A -> Prop :=
fun (A : Type) (p p' : Key * A) => Key.lt (fst p) (fst p')
\end{lstlisting}
Coq automatically generates a constructor for the defined record:
\begin{lstlisting}
Build_slist: forall (A : Type) (xs : list A),
             sort (ltk A) this -> slist A
\end{lstlisting}
In our implementation, we instantiate the \texttt{FMapList} module from the standard
library by the name \icode{En}. We use \icode{En.t} to refer to the environment
type constructor, which corresponds internally to \icodet{slist} from Definition
\ref{def:slist}.

The property given by the \icode{sorted} field of the \icode{slist} record gives
us the canonical representation of environments. The canonical representation
allows us to prove an important property, which makes proofs involving
environments easier to write: if two environments contain the same mappings,
they are propositionaly equal in Coq. We call this property \emph{environments
  extensionality}:
\begin{defn}[Environments extensionality]
  For any type \icodet{A : Type} and for any two environments \icodet{E E' : En.t A}, we have
  \begin{lstlisting}
    (forall k : Key, look k E = look k E') -> E = E'
  \end{lstlisting}
  Here the equals sign ``$=$'' refers to the Coq propositional equality, and
  \begin{lstlisting}
    look : Key -> En.t A -> option A
  \end{lstlisting}
  is a lookup function.
\end{defn}

Having this property for environments, one can use all of Coq's rewriting
machinery instead of using a setoid equality.

Sometimes, using the concrete representation of some abstract notion can be
helpful to prove properties by computation. For example, let us consider the
following lemma (we use the notation $++$ for the environment modification operation):
\begin{lstlisting}
  Lemma plus_empty_r: forall {A} (e : En.t A),
    e $++$ empty = e.
  Proof.
    reflexivity.
  Qed.
\end{lstlisting}
The proof of this lemma makes use of the concrete representation of
environments as lists with a well-formedness condition. In fact, this lemma is
just a definitional equality, which is why we can prove it by
\icode{reflexivity} directly. The same trick does not work for the slightly
different lemma:
\begin{lstlisting}
Lemma plus_empty_l: forall {A} (e : En.t A),
    empty $++$ e = e.
\end{lstlisting}
To prove this lemma using the concrete representation would require us to do
induction and use properties of the \icode{fold_left} function on lists, since this
is how environment modification (addition) is defined. Instead, we can use
environment extensionality to prove this lemma by using the specification of the
environment modification operation.

In general, abstracting from the concrete representation has an advantage that
we can change from one representation to another, but in this case we lose a
computational behavior. As we will see in Section
\ref{subsec:modules:coq-norm-static-int} it is not always possible to use an
abstract representation.

\begin{remark}\label{rem:modules:proof-irrel}
In order to prove the environments extensionality property, we have to add the axiom
of \emph{proof irrelevance}:
\begin{lstlisting}
  proof_irrelevance : forall (P:Prop) (p1 p2:P), p1 = p2.
\end{lstlisting}
  One can read this axiom as ``any two inhabitants (proofs) of the same
  proposition are equal''. Although, we could have taken another approach, and
  define the property of list to be ordered as a boolean-valued predicate,
  instead of $\type{Prop}$-valued. Then the property of a list \icode{xs} being
  ordered could look like this: \icode{sorted : isSorted xs ltk = true}.  In
  this case we could have used the fact, that for types with decidable equality
  (type $\type{bool}$, in our case), uniqueness of identity proofs is provable
  (the Hedberg's theorem \cite{Hedberg98}). We use this approach to prove the
  extensionality property for the \icode{MSetList} (more modern) implementation of
  finite sets from the standard library (the well-formedness predicate is
  boolean-valued in this implementation).  The implementation of
  \icode{FMapList} still states sortedness condition in \icode{Prop}. To enable
  easier reuse of the standard library implementation, we have chosen to assume
  proof irrelevance to prove extensionality of environments.
\end{remark}

Having environments as described above, we can now try to define semantic objects using
Coq's mechanism of mutually inductive definitions. First of all, we define type
environments in the obvious way (value environments \icode{VEnv} are defined similarly):
\begin{lstlisting}
Definition TEnv := AEnv Ty.
\end{lstlisting}

We would like to give a definition of semantic objects in the following way:
\begin{lstlisting}
Inductive Env :=
     | EnvCtr : TEnv -> VEnv -> MEnv -> MTEnv -> Env
with MEnv :=
     | MEnvCtr : AEnv Mod -> MEnv
with MTEnv :=
     | MTEnvCtr : AEnv MTy -> MTEnv
with Mod :=
     | NonParamMod : Env -> Mod
     | Ftor : TSet -> Env -> MTy -> Mod
with MTy :=
     | MSigma : TSet -> Mod -> MTy.
\end{lstlisting}

Unfortunately, this definition does not work because of the conservative strict
positivity check implementation in Coq. Specifically in definitions of \icode{MEnv}
and \icode{MTEnv} we use the \icode{VEnv} type constructor. Coq does not
accept such a definition as strictly positive.

\begin{remark}
  One simple example that violates strict positivity of inductive definitions
  is a definition of lambda terms that uses Coq's function space
  to represent lambda-abstraction:
  \begin{lstlisting}
  Inductive term : Set :=
  | App : term -> term -> term
  | Abs : (term -> term) -> term.
  \end{lstlisting}
  This definition violates strict positivity, since the constructor \icode{Abs}
  takes a function from \icode{term} to \icode{term}. That is, the inductive
  type being defined occurs in the negative position, to the left of an arrow.
  If Coq allowed such definitions, its underlying theory would become unsound,
  since one could than write non-terminating definitions (see examples in
  \cite[Section 3.6]{cpdt}).  For that reason Coq performs a check for
  strict positivity, and this check is an overapproximation. That is, there are
  some definitions, which are strictly positive, but Coq is unable to
  recognize that. In the case of semantic objects, this is exactly the case, as
  we will see later.
\end{remark}

If we drop the condition in the definition of \icode{AEnv} that the list is sorted,
then we can write a definition of \icode{Env}. That is, we could have an
association list for \icode{MEnv} and \icode{MTEnv} instead of proper
environments, and than we would have to add a well-formedness condition to all the
theorems related to semantic objects.  In \cite{Rossberg:f-ing-modules} the authors
have chosen the approach with the separate well-formedness condition in the Coq
formalisation. In our case such a design decision would lead to
significant complications due to the presence of mutual inductive
definitions. From a practical point of view, dependent types are useful
exactly for propagating certain invariants associated with the
data structure \cite{Leroy:deptypes}. For that reason, we would like to keep the
well-formedness condition packed together with the underlying association list.

Instead, we introduce an isomorphic \emph{pair-of-vectors} representation of environments,
where the first vector is an ordered vector of keys and the second vector is
a vector of values. By vector we mean the following inductive family of types indexed by
natural numbers, as it is defined in the \icode{Vector} module from the Coq standard library:
\begin{defn}\label{def:vec}~\\
\begin{lstlisting}
  Inductive Vec A : nat -> Type :=
  |nil : Vec A 0
  |cons : forall (h:A) (n:nat), Vec A n -> Vec A (S n).
\end{lstlisting}
\end{defn}
The important feature of vectors is that they carry information about their
lengths in the type ``by construction''. An alternative way to pack a list with
its length is by using subset types:
\begin{lstlisting}
  Definition VecAlt A n := {xs : list A | length xs = n }
\end{lstlisting}

Such a definition would lead to the same problem as before: the definition of semantic objects would
not pass the strict positivity check.

Separating keys and values in different vectors allows us to
define semantic objects in a way acceptable for Coq's strict positivity checker.
The idea of using an isomorphic representation is similar to Wadler's
notion of \emph{views} \cite{Wadler:1987:VWP:41625.41653}. We can see the two
representations of environments as two views associated with the abstract
``type'' of environments, which corresponds to a module specifying operations on
and properties of environments.

More precisely, we define our new representation of environments using Coq's
records and subset types:
\begin{lstlisting}
Definition skeys n := {vs : Vec Key n | vsort Key.lt vs }.

Record VecEnv (A : Type) :=
  mkVecEnv { v_size : nat;
              keys   : skeys v_size;
              vals   : Vec A v_size }.
\end{lstlisting}
This definition again uses dependent types to maintain two invariants, specifying
(\textit{i})  that the vector of keys is sorted, and (\textit{ii})
that the vector of keys and the vector of values have the same length. Notice that using the
subset types for the definition of \icode{skeys} will not cause problems for
the strict positivity check, because the type of keys is independent of the mutual
inductive structure of semantic objects. The reason why we need the lengths of
two vectors to be the same is the following. We are aiming to define our
alternative representation of environments in such a way that we have enough
information to show the isomorphism between the two representations without
adding any side conditions. In particular, if we did not include the condition
on lengths of vectors and used just lists, we would have to add an extra side
condition when defining a conversion function from the pair-of-vector
representation to the one from the standard library. That would again lead to
the same problem as with adding a well-formedness side condition to all the
theorems related to semantic objects.

Before we start showing the isomorphism between the two representations, let
us give some useful definitions related to propositional equality, since we use
will have to reason about equality of dependent types.  The Homotopy Type
Theory book \cite{hott-book} establishes the terminology and provides
definition of basic notions, which we are going to use when we talk about
equalities.
\begin{defn}[Transport]
  We call the following function \emph{transport}:
  \begin{lstlisting}
    transport : forall {A : Type} {a b : A} {P : A -> Type},
                a = b -> P a -> P b
  \end{lstlisting}
\end{defn}
Informally, one can read this definition as
``if \icode{a} is equal to \icode{b} and \icode{P a} holds, then \icode{P b} also holds''.
Following \cite{hott-book}, we are going to call \icode{p : a = b} a
\emph{path} between \icode{a} and \icode{b}.
\begin{defn}[Path concatenation]\label{def:path-concat}
  Paths may be concatenated, which corresponds to transitivity of equality:
  \begin{lstlisting}
    path_concat : forall {A : Type} {x y z : A},
                  x = y -> y = z -> x = z
  \end{lstlisting}
\end{defn}
\begin{defn}[Action on paths]\label{def:ap}
  The application of \icode{f} to a path (or action on paths) :
  \begin{lstlisting}
    ap : forall {A B : Type} {a a' : A}
                (f : A -> B) (p : a = a'),
         f a = f a'.
  \end{lstlisting}
\end{defn}

\noindent Now, we give some useful properties of transport.
\begin{lemma}[Lemma 2.3.9 in \cite{hott-book}]\label{lemma:transp-concat}
  For any \icodet{A : Type}, type family \icodet{B : A -> Type}, \icodet{x y z : A},
  \icodet{u : B x}, and paths \icodet{p : x = y} and \icodet{q : y = z} we have
  \normalfont
  \begin{lstlisting}
    transport q (transport p u) = transport (path_concat p q) u.
  \end{lstlisting}
\end{lemma}
Lemma \ref{lemma:transp-concat} says, that transporting something twice, first
along the path \kt{p} and then along the path \kt{q} is equal to transporting once, but
along the concatenated path \icode{path_concat p q}. The next lemma allows us to move
transport in and out of the application.

\begin{lemma}[Lemma 2.3.11 in \cite{hott-book}]\label{lemma:move_transp}
  For any \icodet{A : Type}, type families \icodet{F,G : A -> Type}, \icodet{a,a' : A},
  \icodet{u : F x}, dependent function \\
  \icodet{f : forall (a : A), F a -> G a}, and
  path \icodet{p : a = a'} we have
  \normalfont
  \begin{lstlisting}
    f (transport p u) = transport p (f u).
  \end{lstlisting}
\end{lemma}

\begin{lemma}\label{lemma:concat-inv}
  For any \icodet{A : Type}, \icodet{x y : A}, and \icodet{p : x = y} the
  following equations hold:
  \begin{enumerate}[(i)]
  \item \icodet{path_concat p (eq_sym p) = eq_refl}
  \item \icodet{path_concat (eq_sym p) p = eq_refl}
  \end{enumerate}
  Where \icodet{(eq_sym p) : y = x}
\end{lemma}

Having lemmas about equality at hand, we can show the isomorphism between the
two representations. The idea of the approach is simple: we use a well-known
list-of-pairs to pair-of-lists correspondence.  Although, for us it is a bit
more subtle, since we use vectors instead of lists, and we have to maintain an
additional invariant --- vector of keys is sorted. First of all, we define zip
and unzip operations on vectors. The zip operation on vectors has the following
type:
\begin{lstlisting}
  vzip : forall {A B n}, Vec A n -> Vec B n -> Vec (A * B) n
\end{lstlisting}
The important difference from the \icode{zip} function on lists is that dependent
types allow us to ensure that the two input vectors and the output vector have the same
size. Writing definitions like this in Coq sometimes is a non-trivial
task. Pattern-matching in Coq requires additional work to pass all required
information for the definition to type-check. That is, we use a \emph{convoy}
pattern \cite{cpdt} to propagate the information that the length of two input
vectors is the same during pattern-matching. The implementation of the \icode{vzip}
function in inspired by \cite{Brecknell:patternmatch}.
\begin{lstlisting}
Definition vzip {A B : Type} :
  forall {n}, Vec A n -> Vec B n -> Vec (A * B) n :=
  fix zip {n} vs := match vs in Vec _ m
                    return Vec B m -> Vec (A * B) m with
  | [] => fun vs' => []
  | cons _ v n0 tl =>
    fun vs' =>
      (match vs' in Vec _ m'
         return (S n0 = m' -> match m' with 0 =>
                (unit : Type) | S _ => Vec _ _ end) with
       | [] => fun _ => tt
       | cons _ v' n1 tl' =>
         fun H =>
         cons _ (v,v') _
         (zip tl (transport (eq_add_S _ _ (eq_sym H)) tl'))
  end) eq_refl
end.
\end{lstlisting}
There are two details to point out in this definition. First, after matching
the first vector against the \icode{cons} constructor, we already know that the
second input vector cannot be empty, although Coq still requires us to
exhaustively cover all the cases. We apply the usual trick here: use
a \icode{return} clause to specify that in the impossible case we return a
trivial type \icode{unit}.  The second detail is the recursive call. The second
argument (the tail of the second vector) is not exactly of the right type. It
has type \icode{Vec B n1}, but we need a term of type \icode{Vec B n0}, where
n0 is the length of the first vector. At this point we use the convoy pattern
to make a connection between the lengths of the two vectors: the return type of the
match on \icode{vs'} is a function, taking equality as an argument:
\begin{lstlisting}
  S n0  = m' -> match m' with
                | 0 => (unit : Type)
                | S _ => Vec _ _ end.
\end{lstlisting}
\noindent Here \icode{n0} is bound to the length of the first vector \icode{vs},
and we apply this function to \icode{eq_refl} to indicate that we expect
lengths of the two vectors to be equal. Thus, in the \icode{cons} case for the
second vector, we now have \icode{H : S n0 = S n1} at our disposal, and we can
use it to construct a term of the right type for the second argument of the
recursive call. In order to do that, we have to transport \icode{tl'}
along the equality \icode{p : n0 = n1} (we will discuss transport later, when
we state the isomorphism between the two representations of environments).  We
get \icode{p} from \icode{H} using injectivity of constructors (we use the
\icode{eq_add_s} lemma from the standard library).

The definition for the \icode{vunzip} function is simpler, and does not
involve complications with dependent pattern-matching.
\begin{lstlisting}
Definition vunzip {A B : Type} :
    forall {n}, (Vec (A * B) n) -> (Vec A n * Vec B n) :=
  fix unzip {n} vs := match vs with
  | [] => ([],[])
  | (a,b) :: tl => (a :: fst (unzip tl), b :: snd (unzip tl))
end.
\end{lstlisting}

Now, we can show that vzip/vunzip are inverses of each other:
\begin{lemma}[vzip/vunzip inverses]\label{lem:modules:zip-unzip}
  For the functions \icodet{vzip} and \icodet{vunzip}, defined above,
  and any \kt{n}, \kt{A}, \kt{B}, the following holds:
  \begin{enumerate}[(i)]
  \item
    for any vector of pairs \icodet{kvs :  Vec (A*B) n}
    we have\\
    \icodet{(fun p => vzip (fst p) (snd p)) (vunzip kvs) = kvs}
  \item
    for any two vectors \icodet{vs : Vec A n} and \icodet{vs' : Vec A n}
    we have\\
    \icodet{vunzip (vzip vs vs') = (vs, vs')}
  \end{enumerate}
\end{lemma}
\begin{proof}
  We show
  \begin{enumerate}[(i)]
  \item By induction on \icode{kvs}.
  \item By induction on \icode{vs} and performing dependent case analysis
    (using the \icode{dependent destruction} tactics) on \icode{vs'}.
  \end{enumerate}
\end{proof}
So far, we have shown a conversion between a pair of vectors and a vector of
pairs, but we need another ``layer'' of conversion functions, which also
compose to identity in both directions.
\begin{defn}\label{def:to-of-list}
  Conversion function between lists and vectors:
\begin{lstlisting}
  to_list : forall (A : Type) (n : nat), Vec A n -> list A
  of_list : forall (A : Type) (xs : list A), Vec A (length xs)
\end{lstlisting}
\end{defn}
We need an auxiliary lemma, before we start showing that these two functions
are inverses of each other.
\begin{lemma}\label{transp-cons}
  For any \icodet{n m : nat}, a vector \icodet{vs : Vec A n}, \icodet{a : A},
  \icode{S} the successor constructor of \icode{nat}, and a path
  \icode{p : n = m} we have
  \normalfont
  \begin{lstlisting}
      transport (ap S p) (a :: vs) = a :: transport p vs.
  \end{lstlisting}
\end{lemma}
\begin{proof}
  First, let us check if the statement type-checks. On the right-hand side
  we have \icode{(transport p vs) : Vec A m} (transporting \icode{Vec A n} along
  the path \icode{p : n = m}). Applying the \icode{cons} constructor gives
  us (according to the Definition \ref{def:vec})
  \icode{(a :: transport p vs) : Vec A (S m)}. For the left-hand side we have
  \icode{(a :: vs) : Vec A (S n)}, and we use action on paths of the successor
  constructor for natural numbers to get the right path :
  \icode{(ap S p) : S n = S m}. Now, it is easy to see that the whole left-hand side
  has the same type as the right-hand side:
  \icode{transport (ap S p) (a :: vs) : Vec A (S m)}.

  We prove the lemma by case analysis on \icode{p}: it suffices to consider the case,
  when \icode{m} is \icode{n} and \icode{p} is \icode{eq_refl}. Then our goal
  reduces to
  \begin{lstlisting}
    transport (ap S eq_refl) (a :: vs) = a :: transport eq_refl vs
  \end{lstlisting}
  Both, \icode{transport} and \icode{ap} compute on \icode{eq_refl}. After simplification,
  we get
  \begin{lstlisting}
    a :: vs = a :: vs
  \end{lstlisting}
  as required.
\end{proof}
Composition of \icode{to_list} and \icode{of_list} in one direction is easy to
state and prove to be the identity. Specifically, the proof of the fact that for
any \icode{A : Type} and \icode{xs : list ~A},
\icode{to_list (of_list xs) = xs} can be found in the standard library. Let us
focus on the other direction (there is no proof of this property in the
standard library). Even to state it requires some work. If we did it naively, as
``for any \icode{n : nat}, \icode{A : Type}, and \icode{vs : Vec A n},
\icode{(of_list (to_list vs)) = vs}'', then our definition would not type-check
and we would get the error ``The term "\icode{vs}" has type "\icode{Vec A n}" while it is
expected to have type "\icode{Vec A (length (to_list vs))}"''.  The reason for
that error is that the right-hand side and the left-hand side of the equation are of
different type. We can fix this problem using the following observation: if we pass a
vector \icode{vs : Vec A n} to the \icode{to_list} function, then the resulting
list will be exactly of length \icode{n}.
\begin{lstlisting}
  to_list_length : forall (n : nat) (A : Type) (vs : Vec A n),
                   n = length (to_list vs)
\end{lstlisting}

\begin{remark}\label{rem:to-list-length}
  Notice that in our Coq implementation we make the definition
  of \icode{to_list_length} transparent by using the \icode{Defined} keyword
  instead of \icode{Qed}. This way, we make the term, witnessing the equality,
  available for simplification, so we can exploit definitional qualities in
  some proofs.
\end{remark}

Now, we are ready to state properties of \icode{to_list} and \icode{of_list}.
\begin{lemma}\label{lem:modules:to-of-list}
  For the functions \icodet{to_list} and \icodet{of_list} (Definition \ref{def:to-of-list}),
  we have the following equations:
  \begin{enumerate}[(i)]
  \item for any \icodet{A : Type} and \icodet{xs : list ~A} we have\\
    \icodet{to_list (of_list xs) = xs}
  \item for any \icodet{A : Type}, \icodet{n : nat}, \icodet{vs : Vec A n} we have\\
    \icodet{of_list (to_list vs) = transport (to_list_length vs) vs}
  \end{enumerate}
\end{lemma}
\begin{proof}
  We show
  \begin{enumerate}[(i)]
  \item by induction on \icode{xs}. In the base case, both sides of the equation
    evaluate to the empty list. In the induction step, we simplify the goal and
    rewrite it using the induction hypothesis.
  \item by induction on \icode{vs}. In the base case both sides of the equation
    evaluate to the empty list. In the induction step, we apply the following
    equational reasoning:
    \begin{align*}
      &\text{\icode{h :: of\_list (to\_list vs)}}\\
      & = \text{\icode{transport (to_list_length (h :: vs)) (h :: vs)}}\\
      & = \text{\{by definitional equality (see Remark \ref{rem:to-list-length})\}}\\
      & = \text{\icode{transport (ap S (to_list_length vs)) (h :: vs)}}\\
      & = \text{\{by lemma \ref{transp-cons}\}}\\
      & = \text{\icode{h :: transport (to_list_length vs) vs}}\\
      & = \text{\{by induction hypothesis\}}\\
      & = \text{\icode{h :: of_list (to_list vs)}}
    \end{align*}
  \end{enumerate}
\end{proof}
Another invariant, which should be preserved by the conversion functions is the
order of the elements. We have proved several auxiliary lemmas showing that the
order is preserved by the functions \icode{vzip}/\icode{vunzip} and the
functions \icode{to_list}/\icode{of_list}. Proofs of these lemmas are
quite straightforward, since none of these functions rearrange the elements.

We define the conversion function between the two representations of environments
using the definitions above. We use \icode{En.t} to refer to the standard library
implementation of environments corresponding to Definition \ref{def:slist}:
\begin{lstlisting}
Definition toOrdEnv {A : Type} (ve : VecEnv A) : En.t A :=
  match ve with
  | mkVecEnv _ _ (exist _ ks sorted_ks) vs =>
       let skvs := vzip_sorted ks vs sorted_ks
       in Build_slist (to_list_sorted (vzip' (ks,vs)) skvs)
  end.

Definition fromOrdEnv {A : Type} (oe : En.t A) : VecEnv A :=
  match oe with
  | @En.Build_slist _ xs xs_sort =>
    let kvs := vunzip (of_list xs) in
    let vs := snd kvs in
    let skvs := vunzip_sorted (of_list xs)
                              (of_list_sorted _ xs_sort)
    in mkVecEnv A _ (exist _ (fst kvs) skvs) vs
  end.
\end{lstlisting}
Notice that in the definition of \icode{toOrdEnv} we use the function
\begin{lstlisting}
  zip' : forall {A B : Type}, Vec A n * Vec B n -> Vec (A * B) a,
\end{lstlisting}
which is a curried version of the \icode{zip} function.

Let us prove congruence lemmas, which we are going to use in our proof that
the functions \icode{toOdrEnv}/\icode{fromOrdEnv} are inverses of each other.
\begin{lemma}[Congruence for ordered-list environments]\label{ordenv-congr}
  For any \icodet{A : Type}, lists \icodet{xs ys: list (Key.t * A)},
  which are sorted according to the strict order of keys \icodet{sxs : sort ltk xs},
  \icodet{sys : sort ltk ys}, if  \icodet{xs = ys} then two records, corresponding to
  the environments are equal: \\
  \icodet{\{| En.this := xs; En.sorted := sxs |\}} =\\
  \icodet{\{| En.this := ys; En.sorted := sys |\}}.
\end{lemma}
\begin{proof}
  Essentially, this is the same property as for subset types in presence of proof
  irrelevance. The first components are equal by assumption, and \icodet{sxs = sys}
  by proof irrelevance.
\end{proof}

\begin{lemma}[Congruence for pair-of-vectors environments]\label{vecenv-congr}
  For any \icodet{A : Type}, \icodet{n,n' : A}, two sorted vectors of keys\\
  \icodet{ks : Vec Key.t n}, \icodet{ks_sort : vsort Key.lt ks},\\
  \icodet{ks' : Vec Key.t n'}, \icodet{ks_sort' : vsort Key.lt ks'},\\
  two vectors of values \icodet{vs : Vec A n} and \icodet{vs' : Vec A n'}, and
  path \icodet{p : n = n'}, \\
  if  \icodet{(transport p ks) = ks'} and  \icodet{(transport p vs) = vs'} then
  \normalfont
  \begin{lstlisting}
  {| v_size := n;  keys := exist _ ks ks_sort;   vals := vs  |} =
  {| v_size := n'; keys := exist _ ks' ks_sort'; vals := vs' |}.
  \end{lstlisting}
\end{lemma}
\begin{proof}
  The proof follows essentially the same patterns as the proof of Lemma \ref{ordenv-congr}.
  First, by case analysis on \icodet{p}, we get \icodet{ks = ks'} and \icodet{vs = vs'},
  because transport computes on \icodet{eq_refl}. It only remains to be shown, that
  \icodet{ks_sort = ks_sort'}, but this equality holds by proof irrelevance.
\end{proof}

\begin{lemma}[toOrdEnv/fromOrdEnv inverses]\label{ordvec-inv}
  For any \icodet{A}, \icodet{ve : VecEnv A}, and \icodet{oe : En.t} we have
  \begin{enumerate}[(i)]
  \item \icodet{(toOrdEnv (fromOrdEnv oe)) = oe}
  \item \icodet{(fromOrdEnv (toOrdEnv ve)) = ve}
  \end{enumerate}
\end{lemma}
\begin{proof}
  We prove
  \begin{enumerate}[(i)]
    \item using Lemmas \ref{ordenv-congr}, \ref{lem:modules:to-of-list}(\textit i) and
      \ref{lem:modules:zip-unzip}(\textit i);
    \item as follows. This direction is a bit harder to prove, since we have to reason
      about equality of vectors. We use Lemma \ref{vecenv-congr}, which gives us
      two goals that are quite similar. We prove them using lemmas about
      transport (Lemmas \ref{lemma:move_transp} and \ref{lemma:transp-concat}) to bring
      together paths that give reflexivity by Lemma \ref{lemma:concat-inv}, and
      Lemmas \ref{lem:modules:to-of-list}(\textit{ii}) and
      \ref{lem:modules:zip-unzip}(\textit{ii}).
  \end{enumerate}
\end{proof}

We could have defined similar operations on pair-of-vectors environments
and prove all the required properties as for the environments from the standard
library, but this would be a quite time consuming process. Instead, we use Coq's
coercion mechanism. The coercion functions are precisely the functions given by
the isomorphism between the two representations.
\begin{lstlisting}
  Coercion _to {A} := toOrdEnv (A:=A).
  Coercion _from {A} := fromOrdEnv (A:=A).
\end{lstlisting}
In most situations, Coq inserts coercion functions automatically in an expected way,
which simplifies development using a pair-of-vectors representation of environments.
However, if there is some ambiguity in what way coercions could be applied, we
have to fallback to manual application of the coercion functions.

Our Coq development shows that in most of the proofs involving the pair-of-vectors
environments it suffices to just use Lemma \ref{ordvec-inv}.

We can now define semantic objects using a pair-of-vectors representation
without violating Coq's strict positivity check. The isomorphism between the
two representations ensures that we can transfer properties of one representation
to the other.
\begin{defn}[Semantic objects]\label{def:semobj}~
  \begin{lstlisting}
    Inductive Env :=
    | EnvCtr : TEnv -> VEnv -> MEnv -> MTEnv -> Env
    with MEnv :=
    | MEnvCtr : VecEnv Mod -> MEnv
    with MTEnv :=
    | MTEnvCtr : VecEnv MTy -> MTEnv
    with Mod :=
    | NonParamMod : Env -> Mod
    | Ftor : TSet -> Env -> MTy -> Mod
    with MTy :=
    | MSigma : TSet -> Mod -> MTy.
  \end{lstlisting}
\end{defn}

The definition of interpretation environments has a similar structure and uses the
same approach with a pair-of-vector representation.

\begin{defn}[Interpretation environments]\label{def:modules:int-envs}~
\begin{lstlisting}
Definition IVEnv := EnvMod.t (label*Ty).

Inductive IEnv :=
| IEnvCtr : TEnv -> IVEnv -> IMEnv -> MTEnv -> IEnv
with IMEnv :=
| IMEnvCtr : (VecEnv.VecEnv IMod) -> IMEnv
with IMod :=
| INonParamMod : IEnv -> IMod
| IFtor : IEnv -> TSet -> Env -> MTy -> mid -> mexp -> IMod.
\end{lstlisting}
\end{defn}

\subsubsection{Operations on Semantic Objects}
We briefly describe operations on semantic objects, such as lookup for
different types of identifiers, including \emph{long} identifiers, which
represent paths in the nested module structure. Most of the operations are just
liftings of corresponding operations of ``flat'' environments, such as
described in the subsection \ref{subsec:modules:sem-obj}. All the implicit
injections and operations on components of semantic objects, mentioned in
Section \ref{subsec:modules:sem-obj}, are made explicit in our Coq
implementation.

We consider several examples of definitions to sketch the overall idea. First,
we start with the concept of long identifiers. Long identifiers are defined as
an inductive data type with two constructors: one for the type, value or module
identifier, and the other one to build a path out of a sequence of module
identifiers.

\begin{lstlisting}
Inductive longtid :=
| Tid_longtid : tid -> longtid
| Long_longtid : mid -> longtid -> longtid.
\end{lstlisting}
This definition shows how the long identifiers for type lookup are
defined. Our implementation contains two more similar definitions of long
identifiers, \icode{longvid} and \icode{longmid} for looking up values and
modules, respectively.

The lookup function can be defined by recursion on the structure of a
long identifier.

\begin{lstlisting}
Fixpoint lookLongTid (longk : longtid) (e : Env) : option Ty :=
  match e  with
  | EnvCtr te _ me _ =>
    match longk with
    | Tid_longtid k => look k te
    | Long_longtid m_id longk' =>
      match (lookMid m_id me) with
      | Some (NonParamMod e') => lookLongTid longk' e'
      | _ => None
      end
    end
  end.
\end{lstlisting}

For the values and modules components of the semantic objects we define the similar functions,
following the name convention.
\begin{lstlisting}
  lookLongVid (longk : longvid) (e : Env) : option Ty
  lookLongMid (longk : longmid) (e : Env) : option Mod
\end{lstlisting}
The lookup function for ``flat'' environments
\begin{lstlisting}
  look : forall A : Type, En.key -> En.t A -> option A
\end{lstlisting}
is polymorphic with respect to the type of values in the
environment. Environments containing modules and module types (\icode{MEnv} and
\icode{MTEnv}) are part of the mutually recursive structure and therefore
wrapped in constructors. Before looking up in these environments we
pattern-match on the respective constructor to ``unwrap'' the environment and
then apply the \icode{look} function. The \icode{add} operation (the operation
that adds a new mapping to an environment) in case of the pair-of-vectors
environment types \icode{MEnv} and \icode{MTEnv}, first transports an
environment through the isomorphism, applies the usual \icode{add} (of standard
library implementation of environments) and then transports the result
back. The operation of environment modification for the semantic objects is
defined componentwise.

\subsection{Induction Principles}
In order to prove theorems by induction over the structure of semantic objects,
or relations containing mutual definitions, Coq's \texttt{Scheme} command is
used to generate suitable induction principles. For some of the definitions,
such as those for semantic objects and interpretation environments, the
generated induction principles are not sufficiently strong, which is caused
by the presence of nested inductive types; some constructors take environments
as parameters, and the environments, being essentially list-like structures,
make the whole definition a nested inductive definition.
For each of these cases, a suitable induction principle is defined manually,
following essentially the same approach as in \cite[Section 3.8]{cpdt}.

Let us consider an induction principle for the semantic objects given by
Definition \ref{def:semobj}. Usually, Coq generates an induction principle from
the definition of an inductive data type, but nested inductive definitions are
not covered by this procedure. In the case of \icode{MEnv} (the case of
\icode{MTEnv} is similar), we want the induction hypothesis saying that some
predicate \icode{P : Mod -> Prop} holds for all the values in the
pair-of-vectors environment argument of the \icode{MEnvCtr} constructor.  Since
we are interested only in a predicate on values in the environment (a
codomain), we can use the projection from the pair-of-vectors representation to
a vector of values:
\begin{lstlisting}
  vals : forall {A : Type} (v : VecEnv A), Vec A (v_size A v)
\end{lstlisting}
That is, we just need a predicate stating that some property holds for all the
elements in a vector.  There are at least two ways to define such a predicate:
by recursion and by induction.  We are going to use the inductive variant of the
predicate from the \icode{Vector} module of the standard library:
\begin{lstlisting}
Inductive Forall {A} (P: A -> Prop)
          : forall {n} (v: t A n), Prop :=
 |Forall_nil: Forall P []
 |Forall_cons {n} x (v: t A n) : P x -> Forall P v ->
                                 Forall P (x::v).
\end{lstlisting}
With this definition of the \icode{Forall} predicate, we can define a
sufficiently strong induction principle for semantic objects:

\begin{lstlisting}
Definition Env_mut' (P : Env -> Prop) (P0 : MEnv -> Prop)
  (P1 : MTEnv -> Prop)
  (P2 : Mod -> Prop) (P3 : MTy -> Prop)
  (f : forall (t : TEnv) (v : VEnv) (m : MEnv),
       P0 m -> forall m0 : MTEnv, P1 m0 -> P (EnvCtr t v m m0))
  (f0 : forall (t : VecEnv Mod),
      Forall P2 t -> P0 (MEnvCtr t))
  (f1 : forall (t : VecEnv MTy),
          Forall P3 t -> P1 (MTEnvCtr t))
  (f2 : forall e : Env, P e -> P2 (NonParamMod e))
  (f3 : forall (t : TSet) (e : Env),
        P e -> forall m : MTy, P3 m -> P2 (Ftor t e m))
  (f4 : forall (t : TSet) (m : Mod), P2 m -> P3 (MSigma t m)) :=
fix F (e : Env) : P e :=
  match e as e0 return (P e0) with
  | EnvCtr t v m m0 => f t v m (F0 m) m0 (F1 m0)
  end
with F0 (m : MEnv) : P0 m :=
  match m as m0 return (P0 m0) with
    | @MEnvCtr t => let fix step {n} (ms : Vec Mod n) : Forall P2 ms :=
                       match ms in Vec _ n' return @Forall _ P2 n' ms with
                         | [] => Forall_nil P2
                         | y :: l =>
                           @Forall_cons _ P2 _ y _ (F2 y) (step l)
                       end
                   in f0 t (step (vals _ t))
  end
with F1 (m : MTEnv) : P1 m :=
  match m as m0 return (P1 m0) with
  | @MTEnvCtr t => let fix step {n} (ms : Vec MTy n) : Forall P3 ms :=
                       match ms with
                         | [] => Forall_nil P3
                         | y :: l =>
                           @Forall_cons _ _ _ y l (F3 y) (step l)
                       end
                  in f1 t (step t)
  end
with F2 (m : Mod) : P2 m :=
  match m as m0 return (P2 m0) with
  | NonParamMod e => f2 e (F e)
  | Ftor t e m0 => f3 t e (F e) m0 (F3 m0)
  end
with F3 (m : MTy) : P3 m :=
  match m as m0 return (P3 m0) with
  | MSigma t m0 => f4 t m0 (F2 m0)
  end
for F.
\end{lstlisting}

A large portion of this induction principle is generated by the \icode{Scheme}
command.  The important modifications we have had to do manually are concerned with
the \icode{f0} and \icode{f1} cases (and the \icode{F0} and \icode{F1} cases of the
proof term respectively), where we use the \icode{Forall} predicate to specify the
desired property.

\subsection{Nominal Techniques in Coq}\label{subsec:modules:nominal-in-coq}
There are existing developments of nominal techniques for proof assistants.
Probably, the most developed one is the Nominal Isabelle
package for the Isabelle proof assistant \cite{Urban2005}, which includes generalised
name abstraction \cite{Urban2011}. On the other hand, for the Coq proof assistant,
there are no packages for nominal techniques in the standard distribution. Probably,
the only known work on nominal techniques in Coq is \cite{Aydemir:2007}. This work is mostly
focused on the case of simply typed lambda calculus and does not cover generalised name
abstraction.

We have developed an implementation of notions described in
Section \ref{sec:modules:var-binding-nominal} using Coq's module system along with
dependent records. Ideally, we would like to use only dependent records in our
implementation, but unfortunately, finite sets from the Coq's standard library
are implemented as parameterised modules. We wanted to use the standard library in
our development as much as possible to avoid extra efforts spent on implementing
standard functionality.

We started with the definition of atoms. Since the definition of atoms involves
finite sets, and since the nominal techniques use finite sets extensively, we decided
to use an \icode{MSet} implementation of finite sets from Coq's standard
library. We expose finite sets through our own module type, which adds
required operations and properties missing in the \icode{MSet} interface:
\begin{itemize}
\item set disjointness relation;
\item set extensionality;
\item map operation on sets.
\end{itemize}
We call our module type of finite sets \icode{SetExtT} and the implementation of
this module type \icode{SetExt}.

We define the following module type for atoms:
\begin{lstlisting}
Module Type Atom.
  Declare Module V : SetExtT.
  Axiom Atom_inf : forall (X : V.t), {x : V.elt | $\sim$ V.In x X}.
End Atom.
\end{lstlisting}
We use \icode{V.t} for the type of finite sets and \icode{V.elt} for the type of
elements.  The \icode{Atom_inf} axiom says that for any given finite subset of
atoms, one can always find an element, which is not in this finite subset. We use
subset types to specify that there exists such an element. It is important to
use \icode{Type} in this definition and not \icode{Prop}, since if we defined
the \icode{Atom_inf} as \icode{forall (X : V.t), exists x : V.elt, $\sim$ V.In x
  X}, we would not be able to use \icode{Atom_inf} to construct functions that
generate fresh atoms. This is due to limitations on \icode{Prop}, allowing to
eliminate propositions only to \icode{Prop} again. In other words, proofs can be
used to construct other proofs, but not programs.

We define a parameterised module \icode{Nominal} that accepts an implementation of
the \icode{Atom} module type. Definitions below are contained in the \icode{Nominal}
module.

Let us first give definitions of notions required to define a permutation.
We define predicates \icode{is_inj} and \icode{is_surj}, representing injectivity
and surjectivity of a function, respectively, as follows:

\begin{lstlisting}
  Definition is_inj  {A B : Type} (f : A -> B) : Prop :=
    forall (x y : A), f x = f y -> x = y.
  Definition is_surj {A B : Type} (f : A -> B) : Prop :=
    forall (y : B), exists (x : A), f x = y.
\end{lstlisting}

We then say that a function \icode{f} is bijective if it is both injective and surjective:
\begin{lstlisting}
  Definition is_biject {A B} (f : A -> B) :=
    (is_inj f) /\ (is_surj f).
\end{lstlisting}

We define a predicate \icode{has_fin_support} of a function \icode{f} as
\begin{lstlisting}
  Definition has_fin_supp f :=
      exists S, (forall t, $\sim$ V.In t S -> f t = t).
\end{lstlisting}

Finally, we define a finitely supported permutation using Coq's dependent records:
\begin{lstlisting}
Record Perm :=
      { perm : V.elt -> V.elt;
        is_biject_perm :  is_biject perm;
        has_fin_supp_perm : has_fin_supp perm}.
\end{lstlisting}
That is, to define an inhabitant of \icode{Perm}, one needs to provide
a function and proofs of two properties: that the function is a bijection, and
that is has a finite support. We call the projection \icode{perm} of the \icode{Perm}
record the \emph{underlying function} of a permutation.

Notice that since we defined properties of the underlying function of a
permutation as inhabitants of type \icode{Prop}, in presence of the proof
irrelevance axiom, we can prove that two permutations are equal if their
underlying functions are equal.

As an example, let us define first an identity permutation. We take the identity
function \icode{id} as an underlying permutation function.
Proofs of required properties of permutation are simple, and we use the \icode{refine}
tactic here to construct the permutation. The \icode{refine} tactic allows
one to provide parts of the definition leaving other parts as ``holes''
that generate proof obligations.
\begin{lstlisting}
Definition id_perm : Perm.
  refine ({| perm:=id; is_biject_perm := _; has_fin_supp_perm := _ |}).
  + split. auto. refine (fun y => ex_intro _ y _);reflexivity.
  + exists V.empty;intros;auto.
Defined.
\end{lstlisting}

Next, we define the composition of permutations. We use the same approach here,
namely the \icode{refine} tactic with a partially constructed record,
corresponding to the permutation.
\begin{lstlisting}
Definition perm_comp (p p' : Perm) : Perm.
  refine ({| perm:= (perm p) ∘ (perm p');
             is_biject_perm := _;
             has_fin_supp_perm := _ |}).
    (* Proofs are omitted *)
Defined.
\end{lstlisting}
We omit proofs of obligations generated by \icode{refine} here. The proof that composition
is a bijection boils down to facts that injectivity and surjectivity are preserved by
function composition. For the finite support we choose the union of supports of the
composed permutations. We define the following notation for the composition of permutations:
\begin{lstlisting}
  Notation "p ∘p p'" := (perm_comp p p') (at level 40).
\end{lstlisting}

The transposition function of two atoms follows Definition \ref{def:modules:transposition}:
\begin{lstlisting}
  Definition swap_fn (a b c : V.elt) : V.elt :=
    if (V.E.eq_dec a c) then b
      else (if (V.E.eq_dec b c) then a
        else c).
\end{lstlisting}
We prove that the \icode{swap_fn} function is a bijection and that its finite
support is the two element set $\{a,b\}$. By packing the \icode{swap_fn}
function with these proofs we define the corresponding instance of
\icode{Perm}.

To define the underlying function of a generalised transposition
(Definition \ref{def:modules:gen-transposition}) we use the \icode{fold_right} function
that accumulates elementary swaps by composing them as functions starting from the
identity function:
\begin{lstlisting}
  Definition swap_iter_fn (vs : list (V.elt * V.elt))
    : V.elt -> V.elt :=
      fold_right (fun (e' : (V.elt * V.elt)) (f : V.elt -> V.elt) =>
                  let (e1,e2) := e' in f ∘ (swap_fn e1 e2)) id vs.
\end{lstlisting}
We prove that this function satisfies properties of finitary permutations. That is,
\icode{swap_iter} is bijective and has a finite support, which is a set of all
variables from the argument list. With these properties at hand we can construct an
inhabitant of the \icode{Perm} type.

Our implementation also provides the functionality for generation of fresh names
for a given set of atoms. The freshness relation is defined as it is described in
Section \ref{sec:modules:var-binding-nominal} (Definitions \ref{def:modules:freshness}
and \ref{def:modules:gen-freshness}).
\begin{lstlisting}
  Definition fresh a A := ~ V.In a A.
  Definition all_fresh (x y : V.t) :=
      forall k, ~(V.In k x /\ V.In k y).

  Infix "$\#$" := fresh (at level 40) : a_scope.
  Infix "$\#$" := all_fresh (at level 40) : as_scope.

  Delimit Scope a_scope with Atom.
\end{lstlisting}
We overload the notation for the freshness relation to use it in both cases: for
a single atom and for a set of atoms.

We start with generating one fresh name along with the proof of
freshness. First, we define a type of functions from the type of finite sets
to the type of atoms, with the property that when applied to a finite set, it returns a
fresh atom (with respect to the given set).
We use subset types to equip a function with the property:
\begin{lstlisting}
  Definition FreshFn a :=
    {f : V.t -> V.elt | forall x, ((f x) $\#$ a)%Atom.
\end{lstlisting}
We use explicit scope annotation here to point out which freshness relation we use.

Now, we can define a function that takes a set of atoms and returns a fresh atom with
the proof of freshness. In order to obtain a fresh atom we use the fact that the
set of atoms is countably infinite:
\begin{lstlisting}
Definition fresh_fn : forall a : V.t, FreshFn a :=
fun a => let H := Atom.Atom_inf a in
         exist (fun f : t -> elt => forall x : t, (f x # a)%Atom)
               (fun _ : t => proj1_sig H)
               (fun _ : t => proj2_sig H).
\end{lstlisting}

In this way we abstract the mechanism of fresh atoms generation, since it will
work with any implementations of module type \icode{Atom} used to instantiate
the \icode{Nominal} module.

Next, we generalise the \icode{fresh_fn} to return a set of fresh atoms with the proof
of freshness. Again, we first define a subset type that packs together a finite set,
the property of freshness, and the cardinality of the set of fresh atoms:
\begin{lstlisting}
Definition AllFresh a n :=
   { b : V.t | (b # a) /\ V.cardinal b = n }.
\end{lstlisting}

To generate a set of \icode{n} fresh atoms (with respect to some finite set
\icode{X}), we have to recursively pass the set \icode{X} with a newly
generated atom added to the set to ensure freshness of all \icode{n} new atoms. We
define a function that generates a set of \icode{n} atoms by recursion on \icode{n}:
\begin{lstlisting}
Fixpoint get_freshs_internal (X : V.t) (n : nat) : V.t :=
  match n with
  | O => empty
  | S n' => let fatom := (proj1_sig (fresh_fn X)) X in
            add fatom (get_freshs_internal (add fatom X) n')
  end.

\end{lstlisting}

By induction on \icode{n} one can show that these atoms are indeed fresh and
the cardinality of the resulting set is equal to \icode{n}. Finally, we can define
a function that returns a value of type \icode{AllFresh a n} for some finite
set \icode{a} and a natural number \icode{n}:
\begin{lstlisting}
Definition get_freshs (X : V.t) (n : nat) : AllFresh X n :=
  exist _ (get_freshs_internal X n)
          (conj (get_freshs_internal_all_fresh n X)
                (get_freshs_cardinality n X)).
\end{lstlisting}

To implement a nominal set we use Coq's module system. That is, we define a
module type of nominal sets, which follows Definition
\ref{def:modules:nominal-set} and includes such components as the type of
elements, the action of a finitary permutation on the elements of this type (along
with properties of the action), and the finite support.

\begin{lstlisting}
Module Type NominalSet.
    Import V.

    Parameter X : Type.

    Parameter action : Perm -> X -> X.
    Notation "r @ x" := (action r x) (at level 80).

    Axiom action_id : forall (x : X), (id_perm @ x) = x.
    Axiom action_compose : forall (x : X) (r r' : Perm),
                           (r @ (r' @ x)) = ((r ∘p r') @ x).

    Parameter supp : X -> V.t.
    Axiom supp_spec : forall  (r : Perm)  (x : X),
      (forall (a : elt), In a (supp x) -> (perm r) a = a) ->
         (r @ x) = x.

End NominalSet.
\end{lstlisting}

\begin{remark}
  \note{The module type \icode{NominalSet} allows for defining a
    support function \icode{supp} that returns a finite support of an element that
    is not necessarily the smallest one, meaning that not all the definitions of \icode{supp}
    will be equivariant functions. In each instance
    of \icode{NominalSet} in our Coq development there is an obvious
    way to define a \icode{supp} function such that it returns a
    smallest support of an element. Although, this is not enforced by
    the definition of the nominal set that we have. One can remedy
    this by adding an explicit constraint on a \icode{supp} function to \icode{NominalSet}
    saying that the function must be equivariant.}
\end{remark}

Our \icode{Nominal} module includes an implementation of \icode{NominalSet} for
the type of finite sets of atoms in the way it is given by Definition
\ref{ex:modules:finset-nom}.  This module also includes additional facts like an
action on a singleton set, equivariance of the union and intersection operation
on finite sets, and equivariance of the set disjointness relation.

The running example of simplified semantic objects from Section
\ref{sec:modules:var-binding-nominal} is implemented as a nominal set using the
\icode{NominalSet} signature. Our implementation includes both definitions of
$\alpha$-equivalence: the one using generalised transpositions and the other
with a condition on a permutation. Moreover, we have developed a proof of
equivariance of the elaboration relation in the simplified setting (Figure
\ref{fig:modules:simpl-sem-obj}).

We have defined a nominal set of full semantic objects. One interesting aspect
to point out is that the action and the support, being defined as fixpoints,
were accepted by Coq, despite the fact that they call \icode{map} and
\icode{fold_right} functions for the case of module environments \icode{MEnv}
and module type environments \icode{MTEnv}. We provide one example of the permutation
action definition for semantic objects:
\begin{lstlisting}
Fixpoint action (p : Perm) (E : X) :=
  match E with
  | EnvCtr te ve me mte =>
      EnvCtr (PermPlainEnv.action p te)
             (PermPlainEnv.action p ve)
             (action_me p me) (action_mte p mte)
  end
with action_me p me :=
  match me with
  | MEnvCtr {| v_size := nn;  keys := ks; vals := vs |} =>
    MEnvCtr {| v_size := nn;  keys := ks;
               vals := (map (action_mod p) vs) |}
   end
with action_mte (p : Perm) mte:=
  match mte with
  | MTEnvCtr {| v_size := nn;  keys := ks; vals := vs |} =>
    MTEnvCtr {| v_size := nn;  keys := ks;
                vals := (map (action_mty p) vs) |}
  end
with action_mod p (md : Mod) : Mod :=
   match md with
   | NonParamMod e => NonParamMod (action p e)
   | Ftor ts e mty => Ftor (PFin.action p ts)
                           (action p e)
                           (action_mty p mty)
   end
with action_mty p (mty : MTy) : MTy :=
  match mty with
  | MSigma ts m => MSigma (PFin.action p ts) (action_mod p m)
  end.
\end{lstlisting}
The definition above uses previously constructed nominal sets for ``flat''
environments \icode{PermPlainEnv}, finite sets \icode{PFin}, and the \icode{map}
function on vectors to apply actions \icode{action_mod} and \icode{action_mty} to all
values in the corresponding environments.

The implementation of \icode{NominalSet} for the interpretation environments
follows the same pattern, since they have a structure similar to semantic
objects. Some components of the interpretation environments definition includes
pieces of semantic objects, such as \icode{Env}, \icode{MTEnv}, \icode{MTy}, and
we use respective functions from the nominal set of semantic objects in the
definition of the action and the support for interpretation environments.

We define the $\alpha$-equivalence relation on semantic objects using a
similar approach as in Definition \ref{eq:modules:MTy-alpha-equiv-perm} in Section
\ref{sec:modules:var-binding-nominal}, following the mutual inductive structure
of semantic objects.

We use the following notation for the difference operation on sets, the action of a
permutation on finite sets, and the action of a permutation on modules.
\begin{lstlisting}
Infix ":-:" := Atom.V.diff (at level 40).
Notation "r @ x" := (PFin.action r x) (at level 80) : set_scope.
Notation "r @ x" :=
    (PermSemOb.action_mod r x) (at level 80) : env_scope.

Delimit Scope set_scope with S.
Delimit Scope env_scope with E.
\end{lstlisting}

The $\alpha$-equivalence relation is defined as follows.
\begin{lstlisting}
Inductive ae_env : Env -> Env -> Prop :=
| ae_env_c : forall (ve' ve : VEnv) (te' te : TEnv)
                    (me' me : MEnv) (mte' mte : MTEnv),
    ve' = ve ->
    te' = te ->
    ae_menv me' me ->
    ae_mte mte' mte ->
    ae_env (EnvCtr te' ve' me' mte') (EnvCtr te ve me mte)
  with
  ae_menv : MEnv -> MEnv -> Prop :=
  | ae_menv_c : forall (me' me : VE.VecEnv Mod),
      (forall mid (e' e : Mod),
          look mid (_to me') = Some e' ->
          look mid (_to me) = Some e ->
          ae_mod e' e) ->
      ae_menv (MEnvCtr me') (MEnvCtr me)
  with
  ae_mte : MTEnv -> MTEnv -> Prop :=
  | ae_mte_c : forall (mte' mte : VE.VecEnv MTy),
      (forall mtid (e' e : MTy),
          look mtid (_to mte') = Some e' ->
          look mtid (_to mte) = Some e ->
          ae_mty e' e) ->
      ae_mte (MTEnvCtr mte') (MTEnvCtr mte)
  with
  ae_mod : Mod -> Mod -> Prop :=
  | ae_mod_np : forall e' e,
      ae_env e' e -> ae_mod (NonParamMod e') (NonParamMod e)
  | ae_mod_ftor : forall t e e' mty mty',
      ae_env e e' ->
      ae_mty mty mty' ->
      ae_mod (Ftor t e' mty') (Ftor t e mty)
 with
 ae_mty : MTy -> MTy -> Prop :=
 | ae_mty_c : forall m m',
     forall (T T' : Atom.V.t) p,
       (forall a, Atom.V.In a ((PermSemOb.supp_mod m) :-: T)
                  -> (perm p) a = a) ->
    ae_mod m' (p @ m)%E ->
    T' = (p @ T)%S ->
    ae_mty (MSigma T' m') (MSigma T m).

\end{lstlisting}

The essential part of this definition is the \icode{ae_mty} case. It allows for relating
module types up to permutations that affect only variables distinct form the set \icode{T}.

Let us show an example, where explicit $\alpha$-equivalence is needed
in order to elaborate a module declaration.

\begin{example}\label{ex:modules:ftor-app-alpha-equiv}
Let $F = \forall \emptyset.(\{\},\exists \{x\}.\{a \mapsto x\})$ be a functor
and $E = \{f \mapsto F\}$ be an environment containing this functor.
We want to elaborate a sequence of module declarations
\[ \kw{module}~m_1 = f(\eps);\kw{module}~m_2 = f(\eps) \]
in the environment $E$ into the following module type
\[ \exists \{x,y\}.\{m_1 \mapsto \{a \mapsto x\}, \{m_2 \mapsto \{a \mapsto y\}\} \]
Two different functor applications give us results that could differ in
names of bound variables. According to Rule
\Ref{rule:modules:elab-mdec-seq} we have to build two derivations (side
conditions are trivially satisfied in this case). We start with the first one.
\[ \infer{
  \infer{E |- \eps : \{\} \\ E(f) = (\{\},\exists \{x\}.\{a \mapsto x\})} {E |-
    f(\eps) : \exists \{x\}.\{a \mapsto x\}}} {E |- \kw{module}~m_1 = f(\eps) :
  \exists \{x\}.\{m_1 \mapsto \{a \mapsto x\}\}} \] The derivation for the
second declaration in the sequence is not very different. The key observation
here is that we can only derive the following:
\[ E + \{ a \mapsto x \} |- \kw{module}~m_2 = f(\eps) :
  \exists \{x\}.\{m_2 \mapsto \{a \mapsto x\}\} \]

The reason for this is that we are looking up $f$ in the environment
$E + \{ a \mapsto x \}$, which gives the same result as in the environment $E$.
Now, we have to apply $\alpha$-renaming. Otherwise it would be impossible to satisfy
the condition $T_1 \cap (\textrm{tvs}(E) \cup T_2) = \emptyset$ in
\Ref{rule:modules:elab-mdec-seq}. That is, we have to rename $x$ in the second
derivation.
\end{example}

The possibility of $\alpha$-renaming as illustrated in Example
\ref{ex:modules:ftor-app-alpha-equiv} is usually implicit on paper, but in the
Coq formalisation, we have to include it explicitly in the rule. Rule
\Ref{rule:modules:elab-mdec-seq} becomes the following.
\[
\fraccn{\label{rule:elab-mdec-seq-alpha}
  E \vd \id{mdec}_1 : \exists T_1.E_1 \LSP E + E_1 \vd \id{mdec}_2 : \exists T_2.E_2 \\
  \exists T_1.E_1 =_\alpha \exists T_1'.E_1' \LSP  \exists T_2.E_2 =_\alpha \exists T_2'.E_2' \\
  T_1' \# T_2' \LSP T_1' \# \textrm{tvs}(E) \\}
       {E \vd \id{mdec}_1 ~\id{mdec}_2 : \exists(T_1' \cup T_2').(E_1' + E_2')}
\]
Notice, that we also express the side condition on variable sets disjointness
using the freshness relation. Reflecting these changes to our implementation
allows us to build a derivation for Example
\ref{ex:modules:ftor-app-alpha-equiv}.

\begin{remark}\label{rem:modules:alpha-equiv}
  We have formalised the proof of Theorem \ref{norm.prop} in a simplified
  setting, by taking sets of variables in binding positions to be empty. For
  sequencing rules (rules \Ref{rule:modules:elab-mdec-seq} and
  \Ref{elab.spec.seq.rule}), addition of $\alpha$-equivalence is not required
  for the proof of Theorem \ref{norm.prop}, since we \emph{assume} elaborate
  modules. However, to \emph{build} a derivation for static interpretation we
  must be able to $\alpha$-rename sets of variables $N$ and $N'$ in Rule
  \Ref{rule:modules-sint-seq} appropriately. In order to achieve this, we can
  add additional premises allowing for $\alpha$-rename such as those in Rule
  \Ref{rule:elab-mdec-seq-alpha}.  The same applies to Lemma
  \ref{lem:modules:term-dec-comp} (see Rule \Ref{rule:modules:comp-decl}).~

  Currently, the proof of Theorem \ref{norm.prop} ignores issues related to
  $\alpha$-conversion. However, we have an implementation of the rules with
  explicit $\alpha$-equivalence demonstrating Example
  \ref{ex:modules:ftor-app-alpha-equiv}; these changes are not yet incorporated
  into the proof of Theorem \ref{norm.prop}.
\end{remark}

\subsection{Proof of Normalisation of Static Interpretation}
\label{subsec:modules:coq-norm-static-int}
The proof of static interpretation normalisation is carried out
as it is described in Section \ref{subset:modules:static-int-norm}.
The logical relation is implemented as a fixpoint rather than as an inductive
relation. The reason for this representation is essential. If the relation was
defined as an inductive predicate, the definition would not pass the strict
positivity constraint for inductive definitions in Coq. From the definition of
our logical relation, it is straightforward to establish that the relation is
well-formed because it is decreasing structurally in its left argument. For
this reason, it can be expressed as a fixpoint definition, using also Coq's
anonymous fix-construct, corresponding to the nested structure of the semantic
objects. Unfortunately, we cannot keep our environment representation
completely abstract, since we define the logical relation recursively on the
structure of environments. Restrictions on fixpoint definitions in Coq require
us to use a nested fixpoint on underlying structures in the definition of
environments. Again, we use a pair-of-vectors view to define a corresponding
nested fixpoint in the definition of the consistency relation of (Figure
\ref{fig:modules:consistency}).
\begin{lstlisting}
Fixpoint consistent_IEnv (E:Env) (IE:IEnv) : Prop :=
  match E, IE with
    EnvCtr TE VE ME MTE, IEnvCtr TE' IVE IME MTE' =>
    TE = TE'
    /\ consistent_IVEnv VE IVE
    /\ consistent_IMEnv ME IME
    /\ MTE = MTE'
  end
with consistent_IMEnv (ME:MEnv) (IME:IMEnv) : Prop :=
  match ME, IME with
    MEnvCtr me, IMEnvCtr ime =>
    dom (SemObjects.VE._to me) = dom (SemObjects.VE._to ime) /\
    match me,ime with
      VecEnv.mkVecEnv _  nn (exist _ ks _) vs,
      VecEnv.mkVecEnv _  nn' (exist _ ks' _) vs' =>
      (fix con_step {n n'} (l : Vec Mod n) (ll : Vec IMod n')
         : Prop :=
           match l,ll with
           | [],[] => True
           | m :: tl, im :: tl'  =>
             con_step tl tl' /\  consistent_IMod m im
           | _,_ => False
           end) nn nn' vs vs'
    end
  end
with consistent_IMod (M:Mod) (IM:IMod) : Prop :=
  match M with
    NonParamMod E =>
    match IM with
      INonParamMod IE => consistent_IEnv E IE
    | IFtor _ _ _ _ _ _ => False
    end
  | Ftor T0 E (MSigma T M) =>
    exists IE0 mid mexp,
    IM = IFtor IE0 T0 E (MSigma T M) mid mexp
    /\ forall IE,
       consistent_IEnv E IE ->
        exists N IM c,
         (Mexp_int (addIEnvMid mid (INonParamMod IE) IE0) mexp N IM c
         /\ consistent_IMod M IM)
  end.
\end{lstlisting}
There are several design decisions that we want to emphasise. Instead of
forcing two vectors to be of the same length in the \icode{consistent_IMEnv}
function, we just return False in case vectors are not aligned. This definition also
makes use of a \emph{concrete} representation of module type environments as a
pair-of-vectors.  We would not be able to use any \emph{abstract}
representation of environments (like a type constructor exposed through the
module type) in the definition. The reason for that is that we would have to
use a recursor provided by the abstract representation to define the inner
fixpoint and Coq would not accept the definition as terminating.

\begin{remark}\label{rem:modules:nested-fix}
  Instead of the inner fixpoint \icode{con_step} we would like to use a
  separately defined function stating that some predicate
  \icode{(P : A -> B -> Prop)} holds for two vectors componentwise.
  The function is defined as follows:
  \begin{lstlisting}
  Fixpoint Forall2_fix {A B} (P : A -> B -> Prop) (n n' : nat)
               (l : Vec A n) (ll : Vec B n') : Prop :=
        match l,ll with
        | [],[] => True
        | m :: tl, im :: tl'  =>
          Forall2_fix P _ _ tl tl'  /\  P m im
        | _,_ => False
        end.
  \end{lstlisting}
  Unfortunately, this does not work for our definition. Probably, Coq does not
  unfold \icode{Forall2_fix} here to see that the argument is
  decreasing. Although, sometimes Coq is able to do some unfoldings, when
  checking a definition for termination (see definition of \icode{ntsize} in
  \cite[Section 2.8]{cpdt} and the definition of \icode{action} in Section
  \ref{subsec:modules:nominal-in-coq}).
\end{remark}

Although we cannot use \icode{Forall2_fix} directly in the definition of the
consistency relation, we can prove that the nested fixpoint \icode{con_step}
corresponds to \icode{Forall2_fix}. This way we can use properties of
\icode{Forall2_fix} in the proofs related to the consistency
relation. Particularly, we are interested in converting this ``intensional''
representation to the extensional one that relates two environments using
environment membership:
\begin{lstlisting}
  Definition EnvRel {A B} (P : A -> B -> Prop)
                          (E : En.t A) (E' : En.t B) : Prop :=
    (forall k, In _ k E <-> In _ k E')
    /\ (forall k v v', look k E = Some v ->
                       look k E' = Some v' ->
                       P v v').
\end{lstlisting}
In most proofs we use properties of environments, which are stated in terms of
\icode{look}, and it is very unwieldy to use \icode{Forall2_fix} in such proofs,
since in this case proofs have to be carried out by induction on the structure
of the underlying vector environment.

Putting it all together, we define equations allowing us to fold the nested
fixpoint in the definition of the consistency relation:
\begin{lstlisting}
Lemma Forall2_fix_fold_unfold {A B} (P : A -> B -> Prop) :
        Forall2_fix P =
        (fix ff {n n'} (l : Vec A n) (ll : Vec B n') : Prop :=
           match l,ll with
           | [],[] => True
           | v :: tl, v' :: tl'  =>
             ff tl tl' /\ P v v'
           | _,_ => False
           end).

Lemma ForallEnv2_fold_unfold {A B} (P : A -> B -> Prop) ve ve' :
        ForallEnv2_fix P ve ve'=
        match ve,ve' with
          VecEnv.mkVecEnv _  n (exist _ ks _) vs,
          VecEnv.mkVecEnv _  n' (exist _ ks' _) vs' =>
             Forall2_fix P n n' vs vs'
        end.
\end{lstlisting}

We also provide a logical equivalence between \icode{Forall2_fix} and \icode{EnvRel}:
\begin{lstlisting}
Lemma ForallEnv2_fix_EnvRel_iff (A B : Type)
                               (P : A -> B -> Prop)
                               (ve : VecEnv A) (ve' : VecEnv B)
 : (dom ve = dom ve' /\ ForallEnv2_fix P ve ve') <-> EnvRel P ve ve'.
\end{lstlisting}

\noindent In proofs involving the consistency relation we use the following pattern:
\begin{itemize}
  \item rewrite using fold/unfold lemmas;
  \item rewrite by \icode{ForallEnv2_fix_EnvRel_iff}
\end{itemize}
In this way we bridge the gap between the ``intensional'' recursive definition and
the ``extensional'' definition, allowing for more convenient reasoning.

We have kept the Coq development close to the representation in this chapter.
However, the filtering relation in Figure \ref{filtering.fig} does
not define a filtering algorithm directly, but serves as a \emph{specification}
for it. In the proof of normalisation of static interpretation, we have to show
the existence of a filtered environment. Due to the limitations of Coq's
fixpoint constructs, we have defined a filtering algorithm as an inductively
defined relation. We consider it future work to investigate the use of general
recursion in Coq for filtering, which would be useful for applying code
extraction to obtain a certified static interpretation implementation. We
believe it is a reasonable approach to separate the relational
``declarative'' definitions from definitions that compute for code
extraction. One can then establish a correspondence between the relational and
functional representations to show soundness of the implementation.

\section{Related Work}
The concept of static interpretation of modules is not new and has been applied
earlier in the context of the MLKit Standard ML compiler \cite{elsman99}. In
the present work we focus on the formalisation of a higher-order module language in
the Coq proof assistant in the style of \cite{elsman99}. For that reason, as
related work we mention mostly alternative approaches at providing mechanised
meta-theories for module languages.

The work on compilation of higher-order modules into F$_\omega$, the
higher-order polymorphic lambda calculus \cite{Rossberg:f-ing-modules},
comes with a Coq implementation of the work. Compared to our work, which
eliminates all module language constructs at compile time,
\cite{Rossberg:f-ing-modules} make no distinction between core
language and module language constructs in the target code. Moreover, the style
of formalisation is different from our Coq development since we use a more direct
encoding of the module language semantics in term of semantic objects.

Earlier works include the work on using Twelf to provide a mechanised
meta-theory for Standard ML \cite{Lee:2007:TMM:1190216.1190245}, based on
Harper-Stone semantics of Standard ML \cite{Harper:2000:TIS:345868.345906}.
Compared to our work, however, this approach also does not address the
problem of eliminating modules at compile time.

Another body of work related to mechanising the meta-theory of ML is the work
on CakeML \cite{Tan:2016:NVC:2951913.2951924}, which, however, supports only
non-parameterised modules.

The category of works related to our proof techniques and the approach to the
formalisation includes works on reasoning with isomorphic representations of
abstract data types in the context of homotopy type theory
\cite{Licata:isomorphic-views,GabeDijkstra:msthesis}. Nominal techniques are
implemented in the Nominal Isabelle package \cite{Urban2005,Urban2011} and to a
limited extend in Coq (mostly focusing on simply-typed lambda calculus)
\cite{Aydemir:2007}.

\section{Conclusion and Future Work}
We have developed a formalisation in Coq of a higher-order module system along
with the static interpretation with the guarantee of termination in the style
of \cite{elsman99}. Our implementation is one of the first attempts to
formalise a module system in this style in the Coq proof assistant.  In the
course of the implementation we have developed the following techniques.
\begin{itemize}
\item Extension of the standard library implementation of sets and environments
  (finite maps) with extensionality property.
\item Isomorphic representations of environments (finite maps) to overcome the
  issue with the conservative strict positivity check in Coq.  The technique
  developed allowed us to implement semantic objects in Coq with almost no proof
  obligation overhead despite the fact of using different environment
  representations.
\item Formalisation of nominal sets and applications of generalised nominal
  techniques allowing to work with sets of variables in binding positions. This
  is the first implementation in Coq dealing with generalised binders.
\item The normalisation proof of the static interpretation using the Tait's
  method of logical relations \cite{tait1967} in the setting of higher-order
  modules.
\end{itemize}

The current version of our Coq development is about 6.5k lines of code,
excluding comments and blank lines. It includes definitions from Section
\ref{sec:modules-spec}, the proof of Theorem \ref{norm.prop}, the module
implementing nominal techniques (with examples in the simplified setting discussed in
Section \ref{sec:modules:var-binding-nominal}), the module implementing a pair-of-vectors
representation of environments, and the proof of strong normalisation for the
simply-typed lambda calculus.

Although our implementation makes some simplifying assumptions, we believe it
is can be extended to a full setting with no fundamental
limitations. Particularly, the nominal techniques gives a uniform structuring
principle for dealing with binders.

Moreover, our Coq implementation and the formal specification given in Section
\ref{sec:modules-spec} has been developed hand-in-hand with the Haskell
implementation integrated with the Futhark compiler, serving as a guiding line
for the development. Using semantic objects allows for the structure of the
Haskell implementation to be in a close correspondence with our Coq implementation.

As future work, we would like to extend our implementation of nominal
sets with more features and eventually expose it as a
library. Regarding the implementation of nominal techniques, we would
like to note that a solution making use of type classes instead of
modules to structure the library would be beneficial. Ideally, such an
implementation would require finite sets to be implemented using type
classes as well, which is not the case in the standard library of Coq,
where they are implemented using modules.\footnote{An experimental standalone implementation
  that uses Coq's type classes to implement nominal sets has been developed by the author,
  and it is available online: \url{https://github.com/annenkov/stlcnorm}.
  This implementation, however, still uses the module-based finite sets implementation from the standard library.}
We believe that such an implementation would allow for better proof automation, especially
for proving properties such as equivariance.

Another direction of extension of our development in Coq would be an
implementation of \emph{algorithms} corresponding to the relational specifications
of elaboration, filtering, and eventually, static interpretation. Having such
implementations, one could use Coq's code extraction mechanism to obtain
certified code in one of the target languages, which could be used as part of a
compiler for the Futhark programming language.

\chapter{Formalising Two-Level Type Theory}\label{chpt:tltt}

\section{Introduction}
Homotopy Type Theory (HoTT) is a variant of Martin-L\"of type theory that pays
particular attention to the equality (or identity) type.  The equality type in
Martin-L\"of type theory is defined as an inductive family generated by the
single constructor \icode{refl}, called reflexivity.  For any two inhabitants
$a$ and $b$ of some type $A$, one can ask if $a$ and $b$ are equal by forming the
equality type $a=b$. Since $a = b$ is also a type, one can ask if two proofs of
equality are equal, i.e. for $p, q : a=b$, one can form $h : p = q$, and so on.
The eliminator of the equality (or identity) types is usually called $J$ (see rule
\nameref{rule:elim-eq} in Section \ref{sec:ttlt-spec}).  As it was
observed by Hofmann and Streicher in \cite{hofmann-streicher:groupoid}, it is
not possible from $J$ to show that two proofs of equality are equal.

Therefore, there are two options: one could add axiom K
(or equivalently, Uniqueness of Identity Proofs axiom),
making any two proofs of equality equal, i.e.
\[\inferrule{\Gamma \vdash a_1, a_2 : A \\
  \Gamma \vdash p, q : a_1 = a_2}
{\Gamma \vdash K(p, q) : p = q}
\quad\deflabel{\textsc{uip}}\label{rule:uip}
\]

The other option is to allow for different ways to identify types by introducing
the \emph{univalence axiom} \cite{hott-book}. The univalence axiom says that
two types are considered equal when they are equivalent, reflecting the informal
principle that is usually used in mathematics. That is, for any two types $A$ and $B$
\[ \id{Univalence} : (A = B) \simeq (A \simeq B) \]
More precisely, is says that a function
$\id{idtoequiv} (A = B) -> (A \simeq B)$, which can be defined by induction
on equality, is an equivalence.  In other words there is another function,
$ua : (A \simeq B) -> (A =B)$ that goes in the opposite direction and allows one
to get proofs of equality from proofs of equivalence (see more on definition of
equivalences in \cite[Chapter 4]{hott-book}).

Types in HoTT are weak $\infty$-groupoids, which allows for capturing important
notions from homotopy theory and for developing it synthetically in type theory.

Homotopy type theory is implemented in a number of proof assistants allowing
for development of machine-checkable proofs in such areas as homotopy theory,
category theory and other areas of mathematics. From the functional programming
point of view, it also allows for solving abstraction problems in dependently
typed programming, since equivalent types can be identified and changing
between equivalent representation does not require additional efforts.

If we restrict ourselves to just sets (types, for which two proofs of equality
of two elements are equal, also called hSets), the notion of equivalence will
be just an isomorphism between types. In this setting, one can benefit from the
fact that isomorphic types become equal in the presence of the univalence
axiom. That means that all the proofs of properties of some type $A$ are
immediately transferred to types isomorphic to $A$, since we
can always substitute equals for equals. This approach allows for better
abstraction preservation in proof and program development. Particularly, one
could have an abstract representation of some structure, for example
environments (or finite maps), as we saw in Chapter \ref{chpt:modules}. One can
define several isomorphic representations satisfying the abstract specification
and move between these representations using the fact that isomorphic types are
considered equal. These ideas are considered in
\cite{Licata:isomorphic-views,GabeDijkstra:msthesis}.

In our development of the Futhark module system formalisation described in
Chapter \ref{chpt:modules} we could have used univalence to switch between the
two environment representations easily, while keeping the specification
completely abstract. Such a possibility would allow us to take advantage of
computational behaviours of concrete representations, since it is impossible
to compute with the specification given by abstract type or opaque module.

Another interesting feature of HoTT are higher inductive types (HITs). For
instance, HITs allow for avoiding so called ``setoid hell'' by adding equality
constructors to inductive definitions. A setoid is a set equipped with an equivalence
relation.  Whenever we want to compare two elements of setoid we have to use
this custom equivalence relation instead of usual equality. Because of that all
the operations on setoids have to respect its equivalence, which require a lot of
proofs. Using HITs (to be precise, a specialised version called quotient
inductive types, or QITs), one can use the usual propositional equality again
when comparing two elements of a setoid.

Homotopy type theory is an active developing field with a number of open
questions.  Particularly, not being able to talk about strict equality in HoTT
sometimes make certain constructions problematic. We will address this
problem in the following sections.

The rest of the chapter is structured as follows. In Section \ref{sec:tltt} we
discuss motivations for two-level type theory (2LTT). In Section
\ref{sec:ttlt-spec} we provide a formal specification of 2LTT and discus
differences with homotopy type system (HTS). Section \ref{tltt:sec:applications}
describes an internalisation of some results on inverse diagrams in 2LTT. The
implementation of 2LTT in the Lean proof assistant is discussed in Section
\ref{tltt:sec:lean-formalisation}, which outlines the overall idea of the
approach to the implementation of 2LTT in a proof assistant and then demonstrates
features of our Lean development including the results from Section
\ref{tltt:sec:applications}. Section \ref{tltt:sec:lean-formalisation}
represents the main contribution of the author to the development of two-level
type theory.

\section{Motivation}\label{sec:tltt}
The motivation for two-level type theory is twofold.

Many results of homotopy type theory are completely internal to HoTT and can be
formalised directly in a proof assistant, and a lot of work has been done using
Agda, Coq, and Lean.  Some other results are only \emph{partially} internal to
HoTT.  One example is the constructions of $n$-restricted semi-simplicial types
which we can do only after fixing the number $n$ \emph{externally} (i.e.\ we have
to decide which $n$ we use before we start writing it down in a proof
assistant).  The reason, why such constructions are problematic to write in HoTT
is that the usual definition of $n$-restricted semi-simplicial type as a
\emph{strict} functor from the category of finite non-empty ordinals and
strictly monotone maps to the category of types: $S : \deltop -> \UU$ require
functor laws to hold strictly (the functor laws correspond to semi-simplicial
identities). In HoTT it would require infinite tower of coherencies ensuring
that certain proofs of equality are equal, proofs of equality of proofs of
equality are equal, and so on.

One can try to avoid writing equalities corresponding to the semi-simplicial
identities by using an equivalent representation of $n$-restricted
semi-simplicial types as a nested $\Sigma$-type with face maps being
projections. For example, if we fix $n = 3$, we can write the following
definition (we use Lean \cite{moura:lean} notation here):
\lstset{language=lean}
\begin{lstlisting}
definition SST₃ :=
  Σ (X₀ : Type)
    (X₁ : X₀ → X₀ → Type),
       Π (x₀ x₁ x₂ : X₀), X₁ x₀ x₁ → X₁ x₁ x₂ → X₁ x₀ x₁ → Type
\end{lstlisting}
In the definition for \icode{SST₃} we think of \icode{X₀} as points, \icode{X₁}
as lines, and
\icode{$\Pi$ (x₀ x₁ x₂ : X₀), X₁ x₀ x₁ → X₁ x₁ x₂ → X₁ x₀ x₁ → Type} as triangles.
To form a triangle, we need for the end of one side and
the beginning of another side to be the same point.  Instead of using
semi-simplicial identities here (formulated using equalities), we just use the
same point again for the point that should match the given point. For example
\icode{X₁ x₀ x₁} is the line that starts at \icode{x₀} and ends at \icode{x₁},
and \icode{X₁ x₁ x₂} starts at \icode{x₁} and ends at \icode{x₂}.

Although it is possible to write such a definition for any fixed $n$, it seems
to be not possible to generalise it for an arbitrary $n$ internally
in HoTT. For the detailed explanation see the introduction section in
\cite{alt-cap-kra:two-level}.

Another example of this kind is the work by Shulman on inverse
diagrams~\cite{shulman:inverse-diagrams}, for which we can do constructions in type
theory once we fix a (finite) inverse category \emph{in the meta-theory}. One of
the examples of the inverse diagrams is $n$-restricted semi-simplicial types
described above.

In many situations, one would like such constructions to be completely internal
(using a variable $n : \mathbb{N}$ or an inverse category expressed internally)
and formalisable in a proof assistant, but unfortunately, it is either unknown
how this is doable or it is known to be impossible.
Two-level type theory gives a way to completely formalise such results.
This is the aspect that we explore in the paper~\cite{ann-cap-kra:two-level}.

A second motivation of two-level type theory is that it allows for extending homotopy
type theory in a ``controlled'' way.
It gives a framework that makes it easy to write down enhancements of the theory,
where one can relatively easily check whether these assumptions hold in some models
(models are explored in~\cite{paolo:thesis}).

In the present work we are focusing on the first motivation, namely
on internalisation of results that can only be partially internalised in HoTT.
We also will demonstrate how one can implement two-level type theory in a proof assistant
and discuss our experience with developing a formalisation in such an implementation.

\section{Two-Level Type Theory}\label{sec:ttlt-spec}
To address the issues arising from the lack of strict equality, presented in
Section \ref{sec:tltt}, we introduce two level type theory, which consists of
two fragments: a \emph{strict} fragment (a form of \MLTT with \UIP) and a
\emph{fibrant} fragment (essentially \HOTT).  The fibrant fragment of our type
theory has all the basic types and type formers found in \HOTT \cite[Appendix
  A.2]{hott-book}:
\begin{itemize}
\item $\Pi$, the type former of dependent functions;
\item $\Sigma$, the type former of dependent pairs;
\item $+$, the coproduct type former;
\item $\unit$, the unit type;
\item $\emptyt$, the empty type;
\item $\N$, the fibrant type of natural numbers;
\item $=$, the equality type (in the sense of $\HOTT$);
\item a hierarchy $\UU_0, \UU_1, \ldots$ of universes;
\item possibly inductive and higher inductive types.
\end{itemize}
Furthermore, we have:
\begin{itemize}
\item $\strict +$, the strict coproduct;
\item $\strict \emptyt$, the strict empty pretype;
\item $\strictN$, the strict pretype of natural numbers;
\item $\steq$, the strict equality;
\item a hierarchy $\strict{\UU_0}, \strict{\UU_1}, \ldots$ of strict universes;
\item possibly inductive types and quotient types.
\end{itemize}

We refer to the elements of $\UU_i$ as \emph{fibrant types}, while the elements
of $\strict{\UU_i}$ are \emph{pretypes}. The intuition is that fibrant types are
the usual types that are found in $\HOTT$, whereas pretypes are what one gets if
one is allowed to talk about strict equality internally.
The rules of $\steq$ :
\begin{mathpar}\label{eq:steq}
\inferrule{\Gamma \vdash A : \strict \UU_i \\
  \Gamma \vdash a, b : A}
{\Gamma \vdash a \steq b : \strict \UU_i}
\quad\deflabel{\textsc{form-$\steq$}}
\label{rule:steq-form}

\and

\inferrule{\Gamma \vdash a : A}
{\Gamma \vdash \strict{\refl{a}} : a \steq a}
\quad\deflabel{\textsc{intro-$\steq$}}
\label{rule:steq-intro}

\and

\inferrule{\Gamma \vdash a : A \\
  \Gamma(b : A)(p : a \steq b) \vdash P : \strict \UU_i \\
  \Gamma \vdash d : P[a, \strict{\refl a}]}
{\Gamma(b : a)(p : a \steq b) \vdash \strict J_P(d): P}
\quad\deflabel{\textsc{elim-$\steq$}}\label{rule:steq-elim},
\end{mathpar}
together with the judgmental computation rule:
\begin{equation*}
\strict J_P(d) [a, \strict{\refl a}] \equiv d.
\end{equation*}

Rules for fibrant equality look very similar to those for strict equality:
\begin{mathpar}
  \infer{\Gamma \vdash A : \UU_i \\ \Gamma \vdash a_1,a_2 : A}
          {\Gamma \vdash a_1 = a_2 : \UU_i}
          \quad\deflabel{\textsc{form-=}}\label{rule:form-eq}
  \and%
  \inferrule{\Gamma \vdash a : A}
            {\Gamma \vdash \refl{a} : a = a}
  \quad\deflabel{\textsc{intro-=}}\label{rule:intro-eq}
  \and%
  \infer{\Gamma \vdash a : A \\
    \Gamma(b : A)(p : a = b) \vdash P : \UU_i \\
    \Gamma \vdash d : P[a, \refl a]}
        {\Gamma(b : a)(p : a = b) \vdash J_P(d): P}
    \quad\deflabel{\textsc{elim-=}}\label{rule:elim-eq}
  \end{mathpar}
It is important to note, that the rules \nameref{rule:form-eq},
\nameref{rule:intro-eq}, and \nameref{rule:elim-eq} only involve fibrant
types. For example, we cannot apply the equality type former to two elements of
$\strict \UU_i$.  We assume that universes $\UU_i$ are univalent, that is for
any two types $X,Y : \UU_i$, the map $(X = Y) -> (X \simeq Y)$ is an equivalence.

\begin{mathpar}

\infer{\Gamma \vdash A : \UU_i}
{\Gamma \vdash A : \strict{\UU_i}}
\quad\deflabel{\textsc{fib-pre}} \label{rule:fib-pre}%
\and%
\inferrule{\Gamma \vdash A : \UU_i \\
  \Gamma.A \vdash B : \UU_i}
{\Gamma \vdash \prd A B : \UU_i}
\quad\deflabel{\textsc{pi-fib}}\label{rule:pi-fib}%
\and%
\infer{\Gamma \vdash A : \UU_i \\ \Gamma.A \vdash B : \UU_i}
{\Gamma \vdash \smsimple A B : \UU_i}
\quad\deflabel{\textsc{sigma-fib}}\label{rule:sigma-fib}
\end{mathpar}

By the rule \nameref{rule:fib-pre}, the type $a_1 = a_2$ is also a pretype, but
it is different from the pretype $a_1 \steq a_2$.

For pretypes $A,B : \strict{\UU_i}$, we can form the pretype of strict isomorphisms,
written $A \stiso B$ (unlike in $\HOTT$, it is enough to have maps in both
directions such that both compositions are pointwise strictly equal to the identity).
However, we do \emph{not} assume that $\strict{\UU_i}$ is univalent.
Instead, we add rules corresponding to the principle of uniqueness of identity proofs~
\nameref{rule:uip} and function extensionality~\nameref{rule:funext} as follows:
\begin{mathpar}
\inferrule{\Gamma \vdash a_1, a_2 : A \\
  \Gamma \vdash p, q : a_1 \steq a_2}
{\Gamma \vdash \strict K(p, q) : p \steq q}
\quad\deflabel{\textsc{uip}}\label{rule:uip}%
\and%
\inferrule{\Gamma \vdash f, g : \prd{a : A} B(a) \\
  \Gamma(a : A) \vdash p(a) : f(a) \steq g(a)}
{\Gamma \vdash \mathsf{funext}(p): f \steq g}
\quad\deflabel{\textsc{funext}}\label{rule:funext}%
\end{mathpar}

That is, strict equality serves as an internalised version of judgmental equality.
These ideas of differentiating pretypes and fibrant types are inspired by Voevodsky's
\emph{Homotopy Type System} (HTS)~\cite{voe:hts}.
Although, there are some important differences. In particular, two-level type
theory does not assume the reflection rule for strict equality. Instead, we
only require that it satisfies \UIP.
Another important difference is that HTS assumes that $\emptyt$, $\N$, and $+$
from the fibrant fragment eliminate to arbitrary types. We do not assume that in
two-level type theory, since all presented results do not depend on these assumptions.
Moreover, we leave a possibility to add such assumptions, which makes
two-level type theory a ``framework'' allowing to explore different variations
of resulting type theory. For example, if we allowed for coercion from strict
natural numbers $\strict \N$ to fibrant natural numbers $\N$ then from the construction
of type of the $n$-restricted semi-simplicial types, one would get a type family $\N -> \UU$
in the fibrant fragment.

\section{Applications}\label{tltt:sec:applications}
One way to look at two-level type theory is to start with ordinary type theory,
which correspond to the fibrant fragment, and then add parts of its meta-theory
on top of it as an additional type-theoretic layer. This additional layer
corresponds to the strict fragment, which can be used to capture meta-theoretic
reasoning. This internalisation leads to a uniform treatment of results, which
traditionally requires mixing external and internal reasoning.
We show applicability of two-level type theory to these kinds of problems
internalising some results on Reedy fibrant diagrams \cite{shulman:inverse-diagrams}.
Particularly, we define the notion of Reedy fibration, and show that Reedy fibrant
diagrams $I \to \UU$ have limits in $\UU$, where $I$ is a finite inverse category,
and $\UU$ is a universe of fibrant types.
Let us first define notions required to formulate the theorem about
Reedy fibrant diagrams, which is one of the central results of our Lean
formalisation (we will discuss implementation details in Section \ref{tltt:sec:lean-inv-diag}).
We closely follow the style of definitions given in \cite{ann-cap-kra:two-level}.

It is often not necessary to know that a pretype $A : \strict \UU$ is a fibrant
type.  Instead, it is usually sufficient to have a fibrant type $B : \UU$ and a
strict isomorphism $A \stiso B$.  If this is the case, we say that $A$ is
\emph{essentially fibrant}. Clearly, every fibrant type is also an essentially
fibrant pretype.

Recall that, in usual type-theoretic terminology, $\Fin_n$ is the finite type
with $n$ elements.  In two-level type theory, for a strict natural number $n :
\strict \N$, we have the finite type $\strict \Fin_n$.  If we have inductive
types in the strict fragment, we can define it as usual, but we do not need to:
we can simply define it as the pretype of strict natural numbers smaller than
$n$.  Similarly, we have a fibrant type $\Fin_n$ (note that a strict natural
number can always be seen as a fibrant natural number, but not vice versa).

We say that a pretype $I$ is \emph{finite} if we have a number $n : \strictN$ and a strict isomorphism $I \stiso \strict {\Fin_n}$.

Similar to essential fibrancy, we have the following definition:

\begin{defn}[fibration; see {\cite[Definition 4]{ann-cap-kra:two-level}}]\label{def:fibration}
Let $p : E \to B$ be a function (with $E, B : \strict \UU$).
We say that $p$ is a \emph{fibration} for all $b : B$, the fibre of $p$ over $b$, i.e.\ the pretype $\sm{e:E} p(e) \steq b$, is essentially fibrant.
\end{defn}

We define the notion of a category in the strict fragment with categorical laws formulated using
strict equality.

\begin{defn}[category; see {\cite[Definition 7]{ann-cap-kra:two-level}}] \label{def:strictcat}
A \emph{strict category} (or simply \emph{category}) $\C$ is given by
\begin{itemize}
 \item a pretype $\obj \C : \strict \UU$ of \emph{objects};
 \item for all pairs $x, y : \obj \C$, a pretype $\C(x, y) : \strict \UU$ of \emph{arrows} or \emph{morphisms};
 \item an \emph{identity} arrow $\mathsf{id} : \C(x, x)$ for every object $x$;
 \item and a \emph{composition} function $\circ : \C(y, z) \to \C(x, y) \to \C(x, z)$ for all objects $x,y,z$;
 \item such that the usual categorical laws holds, that is, we have $f \circ \mathsf{id} \steq f$ and $\mathsf{id} \circ f \steq f$, as well as $h \circ (g \circ f) \steq (h \circ g) \circ f)$.
\end{itemize}

We say that a strict category is \emph{fibrant} if the pretype of objects and
the family of morphisms are fibrant.
\end{defn}
The definition of the strict category corresponds to that of a \emph{precategory} from
\cite[Chapter 9]{hott-book}.

In the following, we will usually drop the attribute ``strict'' and
simply talk about \emph{categories}. A canonical example of a category
is the category of pretypes, whose objects are the pretypes in a given
universe $\strict \UU$, and morphisms are functions. By slight abuse
of notation, we write $\strict \UU$ for this category. The usual
theory of categories can be reproduced in the context of our
categories (at least as long as we stay constructive).  In particular,
one can define the notions of \emph{functor}, \emph{natural
  transformation}, \emph{limits}, \emph{adjunctions} in the obvious
way, or show that limits (if they exist) are unique up to isomorphism,
and so on.

In the context of Reedy fibrations, an important categorical construction is the following one:
\begin{defn}[reduced coslice; see {\cite[Definition 9]{ann-cap-kra:two-level}}]\label{def:red-coslice}
Given a category $\C$ and an object $c : \C$, the \emph{reduced coslice}
$c \sslash \C$ is the full subcategory of non-identity arrows in the coslice
category $c \slash \C$.
A concrete definition is the following.
The objects of $c \sslash \C$ are triples
of an $y : \obj \C$, a morphism $f : \C(x,y)$, and a proof $\neg \left(\trans p
f \steq \mathsf{id} \right)$, for all $p : x \steq y$,
where $\transf p$ denotes the $\mathsf{transport}$ function $\C(x,y) \to \C(y,y)$.
Morphisms between $(a,f,s)$ and $(b,g,s')$ are elements $u : \C(a,b)$ such that $u \circ f \steq g$ in $\C$ (see Figure \ref{fig:tltt:coslice}).

Notice that we have a ``forgetful functor'' $\mathsf{forget} : c \sslash \C \to \C$, given by the first projection on objects as well as on morphisms.
\end{defn}
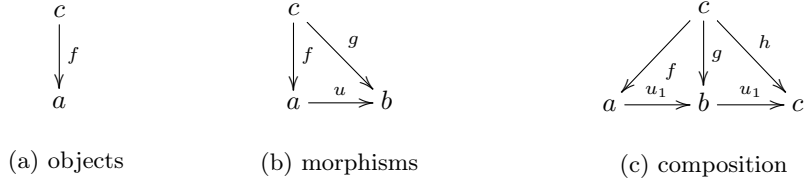
\begin{figure}
  \begin{subfigure}{0.2\textwidth}
    \[
      \xymatrix{
        c \ar[d]^f \\
        a
      }
      \]
      \caption{objects}\label{fig:tltt:coslice-ob}
  \end{subfigure}
  \hspace*{\fill}
  \begin{subfigure}{0.3\textwidth}
    \[
    \xymatrix{
      c \ar[d]^f \ar[dr]^{g} \\
      a \ar[r]^u & b
    }
    \]
    \caption{morphisms}\label{fig:tltt:coslice-hom}
  \end{subfigure}
  \hspace*{\fill}
  \begin{subfigure}{0.4\textwidth}
    \[
    \xymatrix{
      \hole & c \ar[dl]^f \ar[d]^g \ar[dr]^h & \hole \\
      a \ar[r]^{u_1} & \ar[r]^{u_1} b & c
    }
    \]
    \caption{composition}\label{fig:tltt:coslice-comp}
  \end{subfigure}
  \caption{Coslice category.}\label{fig:tltt:coslice}
\end{figure}

Consider the category $\strictNop$ which has $n : \strictN$ as objects, and
\[ \strictNop(n,m) \defeq n \strict \geq m \]
(the function $\strict > : \strictN \to \strictN \to \strict \Prop$ is defined in the canonical way).
Then, we define:
\begin{defn}[inverse category; see {\cite[Definition 10]{ann-cap-kra:two-level}}] \label{def:inverse-category}
 We say that a category $\C$ is an \emph{inverse category}
 if there is a functor $\varphi : \C \to \strictNop$ which reflects identities; i.e.\ if we
 have $f : \C(x,y)$ and $\varphi_x \steq \varphi_y$, then we also have $p : x \steq y$ and $\trans p f \steq \mathsf{id}_y$.
 We call $\varphi$ the \emph{rank functor}, and say that an object $i : \obj \C$ has rank $\varphi(i)$.
\end{defn}

Notice that \emph{reflecting identities} usually means that $f$ is an identity
whenever $\varphi(f)$ is.  In $\strictNop$, a morphism is an identity if and
only if its domain and codomain coincide.  Notice that the expression $f \steq
\mathsf{id}_y$ does not type-check, and to remedy this, we have to transport
$f$ along a strict equality between $x$ and $y$, using the notation $\trans p
f$ from~\cite{hott-book}.

\begin{remark}
  There are several equivalent ways to define inverse categories.  They are
  often characterised as dual to \emph{direct} categories, which in turn can be
  described as not having an infinite sequence of non-identity arrows as in
  $\to \to \to \cdots$. Another way to formulate this following
  \cite{shulman:inverse-diagrams} is ``An inverse category is a category such
  that the relation `x receives a nonidentity arrow from y' on its objects is
  well-founded.'' This formulation allows one to use well-founded induction to
  define diagrams on inverse categories.  One important example of an inverse
  category is $\deltop$, a category of finite non-empty ordinals and strictly monotone maps.
  That is, a functor $S : \deltop -> \UU$ is an inverse diagram on $\deltop$.
\end{remark}

\subsection{Reedy Fibrant Limits}\label{sec:reedy-fibrant-limits}

Recall that our first example of a category was a strict universe $\strict \UU$ of pretypes and functions.
Much of what is known about the category of sets in traditional category theory
holds for $\strict \UU$.
For example, the following result translates rather directly:

\begin{lemma}[see {\cite[Lemma 11]{ann-cap-kra:two-level}}]\label{lem:all-strict-limits}
  The category $\strict \UU$ has all \emph{small} limits, where \emph{small} means
  that the corresponding diagram has an index category whose objects and morphisms are pretypes in $\strict \UU$.
\end{lemma}
\begin{proof}
Let $\C$ be a category with $\obj \C : \strict \UU$ and $\C(x,y) : \strict \UU$ (for all $x,y$).
Let $X : \C \to \strict \UU$ be a functor.
We define $L$ to be the pretype of natural transformations $\mathsf 1 \to X$, where $\mathbf 1 : \C \to \strict \UU$ is the constant functor on $\unit$.
Clearly, $L : \strict \UU$, and a routine verification shows that $L$ satisfies the universal property of the limit of $X$.
\end{proof}

Unfortunately, the category $\UU$ of fibrant types is not as well-behaved.
Even pullbacks of fibrant types are not fibrant in general (but see Lemma~\ref{lem:fibrant-pullback}).
Since $\UU$ is a subcategory of $\strict \UU$,
a functor $X : \C \to \UU$ can always be regarded as a functor $\C \to \strict \UU$, and we always have a limit in $\strict \UU$.
If this limit happens to be essentially fibrant, we say that $X$ has a \emph{fibrant limit}.
Since $\UU$ is a full subcategory of $\strict \UU$, this limit will then be a limit of the original diagram $\C \to \UU$.

\begin{lemma}[see {\cite[Lemma 12]{ann-cap-kra:two-level}}]\label{lem:fibrant-pullback}
The pullback of a fibration $E \to B$ along any function $f : A \to B$ is a fibration.
\end{lemma}
\begin{proof}
We can assume that $E$ is of the form $\sm{b:B} C(b)$ and $p$ is the first projection.
Clearly, the first projection of $\sm{a:A}C(f(a))$
satisfies the universal property of the pullback.
\end{proof}

Lemma~\ref{lem:fibrant-pullback} makes it possible to construct
fibrant limits of certain ``well-behaved'' functors from inverse
categories.

In the subsequent definitions we always assume that $\C$ is an inverse category.

\begin{defn}[matching object; see {\cite[Definition 13]{ann-cap-kra:two-level}} and {\cite[Chapter.~11]{shulman:inverse-diagrams}}]\label{def:matching-object}
Let $X : {\C} \to \strict \UU$ be a functor.
For any $z : \C$, we define the \emph{matching object} $M_z^X$ to be the limit of
the composition
$z \sslash \C \xrightarrow{\mathsf{forget}} \C \xrightarrow{X} \strict \UU$.
\end{defn}

\begin{defn}[Reedy fibrations; see {\cite[Definition 14]{ann-cap-kra:two-level}} and
    {\cite[Def. 11.3]{shulman:inverse-diagrams}}] \label{def:reedy-fibrations}
 Let $X, Y : \C \to \strict \UU$ be two diagrams (functors).
 Further, assume $p : X \to Y$ is a natural transformation.
 We say that $p$ is a \emph{Reedy fibration} if, for all $z : \C$, the canonical map
 \begin{equation*}
  X_z \to M_z^X \times_{M_z^Y} Y_z,
  \end{equation*}
 induced by the universal property of the pullback, is a fibration.

 A diagram $X$ is said to be \emph{Reedy fibrant} if the canonical map
 $X \to \mathbf 1$ is a Reedy fibration, where $\mathbf 1$ is of
 course the diagram that is constantly the unit type.
\end{defn}

The following lemma will be useful for choosing an element with the
maximal rank from an inverse category with non-empty finite type of
objects.
\begin{lemma}\label{lem:tltt:max-rank-elem}
  For some inverse category $\C$ with non-empty finite type of objects,
  i.e. for any $n : \strict \N$, we have
  $\psi : \obj \C \stiso \strict \Fin_{n+1}$,
  and the rank functor $\varphi : \C -> \strictNop$, we can chose an element $z : \obj \C$
  with the maximal rank, i.e.  for any
  $c : \obj \C$, $\varphi_c \leq \varphi_z$.
\end{lemma}
\begin{proof}
  We want to construct an inhabitant of the following type:
  \begin{equation}\label{eq:tltt:max-rank}
    \sm{z : \obj \C} \prd{c : \obj \C} \varphi_c \leq \varphi_z
  \end{equation}
  We use
  \begin{align*}
    f &: \obj \C -> \strict \Fin_{n+1}\\
    g &: \strict \Fin_{n+1} -> \obj \C\\
    l &: \forall x, g (f~x) \steq x\\
    r &: \forall y, f (g~y) \steq y
  \end{align*}
  for the components of the isomorphism $\psi$.

 \noindent We proceed by induction on $n$.
  \begin{case}[$\obj \C \stiso \strict \Fin_1$]
    We take $z \jdeq g~\strict0$. Thus, we have to show that for any
    $c : \obj \C$ we have $\varphi_c \leq \varphi_{(g~\strict0)}$.
    \begin{align*}
      c  & \steq \text{\{by left inverse $l$\}}\\
         & \steq g (f ~ c)\\
         & \steq \text{\{$(f ~ c) \steq \strict0$, since $\strict0$ is the only inhabitant of $\strict \Fin_1$\}}\\
         & \steq g~\strict0
    \end{align*}
    So, we get $\varphi_{(g~\strict0)} \leq \varphi_{(g~\strict0)}$ as required.
  \end{case}
  \begin{case}[$\obj \C \stiso \strict \Fin_{(n'+1)+1}$]
    Let us pick some element $z' : \obj \C$ (this is possible, since $\obj \C$ is finite,
    for example, we could take $z' \jdeq g~\strict0$, but this proof does not depend
    on our choice of an element of $\strict \Fin_{(n'+1)+1}$, to which we apply the function $g$).
    We call $C'$ the category $\C$ with the element $z'$ removed.
    Also, we have $\psi' : \obj \C' \stiso \strict \Fin_{n'+1}$. We call $\varphi'$ a
    function $\varphi$ restricted to $\obj{\C'}$.
    By induction hypothesis with $\psi'$ and $\varphi'$ we have $z'' : \obj{\C'}$ s.t.
    \begin{equation}\label{eq:tltt:IH}
      \prd{c : \obj{C'}}\varphi'_c \leq \varphi'_{z''}
    \end{equation}
    We observe that $\varphi'_c \jdeq \varphi_c$ and $\varphi'_{z''} \jdeq
    \varphi_{z''}$, since both $c$ and $z''$ are in $\obj \C'$. We do
    not know how $\varphi_{z'}$ and $\varphi_{z''}$ are related, but
    since the order on $\N$ is decidable, we proceed by case analysis on
    $\varphi_{z'} \leq \varphi_{z''}$.

    \emph{Subcase 1} $\varphi_{z'} \leq \varphi_{z''}$. We take $z \jdeq z''$
    in \Ref{eq:tltt:max-rank}
    We have to show that for any $c : \obj \C$, $\varphi_c \leq \varphi_z''$. By case analysis
    on decidable equality we again get two cases:
    \begin{itemize}
      \item $c \steq z''$. The claim follows from the assumption $\varphi_z' \leq \varphi_z''$.
      \item $c \strict{\neq} z''$. Since $\C'$ is a category without $z''$, we know that $c : \obj{\C'}$.
        We complete the proof by \Ref{eq:tltt:IH}.
    \end{itemize}
    \emph{Subcase 2} $\varphi_z' > \varphi_z''$. We take $z \jdeq z'$ in \Ref{eq:tltt:max-rank}.
    We have to show that for any $c : \obj \C$, $\varphi_c \leq \varphi_z'$. We again proceed
    by case analysis on decidable equality.
    \begin{itemize}
      \item $c \steq z'$. Follows immediately, since $\varphi_z' \leq \varphi_z'$.
      \item $c \strict{\neq} z'$. Again, we now that $c : \obj{\C'}$, and
        complete the proof by \Ref{eq:tltt:IH}.~
    \end{itemize}
  \end{case}
\end{proof}

\begin{remark}
Notice that in Lemma \ref{lem:tltt:max-rank-elem} we could not just pick the
maximal element in $\strict \Fin_{(n+1)}$ and get an element in $\C$ with
maximal rank, since we want the lemma to be independent of particular
isomorphism $\psi$.
\end{remark}

Using the definition of Reedy fibrations (Definition
\ref{def:reedy-fibrations}), we can make precise the claim that we can
construct fibrant limits of certain well-behaved diagrams.  The following
theorem is (probably) the most involved result of our formalisation:

\begin{thm}[see {\cite[Theorem 15]{ann-cap-kra:two-level}} and {\cite[Lemma 11.8]{shulman:inverse-diagrams}}] \label{thm:fibrant-limits}
Assume that $\C$ is an inverse category with a finite type of objects $\obj \C$.
Assume further that $X : \C \to \strict \UU$ is a Reedy fibrant diagram which is
pointwise essentially fibrant (which means we may assume that it is given as a
diagram $\C \to \UU$).

Then, $X$ has a fibrant limit.
\end{thm}
\begin{proof}
By induction on the cardinality of $\obj\C$. In the case $\obj \C \stiso \strict \Fin_0$, the limit is the unit type.

 Otherwise, we have $\obj \C \stiso \strict \Fin_{n+1}$.
 Let us consider the rank functor
 \[ \varphi : \C \to \strictNop. \]
 We choose an object $z : \C$ such that $\varphi_z$ is maximal using Lemma
 \ref{lem:tltt:max-rank-elem}
 Let us call $\C'$ the category that we get if we remove $z$ from $\C$;
 that is, we set $\obj {\C'} \defeq \sm{x : \obj \C} \neg (x \steq z)$.
 Clearly, $\C'$ is still inverse, and we have $\obj{\C'} \stiso \strict \Fin_n$.

 Let $X : \C \to \UU$ be Reedy fibrant.
 We can write down the limit of $X$
 explicitly as
 \begin{equation}\label{eq:tltt:lim-explicit}
  \sm{c : \prd{y : \obj \C} X_y} \prd{y,y' : \obj \C} \prd{f : \C(y,y')} Xf(c_y) \steq c_{y'}.
 \end{equation}
 Writing the pretype $\obj \C$ as a coproduct $1 \strict + \obj\C'$, we get that the above pretype is strictly isomorphic to
 \begin{equation}\label{eq:tltt-lim-split}
  \begin{alignedat}{1}
   & \sm{c_z : X_z}\sm{c : \prd{y : \obj\C'} X_y} \\
   & \phantom{\Sigma} \left(\prd{f : \C(z, z)} Xf(c_z) \steq c_z \right) \times \\
   & \phantom{\Sigma} \left(\prd{y : \obj\C'} \prd{f : \C(y, z)} Xf(c_y) \steq c_z \right) \times \\
   & \phantom{\Sigma} \left(\prd{y : \obj\C'} \prd{f : \C(z, y)} Xf(c_z) \steq c_y \right) \times \\
   & \phantom{\Sigma} \left(\prd{y,y' : \obj\C'} \prd{f : \C(y, y')} Xf(c_y) \steq c_{y'} \right).
  \end{alignedat}
 \end{equation}
 Using that $z$ has no incoming non-identity arrows (together with \textsc{uip}), two of the components of the above type are contractible and can be removed, leaving us with
 \begin{equation}\label{eq:limit-2}
 \begin{alignedat}{1}
  & \sm{c_z : X_z}\sm{c : \prd{y : \obj\C'} X_y} \\
  & \phantom{\Sigma} \left(\prd{y : \obj\C'} \prd{f : \C(z, y)} Xf(c_z) \steq c_y \right) \times \\
  & \phantom{\Sigma} \left(\prd{y,y' : \obj\C'} \prd{f : \C(y, y')} Xf(c_y) \steq c_{y'} \right).
 \end{alignedat}
 \end{equation}

 Let us write $L$ for the limit of $X$ restricted to $\C'$,
 and let us further write $p$ for the canonical map $p : L \to M_z^X$.
 Further, we write $q$ for the map $X_z \to M_z^X$.
 Then, \eqref{eq:limit-2} is strictly isomorphic to
 \begin{equation}\label{eq:limit-3}
  \sm{c_z : X_z} \sm{d : L} p(d) \steq q(c_z).
 \end{equation}
 Swapping sigmas in \reff{eq:limit-3} gives us
 \begin{equation}\label{eq:limit-4}
  \sm{d : L} \sm{c_z : X_z} p(d) \steq q(c_z).
 \end{equation}
 This is the pullback of the span $L \xrightarrow{p} M_z^X \xleftarrow{q} X_z$:
 \begin{equation}\label{eq:tltt:pullback-L-Xz}
    \xymatrix{
      L \times_{M^X_z} X_z \ar[d]^f \ar[r]^g & X_z \ar[d]^q \\
      L \ar[r]^{p} & M^X_z
    }
 \end{equation}
 By Reedy fibrancy of $X$, the map $q$ is a fibration.
 Thus, by Lemma~\ref{lem:fibrant-pullback}, the map
 $f : \sm{c_z : X_z} \sm{d : L} p(d) \steq q(c_z) -> L$
 on \eqref{eq:tltt:pullback-L-Xz} is a fibration.
 By the induction hypothesis, $L$ is essentially fibrant.
 This implies that \eqref{eq:limit-3} is essentially fibrant, as it is the domain of a fibration whose codomain is essentially fibrant.
\end{proof}

\section{Formalisation in Lean}\label{tltt:sec:lean-formalisation}
With a proof assistant that implements our two-level theory, one would thus be
able to mechanise the results of the paper rather directly, or at least similarly
directly as papers with purely internal results can be implemented in current
proof assistants: of course, there is always still some work to do because some
low-level steps are omitted in informal presentations.
As we do not have such a proof assistant at hand, the task is to implement two
level type theory in conventional proof assistants reusing as many features as possible.

An overall idea of an approach to implementation that is suitable for most
existing proof assistants is the following. We work in a type theory with
universes of ``strict'' types (i.e.\ where \textsc{uip} holds).
Pretypes correspond to the ordinary types of the proof assistant and (fibrant) types are represented as pretypes ``tagged'' with the extra property of being fibrant.
The role of our strict equality is played by the ordinary propositional equality of the proof assistant (which, thanks to \textsc{uip}, is indeed propositional in the sense of HoTT).
The fibrant equality type is postulated together with its elimination rule $J$.
We further postulate fibrancy preservation rules for $\Pi$ and $\Sigma$ as they
are given in Section~\ref{sec:ttlt-spec}. The usual computation rule for $J$ is
defined using strict equality.

The proof assistant Lean~\cite{moura:lean}, which we have chosen for our formalisation,%
\footnote{The code is available at \url{https://github.com/annenkov/two-level}.}
can operate in two different ``modes'': one with a built-in \textsc{uip}, and one which is suitable for $\HOTT$.
Our Lean implementation is based on ``strict'', Lean mode. That means that Lean's
\icode{Type} now becomes a pretype in two-level type theory sense.

\begin{remark}[Notation]
  Before we start describing our Lean development
  we introduce some notation used in the Lean code snippets:
  \begin{itemize}
  \item \icode{A → B} for the type of functions from \icode{A} to \icode{B}
  \item \icode{x ⟶ y} for morphisms (\icode{hom} field in the \icode{category}
    structure)
  \item \icode{C ⇒ D} for functors from a category \icode{C} to a category \icode{D}
  \item \icode{Nat(F,G)} for natural transformations from a functor \icode{F}
    to a functor \icode{G}
  \item \icode{p ▹ a} for transport of \icode{a} along the equality \icode{p}
  \end{itemize}
\end{remark}

 Fibrant types $\texttt{Fib}$ are implemented using Lean dependent
 records with two fields: a pretype, and the property that it is
 fibrant:
\begin{lstlisting}
structure Fib : Type := mk ::
  (pretype : Type)
  (fib : is_fibrant pretype)
\end{lstlisting}

The \texttt{is\_fibrant} property is defined using the type class mechanism
provided by the language.
\begin{lstlisting}
constant is_fibrant_internal : Type → Prop

structure is_fibrant [class] (X : Type) := mk ::
  fib_internal : is_fibrant_internal X
\end{lstlisting}

We declare the \icode{Fib.pretype} field to be a coercion, allowing to project
a pretype out of \icode{Fib}. For that purpose, we use the mechanism of \emph{attributes}.
\begin{lstlisting}
attribute Fib.pretype [coercion]
\end{lstlisting}
Such a declaration implements the \nameref{rule:fib-pre} rule, which says that every
fibrant type is also a pretype.

For the second component of the \icode{Fib} structure, we define the following attribute:
\begin{lstlisting}
attribute Fib.fib [instance]
\end{lstlisting}
This definition makes available for every inhabitant of \icode{Fib} an instance
of the \icode{is_fibrant} type class.

The rules that $\Sigma$- and $\Pi$-types preserve fibrancy
(rules \nameref{rule:sigma-fib} and \nameref{rule:pi-fib}) are also postulated and
exposed as instances of the \icode{is_fibrant} type class:
\begin{lstlisting}
constant sigma_is_fibrant_internal {X : Type}{Y : X → Type}
  : is_fibrant X
  → (Π (x : X), is_fibrant (Y x))
  → is_fibrant_internal (Σ (x : X), Y x)

definition sigma_is_fibrant [instance] {X : Type}{Y : X → Type}
  [fibX : is_fibrant X] [fibY : Π (x : X), is_fibrant (Y x)] :
  is_fibrant (Σ (x : X), Y x) :=
    is_fibrant.mk (sigma_is_fibrant_internal fibX fibY)

\end{lstlisting}
The rule for $\Pi$-types is implemented in a similar way. In the same
way we postulate that the unit type, equality types and \icode{Fib} itself are
fibrant.

The presentation of fibrant types using type classes results in a very elegant
implementation of the fibrant fragment of the type theory.
The class instance resolution mechanism
allows us to leave the property of being fibrant implicit in most cases.
We use $\texttt{Fib}$ in definitions and let Lean insert coercions automatically
in places where a pretype is expected, or where a witness that a type is fibrant
is required. We will consider several examples showing how Lean's class resolution
mechanism helps writing definitions involving fibrant types.

First, we declare some variables, which will be used in our examples:
\begin{lstlisting}
variables {A : Fib} {B : Fib} {C : Fib} (P : A → Fib)
\end{lstlisting}
Now we can use these declarations in any definition in the same
namespace and Lean will automatically add them as arguments to
definitions that use them.\footnote{The details about namespaces,
  variables and other Lean features can be found in the Lean Tutorial
  \url{https://leanprover.github.io/tutorial/tutorial.pdf}}
Another point to note here is that because of the \icode{[instance]}
attribute for the \icode{Fib.fib} field, all instances of the
\icode{is_fibrant} type class for declared variables of type \icode{Fib}
are available for Lean's resolution mechanism.

Our fist example is an equivalence lemma known as associativity of the product type:
\begin{lstlisting}
definition prod_assoc : A × (B × C) ≃ (A × B) × C := sorry
\end{lstlisting}
Since in this example we care only about the statement itself, and not about
the proof, we will use the \icode{sorry} keyword, which allows us to assume a
definition. Because we state a fibrant equivalence between fibrant types in
\icode{prod_assoc}, both sides of the equivalence must be some fibrant
types. All we know from the variable declarations above is that types
\icode{A}, \icode{B}, and \icode{C} are fibrant. To show that the product of these
types is fibrant as well we would have to apply a special case of the
\nameref{rule:sigma-fib} rule (which we call \icode{prod_is_fibrant}) two times
on each side.  Thanks to the class instance resolution mechanism we can leave this
task to Lean.  We change some pretty-printing options to be able to see
implicit arguments for the \icode{prod_assoc} definition.
\begin{lstlisting}
set_option pp.implicit true
set_option pp.notation false
\end{lstlisting}
Running the \icode{check @prod_assoc} command gives us the following result:
\begin{lstlisting}
prod_assoc :
Π {A} {B} {C},
  @fib_equiv (prod A (prod B C)) (prod (prod A B) C)
    -- inferred by Lean --
    (@prod_is_fibrant A (prod B C) (Fib.fib A) (@prod_is_fibrant B C (Fib.fib B) (Fib.fib C)))
    (@prod_is_fibrant (prod A B) C (@prod_is_fibrant A B (Fib.fib A) (Fib.fib B)) (Fib.fib C))
    ----------------------
\end{lstlisting}
This example shows nested applications of \icode{prod_is_fibrant}, which were
resolved automatically by Lean.  The same resolution procedure allows for inferring
implicit fibrancy conditions using fibrancy-preservation rules for different
type formers. In the following example, expressing the universal property of
the product type, rules \nameref{rule:pi-fib} and a special case of
\nameref{rule:sigma-fib} are used.
\begin{lstlisting}
definition prod_universal : (C → A × B) ≃ (C → A) × (C → B) := sorry
\end{lstlisting}
We can inspect the inferred implicit arguments again by running the \icode{check
  @prod_assoc} command:
\begin{lstlisting}
prod_universal :
Π {A} {B} {C},
@fib_equiv (C → prod A B) (prod (C → A) (C → B))
  -- inferred by Lean --
  (@pi_is_fibrant C (λ a, prod A B) (Fib.fib C)
    (λ x, @prod_is_fibrant A B (Fib.fib A) (Fib.fib B)))
  (@prod_is_fibrant (C → A) (C → B)
    (@pi_is_fibrant C (λ a, A) (Fib.fib C) (λ x, Fib.fib A))
      (@pi_is_fibrant C (λ a, B) (Fib.fib C) (λ x, Fib.fib B)))
  ----------------------
\end{lstlisting}

The following example  highlights one of the reasons behind our choice of Lean
for implementing two-level type theory.
\begin{lstlisting}
variables  (Q : A → Type) [Π a, is_fibrant (Q a)]

definition pi_eq (f : Π (a :A), Q a) : f $\sim$ f := refl _
\end{lstlisting}
An experimental implementation in Agda uses a definition like the above, and
the example failed to work in Agda.  As it became clear later, the example
failed to work because of a small difference in the inference of implicit
arguments, That is, changing \icode{$\Pi$ a, is_fibrant (Q a)} to \icode{$\Pi$
  $\{a\}$, is_fibrant (Q a)} (in the corresponding Agda code), would make the
example be accepted by Agda (see also Section \ref{tltt:subsec:agda}).  In our
Lean development \icode{pi_eq} successfully type-checks, and Lean infers the
following:
\begin{lstlisting}
pi_eq : Π {A} Q [_inst_1] f,
  @fib_eq (Π a, Q a)
    -- inferred by Lean --
    (@pi_is_fibrant A Q (Fib.fib A) _inst_1)
    ----------------------
    f f
\end{lstlisting}
Where \icode{_inst_1} is an automatically generated name for the instance of a type class
corresponding to \icode{[Π a, is_fibrant (Q a)]} appearing in the variables declaration.

\begin{remark}\label{rem:tltt-essentially-fibrant}
It is worth pointing out, however, that in our formalisation we do not make a
distinction between fibrant and essentially fibrant pretypes, having instead a
single predicate \texttt{is\_fibrant} to express this property. That is, every
type which is strictly isomorphic to a fibrant type is also considered fibrant
by the axiom we postulate in our implementation. This makes the development more
convenient as long as we use essentially fibrant types for most of the results
presented in the current formalisation. For instance, Theorem~\ref{thm:fibrant-limits}
and a number of auxiliary lemmas for this theorem involve essentially fibrant types.
\end{remark}

\section{Working in the Fibrant Fragment}
To show how to work with the fibrant type theory, we have formalised some
simple facts from the HoTT library. Our implementation shows that many proofs
can be reused almost without change, provided that the same notation is used
for basic definitions. For instance, we have ported some theorems about product
types with only minor modifications. In particular, induction on fibrant
equalities works as expected: we annotate the postulated elimination principle
with the \texttt{[recursor]} attribute, and the \texttt{induction} tactic
applies this induction principle automatically.

A point to note is that the computation (or $\beta$-) rule for the $J$
eliminator of the fibrant equality type is implemented as a strict equality,
using the propositional equality of the proof assistant.  This means that the
rule does not hold judgmentally.  Consequently, this computation does a priori
not happen automatically, and explicit rewrites along the propositional
$\beta$-rules are needed in proof implementations.  This and other issues of
the same kind are addressed by using one of Lean's proof automation
features. We annotate all the ``computational'' rules with the attribute
\texttt{[simp]}. This attribute is used to guide Lean's simplification tactic
\texttt{simp} which performs reductions in the goal according to the base of
available simplification rules. That allows us to use a simple proof pattern:
do induction on relevant equalities and then apply the \texttt{simp} tactic.
However, the \texttt{simp} tactic is not a well-documented feature of
Lean. Sometimes it fails to simplify goals, and in such cases we apply repeated
rewrites using propositional computation rules.

Let us consider the following example from the HoTT Lean library. Notice that
we use $=$ for equality here and that it corresponds to the usual HoTT
equality, since we are considering an example from the HoTT Lean library, which
supports HoTT natively.
\begin{lstlisting}
definition prod_transport (p : a = a') (u : P a $\times$ Q a) :
  p $\triangleright$ u = (p $\triangleright$ u.1, p $\triangleright$ u.2) :=
  by induction p; induction u; reflexivity
\end{lstlisting}
After applying induction on \icode{p} and \icode{u}, we get the following goal:
\begin{lstlisting}
  refl$_a$ $\triangleright$ (a$_1$, a$_2$) = (refl$_a$ $\triangleright$ (a$_1$, a$_2$).1, refl$_a$ $\triangleright$ (a$_1$, a$_2$).2)
\end{lstlisting}
Since we are in the HoTT mode of Lean, computation rule for transport holds
judgmentally, and we can simplify the expression above to
\begin{lstlisting}
  (a$_1$, a$_2$) = ((a$_1$, a$_2$).1, (a$_1$, a$_2$).2)
\end{lstlisting}
This goal we can prove by refl$_{(a_1, a_2)}$, since
(($\mathtt{a}_1$, $\mathtt{a}_2$).1, ($\mathtt{a}_1$, $\mathtt{a}_2$).2) reduces
to ($\mathtt{a}_1$, $\mathtt{a}_2$). The \icode{reflexivity} tactic performs these
reduction steps and proves the goal.
We use ``$\sim$'' in place of  ``='' for the fibrant equality in our two-level type theory.

In the fibrant fragment of type theory we can express this proof in the following form:
\begin{lstlisting}
definition prod_transport (p : a $\sim$ a') (u : P a  $\times$ Q a) :
  p $\triangleright$ u $\sim$ (p $\triangleright$ u.1, p $\triangleright$ u.2) :=
  by induction p; induction u; repeat rewrite transport$_\beta$
\end{lstlisting}
In this case, after applying induction on \icode{p} and \icode{u}, we get the same goal as
in the previous example (up to notation for the equality):
\begin{lstlisting}
  refl$_a$ $\triangleright$ (a$_1$, a$_2$) $\sim$ (refl$_a$ $\triangleright$ (a$_1$, a$_2$).1, refl$_a$ $\triangleright$ (a$_1$, a$_2$).2)
\end{lstlisting}
The difference is that we cannot reduce this goal and apply \icode{reflexivity},
since in the fibrant fragment computation rule for transport is defined using propositional
equality. That is, simplification of the goal gives us only that:
\begin{lstlisting}
  refl$_a$ $\triangleright$ (a$_1$, a$_2$) $\sim$ (refl$_a$ $\triangleright$ a$_1$, refl$_a$ $\triangleright$ a$_2$)
\end{lstlisting}
We can finish the proof with the help of Lean's proof automation by repeatedly
applying rewriting with \texttt{transport}$_\beta$.  An alternative would be to
use the \icode{simp} tactic to simplify the goal; in this case the proof will look
very close to the original one:
\begin{lstlisting}
  by induction p; induction u; simp
\end{lstlisting}

There is another issue which arises, particularly, when defining propositional $\beta$-rules for equality-dependent definitions.
As an example let us consider an action on paths, which depends of some path $p$.
\begin{lstlisting}
  apd {X : Type} {P : X -> Fib} {x y : X}
      (f : Π x, P x) (p : x $\sim$ y) : p $\triangleright$ f x $\sim$ f y
\end{lstlisting}
When defining the computation rule for \texttt{apd} we would like to write the following:
\begin{lstlisting}
  apd f refl$_x$ $\steq$ refl$_{(f x)}$,
\end{lstlisting}
Unfortunately, this term is not well-typed, since the left-hand side of the equation
has type $\mathtt{refl}_x \triangleright (f x) = f x$, while the right-hand side
has type $f x = f x$, where $\triangleright$ stand for transport along the fibrant
equality. In order to make this definition well-typed we have to apply explicitly
the propositional computation rule for transport. This leads to the following equation:
\begin{lstlisting}
  apd$_\beta$ {P : X → Fib} (f : Π x, P x) {x y : X} :
      (transport$_\beta$ (f x)) $\triangleright_s$ (apd f refl$_x$) $\steq$ refl$_{(f x)}$
\end{lstlisting}
where $\triangleright_s$ denotes transport along the strict equality, i.e.
Lean's propositional equality (we could have transported the right-hand side instead).
Writing definitions like that is inconvenient, but there is a way to avoid it.
We can define propositional $\beta$-rules only for some basic cases (like transport)
and unfold definitions in proofs to a form for which these basic rules can be applied.
We tested this strategy while porting some theorems about $\Sigma$-types from
Lean's HoTT library. In general, this issue could appear in more complex cases
than those we have investigated; it is similar to the problem appearing in axiomatic
definitions of higher inductive types in Coq~\cite{barras:hitsincoq}, where a
proposed solution has been to use private inductive types (see section
\ref{tltt:subsec:Boulier-Tabareau}).

\section{Internalising the Inverse Diagrams}\label{tltt:sec:lean-inv-diag}
This section describes the details of the implementation of results from
Section \ref{tltt:sec:applications}. As we are working in strict Lean, we have
decided to use the existing formalisation of category theory from the standard
library.\footnote{The standard library is part of the Lean's distribution.  The
  source code in available at \url{https://github.com/leanprover/lean2}.}
Unfortunately, it is not as developed as the formalisation in HoTT Lean.  For
that reason, additional effort was needed to formalise some concepts from
category theory required for the results given in the paper.

The following definitions from the Lean standard library are used in our
formalisation:
\begin{itemize}
\item categories;
\item functors;
\item natural transformations.
\end{itemize}
We had to implement the following notions:
\begin{itemize}
\item pullbacks and general limits;
\item construction of the limit for Pretype category;
\item coslice and reduced coslice;
\item matching object;
\item inverse categories
\end{itemize}
In addition, we have proved some properties of the strict isomorphism and
finite sets.

We will outline some definitions from the Lean's standard library
and then discuss in details what we have implemented ourselves.

A definition of a category from Lean standard library corresponds to the notion
of strict category (Definition \ref{def:strictcat}). The usual approach to
encode algebraic structures as dependent records (or structures) is used in the
Lean standard library.  Dependent records, being a generalisation of
$\Sigma$-types, allow for fields in the definition to depend on previously
defined fields. In this way, one can combine operations and laws that these
operations must satisfy. Usually, laws are expressed in the form of
propositions (type \icode{Prop} in Lean).

In particular, a category is defined as follows:
\begin{lstlisting}
structure category [class] (ob : Type) : Type :=
  (hom : ob → ob → Type)
  (comp : Π⦃a b c : ob⦄, hom b c → hom a b → hom a c)
  (ID : Π (a : ob), hom a a)
  (assoc : Π ⦃a b c d : ob⦄ (h : hom c d)
           (g : hom b c) (f : hom a b),
     comp h (comp g f) = comp (comp h g) f)
  (id_left : Π ⦃a b : ob⦄ (f : hom a b), comp !ID f = f)
  (id_right : Π ⦃a b : ob⦄ (f : hom a b), comp f !ID = f)

\end{lstlisting}
The definition specifies a category for the pretype of objects \icode{ob}
(note that \icode{Type}in our Lean formalisation corresponds to pretypes in two-level type theory).
Other fields in the definition follow Definition \ref{def:strictcat} quite closely.

The definition of functors follows the same idea of using dependent records.
Functors are structures with four fields, defining how the functor acts on
objects and morphisms, along with two usual functor laws. Natural
transformations could have been represented as structures as well, but library
implementors have chosen an equivalent representation as an inductive data type
with one constructor taking two arguments: the components and the property,
expressing naturality square. Since laws in the definitions of our structures
have type \icode{Prop}, and in the strict Lean mode we have proof irrelevance,
to prove that two inhabitants of the structure are (strictly) equal it is
sufficient to show that only components, for which these laws are defined are
equal.  Moreover, Lean treats any two propositions of the same type as
definitionally equal, allowing for more proofs to be completed by computation.
For example, for any two natural transformations \icode{N M : Nat(F,G)}, for
some functors \icode{F} and \icode{G}, we have the following:
\begin{lstlisting}
  natural_map N = natural_map M → N = M
\end{lstlisting}
Where \icode{natural\_map : Nat(F,G) → Π (a : C), hom (F a) (G a)} projects
components from a given natural transformation.

We defined a reduced coslice directly following Definition \ref{def:red-coslice}
as a coslice with an additional property.
\begin{lstlisting}
structure coslice_obs {ob : Type} (C : category ob) (a : ob) :=
(to  : ob)
(hom_to : hom a to)

open coslice_obs

structure red_coslice_obs {A : Type} (C : category A) (c : A) extends coslice_obs C c :=
(RC_non_id_hom : Π (p : c = to), $\neg$ (p ▹$_s$ hom_to = category.id))
\end{lstlisting}
We use an inheritance mechanism here to add an additional property to the regular coslice
definition. The \icode{rc_non_id_hom} property states that a morphism \icode{hom_to}
cannot be the identity morphism. We use transport along the strict equality ($\triangleright_s$)
to make the definition well-typed.

The reduced coslice forms a category, which is a full subcategory of the
coslice category, although, we do not use this fact in the implementation.
Instead, we define the reduced coslice category directly.\footnote{this code is
  a slightly modified definition of the coslice category from Lean's standard
  library, which was commented out for some reason unknown to the author.} We
will show the full definition of the reduced coslice category to demonstrate
some Lean's features.
\begin{lstlisting}
definition reduced_coslice {ob : Type}
                           (C : category ob) (c : ob)
  : category (red_coslice_obs C c) :=
⦃ category,
  hom := λa b, Σ(u : hom (to a) (to b)),
           u ∘ hom_to a = hom_to b,
  comp := λ a b c g f,
  ⟨ (pr1 g ∘ pr1 f),
    (show (pr1 g ∘ pr1 f) ∘ hom_to a = hom_to c,
    proof
      calc
        (pr1 g ∘ pr1 f) ∘ hom_to a = pr1 g ∘ (pr1 f ∘ hom_to a): eq.symm !assoc
         ... = pr1 g ∘ hom_to b : {pr2 f}
         ... = hom_to c : {pr2 g}
      qed) ⟩,
  ID := (λ a, ⟨ id, !id_left ⟩),
  assoc := (λ a b c d h g f, sigma.eq !assoc !proof_irrel),
  id_left := (λ a b f, sigma.eq !id_left !proof_irrel),
  id_right := (λ a b f, sigma.eq !id_right !proof_irrel) ⦄

\end{lstlisting}
The morphisms in the this category are commutative triangles with a ``tip'' $c$
(Figure  \ref{fig:tltt:coslice-hom}), which we define as
\begin{lstlisting}
  Σ(u : hom (to a) (to b)), u ∘ hom_to a = hom_to b
\end{lstlisting}
The first projection is a morphism $u$ between the codomains of morphisms going from
$c$ (a morphism $u : a --> b$ in Figure  \ref{fig:tltt:coslice-hom}).
The second projection is an equation corresponding to the commutative triangle.

This definition uses one convenient feature of Lean, namely the \icode{calc} environment.
The \icode{calc} environment allows one to combine a sequence of equations combined using
transitivity of equality. Moreover, the same reasoning is possible with any transitive
relation, and we are going to use it later for reasoning with isomorphism.

To construct a composition of morphisms in the category of (reduced) coslices we have to show
that two commutative triangles with one common side form a bigger commutative triangle.
The angle brackets notation $\langle a,~b \rangle$ is used to construct a $\Sigma$-type.
In the definition of \icode{comp} the first component is just a composition of two
morphisms (\icode{pr1 f} and \icode{pr1 g} correspond to morphisms
$u_1 : a --> b$ and $u_2 : b --> c$ in Figure \ref{fig:tltt:coslice-comp}),
and the second component must be a proof that we get a commutative triangle
with a composition \icode{pr1 g} $\circ$ \icode{pr1 f} as a bottom side.
This proof is carried out in three steps using reasoning in the \icode{calc}
environment. First, we use associativity of the function composition (\icode{assoc})
to rearrange the composition. After that we rewrite using the commutative triangle
for $f$ (\icode{{pr2 f}}), and then complete the proof by rewriting with the
commutative triangle for $g$ (\icode{{pr2 g}}).

The proofs of the categorical laws \icode{assoc}, \icode{id_left}, \icode{id_riht}
for \icode{reduced_coslice} boils down to respective properties of morphisms in \icode{C}.
To ``lift'' these properties to morphisms in \icode{c//C} the property of path in
$\Sigma$-type \icode{sigma.eq} is used.

Before we define inverse categories, we have to define a category $\strictNop$:
\begin{lstlisting}
definition nat_cat_op [instance] : category ℕ :=
  ⦃ category,
    hom := λ a b, a ≥ b,
    comp := λ a b c, @nat.le_trans c b a,
    ID := nat.le_refl,
    assoc := λ a b c d h g f, eq.refl _,
    id_left := λ a b f, eq.refl _,
    id_right := λ a b f, eq.refl _ ⦄

definition ℕop : Category := Mk nat_cat_op
\end{lstlisting}
In our Lean implementation we use $\mathbb{N}$\icode{op} to denote
$\strictNop$.  A morphism between two objects in $\mathbb{N}$\icode{op} is the
$\geq$-relation on natural numbers.  composition is given by transitivity of
$\geq$, identity morphism is reflexivity of $\geq$ (we use properties of
$\geq$-relation from the standard library), and proofs of other properties is
just \icode{eq.refl} meaning that they hold definitionally in Lean. This is a
consequence of what we mentioned before: two inhabitants of the same
proposition are definitionally equal in Lean.  Let us consider a case for
associativity. For any \icode{f : a ≥ b, g : b ≥ c, h : c ≥ d} we have to show
that
\begin{lstlisting}
  nat.le_trans h (nat.le_trans g f) : a ≥ d =
  nat.le_trans (nat.le_trans h g) f : a ≥ d
\end{lstlisting}
Since both sides have the same type \icode{a ≥ d}, and this is a proposition, from
Lean's point of view they represent the same value.

We define the property of a functor that it reflects identities in the following way:
\begin{lstlisting}
definition id_reflect {C D: Category} (φ : C ⇒ D) :=
    Π ⦃x y : C⦄ (f : x ⟶ y), (Σ (q : φ x = φ y), q ▹ φ f = id) → Σ (p : x = y),  p ▹ f = id
\end{lstlisting}
The definition of inverse categories uses a specialised version of
\icode{id_reflect} for the case of $\mathbb{N}$\icode{op}, which does not
require $\varphi$ \icode{f} to be an identity.
\begin{lstlisting}
  definition id_reflect_ℕop {C : Category} (φ : C ⇒ ℕop) :=
    Π ⦃x y : C⦄ (f : x ⟶ y), φ x = φ y → (Σ (p : x = y), p ▹ f = id)
\end{lstlisting}
We use \icode{id_reflect} in our Lean implementation, and show that
these two definitions are logically equivalent. The proof of that uses
the fact that the only morphism \icode{f : x} $\longrightarrow$ \icode{x}
in $\mathbb{N}$\icode{op} is the identity morphism.

Now we can define inverse categories by equipping a category \icode{C} with a rank functor.
\begin{lstlisting}
structure has_idreflect [class] (C D : Category) :=
    (φ : C ⇒ D)
    (reflecting_id : id_reflect φ)

structure invcat [class] (C : Category) :=
    (reflecting_id_ℕop : has_idreflect C ℕop)
\end{lstlisting}

For the definition of the matching object we use the previously constructed limit in the
Pretype category (see remark \ref{tltt:rem:limits-pretype}).
\begin{lstlisting}
 lemma limit_nat_unit {C : Category.{1 1}}
                      (X : C ⇒ Type_category) (z : C)
        : limit_obj (limit_in_pretype X) = Nat(𝟙,X) := rfl

 definition matching_object.{u} {C : Category.{1 1}} [invcat C] (X : C ⇒ Type_category.{u}) (z : C) :=
    Nat(𝟙, (X ∘f (forget C z)))
\end{lstlisting}
Where $\mathbf{1}$ \icode{: C} $=>$ \icode{Type_category} is a functor, which is
constantly the unit type. The functor \icode{forget} is defined exactly as
specified in Section \ref{tltt:sec:applications}: an action on objects is
the projection \icode{to} of the \icode{coslice_obs} structure, and on morphisms the
action is a projection of the ``bottom'' morphism of the commutative triangle
(the morphism $u$ in Figure  \ref{fig:tltt:coslice-hom}), which is a first component
of the $\Sigma$-type defining morphisms in \icode{reduced_coslice}.
\begin{lstlisting}
  definition forget (C : Category) (c : C) : (c // C) ⇒ C :=
    ⦃ functor,
      object := λ a, to a,
      morphism := λ a b f, pr1 f,
      respect_id := λa, eq.refl _,
      respect_comp := λ a b c f g, eq.refl _ ⦄

\end{lstlisting}

\begin{remark}\label{tltt:rem:limits-pretype}[Limits and Pullbacks]
  The \texttt{limit.lean} contains general definition of limits as a terminal
  object in the category of cones. We also define here an explicit
  representation of limits (see \Ref{eq:tltt:lim-explicit}) along with the
  proof that this definition is isomorphic to our general definition. We also
  construct limits in the category of pretypes and show that the limit of the
  diagram $\mathbf{2} -> \strict{\UU}$ is isomorphic to the product pretype.
  The \texttt{pullback.lean} file contains definitions of pullbacks constructed
  in different ways along with proofs that the definitions are isomorphic.
\end{remark}

We chose to specialise Definition \ref{def:reedy-fibrations},
by taking $Y$ to be a constant functor to the unit type, instead of implementing a
general definition of a Reedy fibration. That is, we define the canonical map
$X_z \to M_z^X$  as the following:
\begin{lstlisting}
  definition matching_obj_map {C : Category.{1 1}}
                              [invC : invcat C]
                              (X : C ⇒ Type_category) (z : C)
       : X z → matching_object X z :=
         λ x, natural_transformation.mk (λ a u, X (hom_to a) x)
           begin
           -- proof of naturality is omitted
           end
\end{lstlisting}

Since matching object is the limit in the Pretype category, it is a natural transformation.
We give only the natural map in the definition, omitting the proof of the naturality condition for
brevity.

In our implementation we use the following definition of a fibration:
\begin{lstlisting}
definition is_fibration_alt [reducible] {E B : Type} (p : E → B)
  := Π (b : B), is_fibrant (fibreₛ p b)
\end{lstlisting}
Where \icode{fibreₛ} is a ``strict'' fibre, that is a fibre defined using strict
equality:
\begin{lstlisting}
definition fibreₛ {X Y : Type} (f : X → Y) (y : Y)
  := Σ (x : X), f x = y
\end{lstlisting}

Now we have all the ingredients to write a definition of Reedy fibrant diagrams.
\begin{lstlisting}
definition is_reedy_fibrant [class] (X : C ⇒ Type_category) :=
    Π z, is_fibration_alt (matching_obj_map X z)
\end{lstlisting}

\subsection{Proof of the Fibrant Limit Theorem}
In the current implementation, besides the general two-level framework, we have
implemented the machinery required to define Reedy fibrant diagrams and have fully
formalised a proof of Theorem~\ref{thm:fibrant-limits}.
The formalisation of Theorem~\ref{thm:fibrant-limits} closely follows the structure
given in the paper \cite{ann-cap-kra:two-level}.

Let us consider some steps of the proof in details. The base case was relatively
easy to prove, and we will focus on the inductive step. We use the lemma \icode{max_fun_to_ℕ}
to pick an element with the maximal rank from the category $\C$. The way we
do it closely corresponds to the proof of Lemma \ref{lem:tltt:max-rank-elem}
Surprisingly, proving that after removal of $z$ from $\C$ the resulting $\C'$ is
finite, inverse and that diagram $X' : \C' \to \strict \UU$ is still Reedy fibrant
required writing a lot of boilerplate code (see Remark \ref{rem:tltt-removing-z}).

The overall idea of the proof of the inductive step is to show that our goal is
strictly isomorphic to some fibrant type. Thus, one of the central parts of the
proof was a transformation of the limit of a Reedy fibrant diagram through the
chain of strict isomorphisms. In our Lean formalisation, it is implemented using
the \texttt{calc} environment, which gives a very convenient way of chaining
transitive steps.

We proved each step of the reasoning using isomorphisms in a separate lemma and
then chained them together in the \icode{calc} environment. The most essential
transformations are the \icode{limit_two_piece_limit_equiv} and the\\
\icode{two_piece_limit_pullback_p_q_equiv} lemmas.

The \icode{limit_two_piece_limit_equiv} lemma allows us to represent the limit as
a product of two parts:
\begin{lstlisting}[caption={Limit isomorphism},label={lst:tltt-lean-split-lim}]
(Σ (c : Π y, X y), Π y y' f, morphism X f (c y) = c y')
  ≃ₛ
(Σ (c_z : X z) (c : (Π y : C_without_z z, X y)),
   (Π (y : C_without_z z) (f : z ⟶ obj y ), X f c_z = c y) ×
   (Π (y y' : C_without_z z) (f : @hom (subcat_obj _ _) _ y y'),
      (Functor_from_C' z X) f (c y) = c y'))
\end{lstlisting}
Where \icode{C_without_z z} corresponds to the category $\C'$ (a category $\C$ with $z$ removed)
in Theorem \ref{thm:fibrant-limits}.
We explicitly construct a functor from the category $\C'$ using the functor $X : C \to \strict \UU$.
This lemma also implicitly includes the step \Ref{eq:tltt-lim-split}
(see Theorem \ref{thm:fibrant-limits}).
While constructing the isomorphism given in Listing \ref{lst:tltt-lean-split-lim} we perform a case analysis
on the equalities (since they are decidable) $y = z$ and $y' = z$, which gives us four cases
corresponding to the product of four parts in \ref{eq:tltt-lim-split}.
In case where $y=z$ and $y'=z$, we use the property that the only morphism from
$z$ to $z$ in $\C$ is the identity morphism.
The case where $y \neq z$ and $y' = z$ is impossible, since we have a morphism
$f : y --> z$, and $z$ has a maximal rank, the only morphism with $z$ as a codomain
is the identity morphism, but we know that $y \neq z$.

For the next transformation step, we use the \icode{two_piece_limit_pullback_p_q_equiv} lemma,
which allows us to get \Ref{eq:limit-3}:
\begin{lstlisting}
Σ (c_z : X z) d, p d = q c_z
\end{lstlisting}
Where \icode{p} and \icode{q} are defined as the following:
\begin{lstlisting}
  q := matching_obj_map X z
  p := map_L_to_Mz_alt z X
\end{lstlisting}
In the Lean implementation we use a tactic-level \icode{let} construct to declare these
maps. It is important to use \icode{let} here, although we could use the \icode{have} construct
instead. The difference is that \icode{let} allows us to keep definitions transparent for
the simplification.

The map \icode{p} is a map from the limit of the diagram \icode{X} restricted to
the category \icode{C_without_z z} to the matching object.
We use explicit an representation of the limit as a natural transformation:

\begin{lstlisting}
definition lim_restricted (X : C ⇒ Type_category) (z : C)
                          [invC : invcat C]
:= Σ (c : Π y, (Functor_from_C' z X) y),
     Π (y y' : C_without_z z) (f : @hom (subcat_obj C _) _ y y'),
       ((Functor_from_C' z X) f) (c y) = c y'

definition map_L_to_Mz_alt (z : C) (X : C ⇒ Type_category.{u})
                          [invC : invcat C]
                          (L : lim_restricted X z)
 : matching_object X z :=
    match L with
    | ⟨η, NatSq⟩ :=
        by refine natural_transformation.mk
           (λ a u, η (mk _ (reduced_coslice_ne z a)))
           (λ a b f, funext (λ u, NatSq _ _ _))
    end
\end{lstlisting}

To show that \icode{lim_restricted X z} is fibrant we would like to use the induction
hypothesis, but in order to do that we first have to show that \icode{(Functor_from_C' z X)}
is Reedy fibrant and that \icode{C_without_z z} is still a category with a finite object type.
We prove this facts in the separate lemmas \icode{Functor_from_C'_reedy_fibrant} and
\icode{C_without_z_is_obj_finite}, respectively.
\begin{remark}\label{rem:tltt-removing-z}
From the code above, one can see, that removing the element from \icode{C}
requires us to show to propagate this change through all the layers, such as definitions
of functors and limits, properties if category \icode{C_without_z z} etc.
These changes are usually considered ``obvious'' on paper, but in the formal
setting in a proof assistant could require significant efforts.
It was important to write the definitions related to these lemmas in such a way that
they can be simplified as much as possible using Lean's definitional equalities.
Currently, this part of the implementation is the least readable part.
Probably, there is a way to generalise this by developing suitable machinery to
work with subcategories, but we did not explore that possibility.
\end{remark}

By rearranging sigmas using \icode{sigma_swap} lemma, we get the following:
\begin{lstlisting}
  (Σ (c_z : X z) (d : lim_restricted X z), p d = q c_z) ≃ₛ
  (Σ (d : lim_restricted X z) (c_z : X z), q c_z = p d)
\end{lstlisting}
And this is exactly the pullback of the span
\begin{lstlisting}
  lim_restricted X z $-->$ matching_object X z $<--$ X z
\end{lstlisting}
By the lemma \ref{lem:fibrant-pullback} (which we call \icode{Pullback'_is_fibrant}),
and we know that the map
\begin{lstlisting}
  (Σ (d : lim_restricted X z) (c_z : X z), q c_z = p d) -> lim_restricted X z}
\end{lstlisting}
is a fibration.

To complete the proof, we need to show that the domain of a fibration with fibrant codomain
is fibrant, i.e.
\begin{lstlisting}
  Π (p : E → B), is_fibration_alt p → is_fibrant E
\end{lstlisting}
In order to do this we make use of the fact that we can ``contract'' some parts of the type.
This is known as a \emph{singleton contraction} (see Lemma 3.11.8 in \cite{hott-book}):
\begin{lstlisting}
  definition singleton_contrₛ [instance] {A}
    : Π b, (Σ (a : A), b = a) ≃ poly_unit
\end{lstlisting}
Here \icode{poly_unit} is a universe-polymorphic unit type. We mark our definition
with the \icode{[instance]} attribute to make this definition available for Lean's
type class instance resolution mechanism. The full proof of the lemma looks very concise:
\begin{lstlisting}
definition singleton_contr_fiberₛ {E B : Type} {p : E → B}
  : (Σ b, fibreₛ p b) ≃ₛ E :=
    calc
    (Σ b x, p x = b) ≃ₛ (Σ x b, p x = b) : _
                ...  ≃ₛ (Σ (x : E), poly_unit) : _
                ...  ≃ₛ E : _
\end{lstlisting}
This lemma shows another example where type classes are convenient to use in proofs.
All the witnesses for the proof steps are inferred automatically on the base of
available instances of the strict isomorphism. For the second step, Lean's inference
is able to infer that we first need to apply a congruence for $\Sigma$-types and
then use a singleton contraction lemma. The resulting term, which was constructed by
Lean looks as follows (showing all the implicit arguments and replacing
notation with textual representation):
\begin{lstlisting}
  @sigma_congr₂ E (λ x, @sigma B (@eq B (p x)))
                  (λ x, poly_unit)
                  (λ a, @singleton_contrₛ B (p a))
\end{lstlisting}

Now we can show, that for any fibration \icode{p : E → B} and fibrant \icode{B},
the type \icode{E} is fibrant. First, use \icode{singleton_contr_fiberₛ} to show \icode{E} is
strictly isomorphic to \icode{(Σ b, fibreₛ p b)}, and this type is fibrant,
since $\Sigma$-type preserve fibrancy, \icode{B} is fibrant and
\icode{$\Pi$ (b : B), is_fibrant (fibre$_s$ p b)} from the fact that \icode{p} is a fibration.
Again, in Lean it is sufficient to give a hint, which type \icode{E} is equivalent to, and
the rest could be resolved automatically. The resulting proof looks very concise:
\begin{lstlisting}
  definition fibration_domain_is_fibrant {E : Type} {B : Fib} :
    Π (p : E → B), is_fibration_alt p → is_fibrant E
      := λp is_fibr_p, @equiv_is_fibrant (Σ b, fibreₛ p b) _ _ _
\end{lstlisting}
Application of this lemma concludes the proof of the Theorem \ref{thm:fibrant-limits}.

\subsection{Additional facts}
In the proof of Theorem \ref{thm:fibrant-limits} we use some properties of
strict isomorphism and finite sets. These lemmas can be found in the files
\texttt{facts.lean} and \texttt{finite.lean} of our implementation. For the
strict isomorphism (which is called \icode{equiv} in the Lean's standard
library), we implemented congruence lemmas for $\Pi$- and $\Sigma$-types.

\begin{lstlisting}
pi_congr₁ [instance] {F' : A'  → Type} [φ : A ≃ A']
    : (Π (a : A), F' (φ ∙ a)) ≃ (Π (a : A'), F' a)

pi_congr₂ [instance] {F G : A → Type}  [φ : Π a, F a ≃ G a]
    : (Π (a : A), F a) ≃ (Π (a : A), G a)

pi_congr [instance] {F : A → Type} {F' : A' → Type}
                    [φ : A ≃ A'] [φ' : Π a, F a ≃ F' (φ ∙ a)]
    : (Π a, F a) ≃ Π a, F' a

sigma_congr₁ [instance] {A B: Type} {F : B → Type} [φ : A ≃ B]
    : (Σ a : A, F (φ ∙ a)) ≃ Σ b : B, F b

sigma_congr₂ [instance] {A : Type} {F G : A → Type}
                        [φ : Π a : A, F a ≃ G a]
    : (Σ a, F a) ≃ Σ a, G a

sigma_congr {A B : Type} {F : A → Type} {G : B → Type}
            [φ : A ≃ B] [φ' : Π a : A, F a ≃ G (φ ∙ a)]
    : (Σ a, F a) ≃ Σ a, G a
\end{lstlisting}

We used these lemmas in many proofs including the proofs of Theorem \ref{thm:fibrant-limits},
where transformation thorough the sequence of isomorphisms is an important part
of the proof. As mentioned before, properties of the isomorphism are
instances of a corresponding type class. That allows for proving some goals
involving isomorphisms automatically.

Properties of finite sets required for our development are related to the
removal of an element from a given finite set.
\begin{lstlisting}
  fin_remove_max {n : ℕ} :
      (Σ i : fin (nat.succ n), i $\neq$ maxi) ≃ fin n

  fin_remove_equiv {n : ℕ } (z : fin (nat.succ n))
    : (Σ i : fin (nat.succ n), i $\neq$ z) ≃ fin n
\end{lstlisting}
The \icode{fin_remove_max} lemma shows that removing the maximal element gives
us a finite set of smaller cardinality, and \icode{fin_remove_equiv}
generalises this result to the removal of any element. The proof of the last lemma
uses several additional lemmas allowing us to manipulate finite sets using
transpositions:
\begin{lstlisting}
definition fin_transpose {n} (i j k : fin n) : fin n :=
    match fin.has_decidable_eq _ i k with
    | inl _ := j
    | inr _ := match fin.has_decidable_eq _ j k with
               | inl _ := i
               | inr _ := k
               end
    end
\end{lstlisting}

\section{Other Formalisations}
\subsection{The Boulier-Tabareau Coq Development}\label{tltt:subsec:Boulier-Tabareau}

Boulier and Tabareau~\cite{boulier:formalisation} have implemented a theory with two equalities in the Coq proof assistant~\cite{bertot:coq}.  It uses an approach that is somewhat similar to our own development of two-level type theory.
In particular, the authors use Coq type classes to track fibrant types and exploit the corresponding features of the type class resolution mechanism to derive fibrancy automatically. However, there are some differences in the details of how the fibrant equality type is implemented.

In our Lean development we postulate a fibrant equality type and the equality eliminator, while in~\cite{boulier:formalisation} the authors define it as a \emph{private} inductive type~\cite{bertot:private}. This feature of the Coq proof assistant allows one to define an inductive type so that no eliminators are generated and no pattern-matching is allowed outside of the module where this type is defined. Exposing a custom induction principle for such a private inductive type allows one to retain computational behaviour, while restricting the user to explicitly provided eliminators.

Although such an implementation has some advantages, like making more proofs compute, it relies on specific implementation details. In the current version of Coq, private inductive definitions are still an experimental extension. The authors of~\cite{boulier:formalisation} had to use a custom rewrite tactic implemented in OCaml in order to fix an incorrect behaviour of the private definition under the standard Coq rewriting tactic.

Our development in Lean could be seen as more explicit and straightforward approach to the implementation of a two-level type theory, and the simplicity of the encoding of fibrancy constraints makes it potentially more portable to different systems, as long as they are equipped with a powerful enough type class resolution mechanism.

\subsection{Experience with Agda}\label{tltt:subsec:agda}

Our choice of Lean as the language for the formalisation of this paper has been
a consequence of a failed attempt at embedding two-level type theory in the
Agda proof assistant~\cite{norell:towards}.

Analogously to the development that has been eventually realised in Lean, our
plan was to consider Agda's underlying theory, which includes \textsc{uip}, as
the strict fragment of our two-level type theory, and use \emph{instance
  arguments}, which are Agda's implementation of type classes, to express
fibrancy conditions on pretypes.

Unfortunately, due to the way instance and implicit arguments get resolved in
Agda, we were not able to get automatic propagation of fibrancy conditions over
type expressions involving families of types, such as $\Pi$ or $\Sigma$ types
in our initial attempt in Agda.

The self-contained example by Paolo Capriotti shows that certain ways of defining
fibrancy condition fail to resolve implicit arguments.
\begin{lstlisting}
module tltt where


postulate

  is-fibrant : ∀ {i} → Set i → Set i
  instance Π-is-fibrant : ∀ {i}{j}{A : Set i}{B : A → Set j}
                    → ⦃ fibA : is-fibrant A ⦄
                    → ⦃ fibB : (a : A) → is-fibrant (B a) ⦄
                    → is-fibrant ((a : A) → B a)

module _ {i} {A : Set i} ⦃ fibA : is-fibrant A ⦄ where

  postulate

    _~_ : A → A → Set i

module test {i}{j}{A : Set i}{B : A → Set j}

  ⦃ fibA : is-fibrant A ⦄

-- this will work, if we change (a : A) → is-fibrant (B a)
-- to {a : A} → is-fibrant (B a)
  ⦃ fibB : (a : A) → is-fibrant (B a) ⦄ where

  test : (f : (a : A) → B a) → f ~ f

  test = ?

\end{lstlisting}
As it became clear later, a small change in the definition of \icode{fibB}
in the test module from \icode{fibB : (a : A) → is-fibrant (B a)} to
\icode{fibB : \{a : A\} → is-fibrant (B a)} will make Agda's resolution
mechanism work. It was not clear from the documentation why the original definition
fails to work\footnote{Thanks to Nils Danielsson for pointing out what needs to be changed in this example.
  Also, see \url{https://github.com/agda/agda/issues/2755}}. Although there is a way
to make the example above work, it is still not clear if it is possible to develop two-level type
theory in Agda in the same way as we have done in Lean.

We therefore considered alternative approaches, such as postulating a universe
of fibrant types and the corresponding type formers.  Using a certain trick
suggested by Thorsten Altenkirch, one can make sure that the fibrant type
formers agree with the primitive ones.  The trick is similar to the one that
allows higher inductive types with judgemental reduction rules to be
implemented in Agda~\cite{hott-agda}.

However, it appeared that such an approach, although probably feasible, is not
as convenient and immediate as the solution based on type classes that we
eventually settled with in Lean.

\section{Conclusion}
Two-level type theory is a promising approach to internalisation of results
which are currently only partially internal to HoTT, and it is unclear if they
can be fully internalised.  We have demonstrated that two level type theory can
be implemented in an existing proof assistant, outlining the general idea of the
implementation. The approach to the implementation is suitable for most of the
existing proof assistants based on dependent type theory, since we do not rely
on implementation-specific details. Although, for our implementation to be
convenient to use, one will require type classes, proof automation, or some way
to add new judgmental equalities.

Our Lean development should still be considered a proof of concept, as
it does not fully implement all the results presented in the paper
\cite{ann-cap-kra:two-level}.  However, we hope that it serves as a compelling
demonstration of the feasibility of our formalisation approach. To test how one
can work in the fibrant fragment in our Lean development, we have ported some
theorems from the Lean HoTT library. We used proof automation to mimic
computational behavior of the $\beta$-rule for the fibrant equality
eliminator. In most cases that we have considered, modifications of proofs were
not substantial. Although, it is worth pointing out that in some situations it
is inconvenient to write statements where reduction in types involving the
$\beta$-rule for fibrant equality is required.

As a possible extension of results presented in this work, one could consider
to explore the conservativity result from \cite{paolo:thesis}. Having
conservativity, one could take, for example, a definition of $n$-restricted
semi-simplicial types in in two-level type theory. Instantiating the definition
with particular \emph{strict} natural number (i.e. some $n : \strict \N$) 0, 1,
3, etc., and evaluating the term in the strict fragment one could acquire a term,
which belongs to the fibrant fragment of our two-level type theory. Since the
fibrant fragment represents HoTT, it should be possible to use this term in the
proof assistant, where HoTT is supported directly (after converting the term
appropriately). In the context of our Lean development it would require only
minimal efforts, since Lean (version 2) has a mode supporting HoTT
``natively''.

\chapter{Conclusion}

The results of our work show how we have addressed the questions that we stated
in the introduction to the thesis.

First, we have developed a payoff intermediate language along with a
compilation procedure. The compilation procedure allows us to translate contract
specifications in a domain-specific language for financial contracts (CL) to
payoff expressions. The template extension to the original contract language and
the parameterisation of payoff expressions with the ``current time'' parameter
allow for compiling contact templates (or instruments) once and then reuse the
compiled code, instantiating the parameters with different values. Our
experience with Haskell code generation shows that payoff expressions are relatively
easy to map to a subset of a functional language, and we expect that code
generation into the Futhark language \cite{Henriksen:Futhark} could be
implemented in a similar way.  Moreover, performance measurements for payoff expressions
compiled ``by hand'' to OpenCL show that for simple contracts the runtime
overhead is very small in comparison with the recompilation time. Although,
one could potentially find some limitations for more complicated contracts,
properties of our implementation do not restrict one from using the contract
reduction and recompilation approach. All the development has been carried out
in the Coq proof assistant from the beginning and all proofs and definitions
related to the payoff intermediate language (including the soundness proofs in
Chapter \ref{chpt:contracts}) are high-level explanations of the Coq
formalisation.  We can say that we successfully addressed most of the questions
related to this part of the thesis.

Second, we have partially formalised static interpretation of a higher-order
module system for Futhark. Particularly, we have defined the core notions
required for the formalisation, namely semantic objects, along with required
relations.  We have proved one part of the normalisation theorem, specifically, that
static interpretation terminates with a target language expression. Another
part of the theorem related to the typing of target language expressions has
not been formalised (we leave this for future work).

The implementation of semantic objects turned out to be surprisingly tricky due
to the conservative strict positivity check in Coq. We have developed a
technique allowing for another representation of finite maps to be used in the
definition of semantic objects. We have proved that this representation is
isomorphic to the one from the standard library of Coq and used the implicit
coercion mechanism to reuse operations on finite maps from the standard
library. We have built most of the required machinery to deal with bindings in
semantic objects using nominal techniques. For a simplified notion of semantic
objects, we have developed a corresponding nominal set and have defined
$\alpha$-equivalence. In the full setting we have sketched the approach and
have developed examples showing the feasibility of our approach. However these
changes are not yet incorporated into the proof of Theorem \ref{norm.prop}.

Our experience shows that some restrictions on definitions in Coq do not allow
one to use abstractions properly. For example, in the definition of
semantic objects, it is impossible to use an abstract type of finite maps,
otherwise Coq will not be able to check the definition for strict
positivity. When using the \icode{Fixpoint} construct one has to be very
explicit about recursive calls and in most cases it is not possible to call
another function, implementing for example a nested fixpoint. Instead, one has
to inline the definition (see Remark \ref{rem:modules:nested-fix}).

We have developed most of the machinery required for the full formalisation of
static interpretation normalisation, and we believe it is possible to finish
the formalisation given enough time. Having the full result we can then approach
the problem of extracting a certified implementation from our formalisation.

Third, we have implemented a formalisation of two-level type theory in the Lean
proof assistant using only features available for the users and not by modifying
Lean's code.  We showed that in order to make such an encoding usable in an
existing proof assistant, support for type classes, or support for a similar
mechanism is required. Our example application to internalisation of inverse
diagrams shows that our implementation is suitable for this purpose.

We believe that there are no conceptual difficulties in extending our Lean
development with other results from \cite{ann-cap-kra:two-level} because most of
the reasoning happens in the strict fragment. However, such an extension requires a
well-developed standard library to avoid formalising notions not directly related
to the results. We would like to mention that working in the fibrant fragment is
somewhat less convenient in comparison with the proof assistant with
``native'' support for homotopy type theory. This is mainly due to the lack of
computational behavior of the fibrant equality eliminator. We have found ways
to overcome this limitation, but it could still cause inconveniences for more
complicated proofs in the fibrant fragment.

\bibliographystyle{alpha} \bibliography{thesis}
\end{document}